\documentclass[acmsmall,screen]{acmart}
\usepackage{breakcites} 

\startPage{1}

\setcopyright{rightsretained}
\acmPrice{}
\acmDOI{10.1145/3360605}
\acmYear{2019}
\copyrightyear{2019}
\acmJournal{PACMPL}
\acmVolume{3}
\acmNumber{OOPSLA}
\acmArticle{179}
\acmMonth{10}

\renewenvironment{acks}{%
  \makeatletter\if@ACM@anonymous\makeatother
  \else\makeatother 
  \begingroup
  \section*{Acknowledgments}
  \phantomsection\addcontentsline{toc}{section}{Acknowledgments}
}{%
  \endgroup
  \fi
}

\makeatletter
\@ifpackageloaded{subfig}
{}
{\usepackage[labelformat=simple,subrefformat=parens,caption=false]{subfig}
 
 }
\makeatother
\usepackage{graphicx}

\usepackage{color}
\usepackage{listings}
\usepackage{relsize}
\usepackage{multirow}
\usepackage[normalem]{ulem} 

\usepackage{xspace}

\newcommand{\originalgrumbler}[2]{\begin{quote}\textcolor{blue}{\sl{\bf #1 says:} #2}\end{quote}}
\newcommand{\grumbler}[2]{\originalgrumbler{#1}{#2}}

\newcommand{\mike}[1]{\grumbler{Mike}{#1}}

\newcommand{\jake}[1]{\grumbler{Jake}{#1}}

\newcommand{\kaan}[1]{\grumbler{Kaan}{#1}}

\definecolor{darkgreen}{rgb}{0,0.4,0}

\newcommand{\sfsmaller}{}

\newcommand{\bench}[1]{\textsf{\sfsmaller#1}}
\newcommand{\code}[1]{\textsf{\sfsmaller#1}}

\newcommand{\mc}[3]{\multicolumn{#1}{#2}{#3}}

\newcommand{\eg}{e.g.\xspace}
\newcommand{\ie}{i.e.\xspace}

\newcommand{\etal}{et al.\xspace}

\newcommand{\thr}[1]{\textsf{\sfsmaller#1}}

\makeatletter
\@ifundefined{state}{%
}{}
\makeatother

\usepackage{etoolbox}

\usepackage{stackengine} 

\usepackage{amsmath}
\usepackage{amsthm}
\usepackage{amsfonts}
\makeatletter
\def\amsbb{\use@mathgroup \M@U \symAMSb}
\makeatother
\usepackage{bbold}
\usepackage{array}
\usepackage{xspace}
\usepackage{xparse}
\usepackage{ulem}

\usepackage{enumerate}
\usepackage{enumitem} 

\usepackage{multicol}

\usepackage{minibox}

\newtheorem*{theorem*}{Theorem}

\usepackage{color}
\definecolor{darkgreen}{rgb}{0,0.4,0}
\definecolor{darkred}{rgb}{0.4,0,0}

\usepackage{algorithm}
\usepackage[noend]{algpseudocode} 
\algnewcommand{\LineComment}[1]{\hfill \(\triangleright\) #1} 
\makeatletter
\algnewcommand{\LineCommentx}[1]{\Statex \hskip\ALG@thistlm \(\triangleright\) #1}
\algnewcommand{\LineCommentxx}[1]{\Statex \hskip\ALG@tlm \(\triangleright\) #1}
\algnewcommand{\CaseComment}[1]{\hfill #1}
\makeatother
\algnewcommand{\lIf}[2] {\State \algorithmicif\ #1 \algorithmicthen\ #2} 
\algnewcommand{\lIfElse}[3] {\State \algorithmicif\ #1 \algorithmicthen\ #2 \algorithmicelse\ #3} 
\algnewcommand{\lElse}[1] {\State \algorithmicelse\ #1} 
\algnewcommand{\lForAll}[2]{\State \algorithmicforall\ #1 \algorithmicdo\ #2} 
\algnewcommand\algorithmicforeach{\textbf{foreach}} 
\algdef{S}[FOR]{ForEach}[1]{\algorithmicforeach\ #1\ \algorithmicdo}
\algnewcommand{\lForEach}[2] {\State \algorithmicforeach\ #1 \algorithmicdo\ #2} 

\algdef{SE}[SUBALG]{Indent}{EndIndent}{}{\algorithmicend\ }%
\algtext*{Indent}
\algtext*{EndIndent}

\newtoggle{extended-version}

\usepackage{tikz}
\usetikzlibrary{decorations.pathmorphing}
\newcommand{\tikzmark}[1]{\tikz[remember picture, baseline] \node[inner sep=0pt, outer sep=0pt] (#1){};}

\newcommand{\textlink}[3]{
\begin{tikzpicture}[remember picture, overlay, >=stealth, shift={(0,0)}] 
	\draw[arrows=->] (#1) to node[sloped,anchor=center,above] {{\smaller #3}} (#2);
\end{tikzpicture}}
\newcommand{\textunderlink}[3]{
\begin{tikzpicture}[remember picture, overlay, >=stealth, shift={(0,0)}] 
	\draw[arrows=->] (#1) to node[sloped,anchor=center,below] {{\smaller #3}} (#2);
\end{tikzpicture}}
\newcommand{\textcurvelink}[5]{
\begin{tikzpicture}[remember picture, overlay, >=stealth, shift={(0,0)}] 
	\draw[arrows=->] (#1) to [out=#3,in=#4] node[sloped,anchor=center,above] {{\smaller #5}} (#2); 
\end{tikzpicture}}

\newcommand{\back}[4]{
\begin{tikzpicture}[remember picture, overlay, >=stealth, shift={(0,0)},circle dotted/.style={dash pattern=on 0.2mm off 1mm,line cap=round}] 
	\draw[circle dotted,line width=0.5mm,arrows=->] (#1) to [out=#3,in=#4] (#2); 
\end{tikzpicture}}
\newcommand{\forwardcurved}[4]{
\begin{tikzpicture}[remember picture, overlay, >=stealth, shift={(0,0)},circle dotted/.style={dash pattern=on 0.2mm off 1mm,line cap=round}]
	\draw[circle dotted,line width=0.5mm,arrows=->] (#1) to [out=#4,in=#3] (#2);
\end{tikzpicture}}

\newcommand{\init}[4]{
\begin{tikzpicture}[remember picture, overlay, >=stealth, shift={(0,0)}] 
	\draw[dashed,arrows=->] (#1) to [out=#3,in=#4] (#2); 
\end{tikzpicture}}

\newcommand{\num}[1]{}

\newcolumntype{H}{>{\setbox0=\hbox\bgroup}c<{\egroup}@{}} 
\newcolumntype{Z}{>{\setbox0=\hbox\bgroup}c<{\egroup}@{\hspace*{-\tabcolsep}}} 

\newcommand{\CS}[1]{\ensuremath{\mathit{CS(}#1)}\xspace}
\newcommand{\inCS}[2]{\ensuremath{#1\in\CS{#2}}}
\newcommand{\lockset}[1]{\ensuremath{\mathit{lockset}(#1)}\xspace}
\newcommand{\commonLocks}[2]{\ensuremath{\lockset{#1} \cap \lockset{#2}}}
\newcommand{\lastwr}[2]{\ensuremath{\mathit{lastwr}_{#2}(#1)}\xspace}
\newcommand{\nolastwr}{\ensuremath{\varnothing}\xspace}
\newcommand{\conflicts}[2]{\ensuremath{#1 \asymp #2}} 

\newcommand{\getLock}[1]{\ensuremath{\mathit{L}(#1)}\xspace}
\newcommand{\getAcquire}[1]{\ensuremath{\mathit{A}(#1)}\xspace}
\newcommand{\getRelease}[1]{\ensuremath{\mathit{R}(#1)}\xspace}
\newcommand{\getStatic}[1]{\textcolor{red}{\ensuremath{\mathit{srcloc}(#1)}}\xspace}
\newcommand{\getDeps}[1]{\ensuremath{\mathit{deps}(#1)}\xspace}
\DeclareDocumentCommand{\event}{mg}{\ensuremath{#1\IfValueT{#2}{#2}}\xspace}

\newcommand{\Gpath}[3]{\ensuremath{#1 \leadsto_{#3} #2}\xspace}

\newcommand{\edge}[2]{\ensuremath{(#1, #2)}\xspace}
\newcommand\Gv{\ensuremath{G}\xspace}
\newcommand{\dist}[2]{\ensuremath{\mathit{d(#1, #2)}}\xspace}

\newcommand{\CORE}{\textcolor{red}{CORE}\xspace}

\newcommand{\isCausal}[3]{\ensuremath{\mathit{causal(#1, #2, #3)}}\xspace}

\newcommand{\WBRdist}{\WBR-distance\xspace}
\newcommand{\Gdist}{\Gv-distance\xspace}

\newcommand{\LwFull}{Last writer\xspace}

\newcommand{\LW}{LW\xspace}

\newcommand{\LsFull}{Lock semantics\xspace}

\newcommand{\LS}{LS\xspace}

\newcommand{\tr}{\ensuremath{\mathit{tr}}\xspace}
\newcommand{\trPrime}{\ensuremath{\mathit{tr'}}\xspace}
\newcommand{\trDoublePrime}{\ensuremath{\mathit{tr''}}\xspace}
\newcommand{\PO}{PO\xspace}

\newcommand{\PoFull}{Program-order\xspace}

\newcommand{\HB}{HB\xspace}

\newcommand{\hbFull}{happens-before\xspace}
\newcommand{\CP}{CP\xspace}

\newcommand{\cpFull}{causally-precedes\xspace}
\newcommand{\WCP}{WCP\xspace}

\newcommand{\wcpFull}{weak-causally-precedes\xspace}
\newcommand{\WDC}{DC\xspace}

\newcommand{\wdcFull}{doesn't-commute\xspace}

\newcommand{\BR}{SDP\xspace}

\newcommand{\brFull}{strong-dependently-pre\-cedes\xspace}

\newcommand{\WBR}{WDP\xspace}

\newcommand{\wbrFull}{weak-dependently-pre\-cedes\xspace}

\newcommand{\DC}{\WDC}

\newcommand{\dcFull}{\wdcFull}
\newcommand{\DCOrdered}[2]{\WDCOrdered{#1}{#2}}

\newcommand\GetCOREReads{\textsc{Get\CORE{}Reads}\xspace}

\newcommand\checkWBRRace{\textsc{Vindicate\WBR{}Race}\xspace}
\newcommand\CheckWBRRace{\checkWBRRace}
\newcommand\checkDCRace{\textsc{Vindicate\DC{}Race}\xspace}

\newcommand\Def{Definition}

\newcommand{\ltTR}{\ensuremath{<_\textsc{\tr}}\xspace}

\newcommand{\ltTRPrime}{\ensuremath{<_\textsc{\trPrime}}\xspace}

\newcommand{\ltPO}{\ensuremath{\prec_\textsc{\tiny{\PO}}}\xspace}

\newcommand{\ltHB}{\ensuremath{\prec_\textsc{\tiny{\HB}}}\xspace}

\newcommand{\nltHB}{\ensuremath{\not\prec_\textsc{\tiny{\HB}}}\xspace}

\newcommand{\ltWCP}{\ensuremath{\prec_\textsc{\tiny{\WCP}}}\xspace}

\newcommand{\ltWDC}{\ensuremath{\prec_\textsc{\tiny{\WDC}}}\xspace}

\newcommand{\ltBR}{\ensuremath{\prec_\textsc{\tiny{\BR}}}\xspace}

\newcommand{\nltBR}{\ensuremath{\not\prec_\textsc{\tiny{\BR}}}\xspace}
\newcommand{\ltWBR}{\ensuremath{\prec_\textsc{\tiny{\WBR}}}\xspace}

\newcommand{\nltWBR}{\ensuremath{\not\prec_\textsc{\tiny{\WBR}}}\xspace}

\newcommand{\Ordered}[3]{\ensuremath{#1 #2 #3}}

\newcommand{\TROrdered}[2]{\Ordered{#1}{\ltTR}{#2}}

\newcommand{\TRPrimeOrdered}[2]{\Ordered{#1}{\ltTRPrime}{#2}}

\DeclareDocumentCommand{\POOrdered}{mmg}{\Ordered{#1}{\ltPO\IfValueT{#3}{#3}}{#2}}

\newcommand{\nHBOrdered}[2]{\Ordered{#1}{\nltHB}{#2}}
\newcommand{\HBOrdered}[2]{\Ordered{#1}{\ltHB}{#2}}

\newcommand{\WCPOrdered}[2]{\Ordered{#1}{\ltWCP}{#2}}

\newcommand{\WDCOrdered}[2]{\Ordered{#1}{\ltWDC}{#2}}

\newcommand{\BROrdered}[2]{\Ordered{#1}{\ltBR}{#2}}

\newcommand{\nBROrdered}[2]{\Ordered{#1}{\nltBR}{#2}}
\newcommand{\WBROrdered}[2]{\Ordered{#1}{\ltWBR}{#2}}

\newcommand{\nWBROrdered}[2]{\Ordered{#1}{\nltWBR}{#2}}

\newcommand{\Write}[1]{\ensuremath{\code{wr(#1)}}}
\newcommand{\Read}[1]{\ensuremath{\code{rd(#1)}}}
\newcommand{\Acquire}[1]{\ensuremath{\code{acq(#1)}}}
\newcommand{\Release}[1]{\ensuremath{\code{rel(#1)}}}
\newcommand{\Sync}[1]{\ensuremath{\code{sync(#1)}}}
\newcommand{\BrUniv}{\ensuremath{\code{br}}\xspace}

\newcommand\AcquireT[2]{\ensuremath{\Acquire{#1}^\thr{#2}}}

\newcommand\ReleaseT[2]{\ensuremath{\Release{#1}^\thr{#2}}}

\newcommand\WriteT[2]{\ensuremath{\Write{#1}^\thr{#2}}}

\newcommand\BrDepsOn[2]{\ensuremath{\mathit{brDepsOn}(#1, #2)}}

\newcommand{\addvc}{\sqcup}
\newcommand{\lessvc}{\sqsubseteq}

\newcommand{\Bfor}[3]{\textbf{foreach}\xspace \ensuremath{#1 \in #2} \textbf{do} #3}
\newcommand{\update}[2]{\ensuremath{#1 \gets #1 #2}}

\newcommand\notes[1]{\begin{quote}\textcolor{darkgreen}{\textbackslash \textbf{notes\{}} #1 \textcolor{darkgreen}{\}}\end{quote}}

\newcommand\later[1]{\begin{quote}\textcolor{darkgreen}{\textbackslash \textbf{later\{}} #1 \textcolor{darkgreen}{\}}\end{quote}}

\renewcommand{\mid}[0]{:}

\renewcommand{\grumbler}[2]{}
\renewcommand{\notes}[1]{}
\renewcommand{\later}[1]{}

\toggletrue{extended-version}

\begin{document}

\title{Dependence-Aware, Unbounded Sound Predictive Race Detection%
\iftoggle{extended-version}{\title{{Dependence-Aware, Unbounded Sound Predictive Race Detection}\subtitle{}}}{}}

\titlenote{This material is based upon work supported by the National Science Foundation
under Grants CAREER-1253703, CCF-1421612, and XPS-1629126.}


\iftoggle{extended-version}{\subtitle{\normalsize\minibox[frame]{This extended arXiv version of an OOPSLA 2019 paper adds Appendices~\ref{sec:completeness-lemma}--\ref{sec:extended-results}
and \\ corrects a technical issue (see corrigendum on the ACM DL)}}{}}


\author{Kaan Gen\c{c}}
\affiliation{
  \institution{Ohio State University}            
  \country{USA}                    
}
\email{genc.5@osu.edu}          

\author{Jake Roemer}
\affiliation{
  \institution{Ohio State University}
  \country{USA}                    
}
\email{roemer.37@osu.edu}          

\author{Yufan Xu}
\affiliation{
  \institution{Ohio State University}
  \country{USA}                    
}
\email{xu.2882@osu.edu}          

\author{Michael D. Bond}
\affiliation{
  \institution{Ohio State University}
  \country{USA}                    
}
\email{mikebond@cse.ohio-state.edu}          

\begin{abstract}

Data races are a real problem for parallel software, yet hard to detect.
Sound predictive analysis observes a program execution and detects data races
that exist in some \emph{other, unobserved} execution.
However, existing predictive analyses miss races because they do not
scale to full program executions or do not precisely incorporate data and control dependence.

This paper introduces two novel, sound predictive approaches that incorporate data and control dependence and handle full program executions.
An evaluation using real, large Java programs shows that these approaches
detect more data races than the closest related approaches,
thus advancing the state of the art in sound predictive race detection.

\end{abstract}

	

\maketitle


\section{Introduction}
\label{sec:intro}

\notes{
\mike{\BR is the weakest known sound partial order, and \WBR is the strongest known complete partial order (?).}
}


With the rise in parallel software,
\emph{data races} represent a growing hazard.
Programs with data races written in shared-memory languages including Java and C++
have weak or undefined semantics,
as a result of assuming data race freedom for performance reasons~\cite{java-memory-model,c++-memory-model-2008,memory-models-cacm-2010}.
Data races are culprits in real software failures, resulting in substantial financial losses and even harm to humans~\cite{boehm-miscompile-hotpar-11, portend-toplas15,
conc-bug-study-2008, portend-asplos12, benign-races-2007, prescient-memory,
adversarial-memory, racefuzzer, relaxer, cost-of-software-errors,blackout-2003-tr,therac-25,nasdaq-facebook}.

Writing scalable, data-race-free code is challenging, as is
detecting data races, which occur nondeterministically
depending on shared-memory interleavings and program inputs and environments.
The most common approach for dealing with data races is to detect them
during in-house testing using dynamic \emph{\hbFull (\HB)} analysis~\cite{fasttrack,multirace,goldilocks-pldi-2007,google-tsan-v1,google-tsan-v2,intel-inspector},
which detects conflicting accesses (two memory accesses, at least one of which is a write, to the same variable by different threads)
unordered by the \HB partial order~\cite{happens-before}.
However, \HB analysis misses data races when accesses \emph{could} race
in some \emph{other} execution but are ordered by critical sections on the same lock in the observed execution.

A promising alternative to \HB analysis is \emph{sound predictive analysis},
which detects additional predictable data races from an observed
execution~\cite{rvpredict-pldi-2014,said-nfm-2011,rdit-oopsla-2016,jpredictor,maximal-causal-models,ipa,causally-precedes,wcp,vindicator,pavlogiannis-2019};
an analysis is \emph{sound} if it detects no false races (Section~\ref{sec:background}).
Some predictive analyses rely on generating and solving SMT constraints, so in practice they
cannot scale to full program executions and instead analyze \emph{bounded windows} of execution,
missing races between accesses that do not execute close
together~\cite{rvpredict-pldi-2014,said-nfm-2011,rdit-oopsla-2016,jpredictor,maximal-causal-models,ipa} (Section~\ref{sec:related}).
In contrast, \emph{unbounded} predictive analyses avoid this limitation
by detecting races based on computing a partial order weaker than \HB,
using analyses with linear running time in the length of the trace~\cite{wcp,vindicator}.
However, these partial-order-based analyses miss predictable races
because they do not incorporate precise notions of \emph{data and control dependence}.
More precisely, existing predictive partial orders do \emph{not} encode the precise
conditions for reordering memory accesses to expose a race:
The reordering can change the last writer of a memory read (data dependence)
if the read in turn cannot affect whether the racing accesses execute (control dependence).
Encoding data and control dependence precisely in a partial order is fundamentally challenging (Section~\ref{sec:background}).

\paragraph{Contributions.}


This paper designs and evaluates new predictive analyses,
making the following contributions:

\begin{itemize}

  \item A partial order called \emph{\brFull (\BR)}
  that improves over the highest-coverage sound partial order from prior work~\cite{wcp} by
  incorporating data dependence more precisely
  (Section~\ref{sec:BR-WBR-relations}).
  
  \item A proof that \BR is sound, \ie, detects no false races (Section~\ref{sec:br-soundness-proof}).

  \item A partial order called \emph{\wbrFull (\WBR)}
  that improves over the previous highest-coverage partial order from prior work~\cite{vindicator} by incorporating
  data and control dependence precisely
  (Section~\ref{sec:BR-WBR-relations}).
  


  \item A proof that \WBR is complete (sometimes called
  \emph{maximal}~\cite{rvpredict-pldi-2014,maximal-causal-models}),
  detecting all races knowable from an observed execution (Section~\ref{sec:wbr-completeness}).

  \item Dynamic analyses that compute \BR and \WBR and detect \BR- and \WBR-races
  (Section~\ref{sec:wbr-analysis}).
  
  \item An algorithm for filtering out \WBR-races that are false races,
  by extending prior work's \emph{vindication} algorithm~\cite{vindicator},
  yielding an overall sound approach (Section~\ref{sec:vindication-summary}).
  \notes{Our results show that, in contrast with prior work's results vindicating \DC-races~\cite{vindicator},
  vindication is essential in practice to rule out false \WBR-races.}%

  \item An implementation and evaluation of \BR and \WBR analyses and \WBR-race vindication
  on benchmarked versions of real, large Java
  programs (Section~\ref{sec:evaluation}).
  The evaluation shows that the analyses
  find predictable races missed by the closest
  related approaches~\cite{wcp,vindicator,rvpredict-pldi-2014}.
  \notes{which miss races due to imprecise handling of data and control
  dependence or being restricted to analyzing bounded windows.}%
\end{itemize}

\section{Background and Motivation}
\label{sec:background}

Recent partial-order-based predictive analyses can scale to full program executions,
enabling detection of predictable races that are millions of executed operations apart~\cite{wcp,vindicator}.
However, these partial orders are fundamentally limited and miss predictable races,
as this section explains. First, we introduce formalisms used throughout the paper.

\subsection{Execution Model}
\label{subsec:execution-trace}


An \emph{execution trace} \tr is a sequence of events, ordered by the total order \ltTR,
that represents a multithreaded execution without loss of generality,
corresponding to a linearization of a sequentially consistent (SC) execution.\footnote{Although
programs with data races may violate SC~\cite{java-memory-model,c++-memory-model-2008,memory-models-cacm-2010,dolan-bounding-races},
dynamic race detection analyses (including ours) add synchronization instrumentation before accesses, generally ensuring SC.}
We assume
every event in \tr is a unique object (\eg, has a unique identifier),
making it possible to identify the same event \event{e}{} across other, predicted traces.
Each event has two attributes:
(1) an identifier for the thread that executed the operation; and
(2) an operation, which is one of \Write{x}, \Read{x}, \Acquire{m}, \Release{m}, or \BrUniv,
where \code{x} is a program memory location and \code{m} is a program lock.
(Later we consider how to extend analyses to handle lock-free accesses and Java \code{volatile} / C++ \code{atomic} accesses.)
An execution trace must be \emph{well formed}:
a thread may only acquire an unheld lock
and may only release a lock it has acquired.

Each \BrUniv (branch) event \event{b}{} represents an executed conditional operation---such as a conditional jump,
polymorphic call, or array element access---that may be dependent on some prior read event(s) by the same thread.
We assume a helper function \BrDepsOn{\event{b}{}}{\event{r}{}} exists
that returns true if the value read by read event \event{r}{}
may affect \event{b}{}'s outcome.
An
implementation could use
static dependence analysis to identify reads on which a branch is data dependent.
For simplicity, the paper's examples assume \BrDepsOn{\event{b}}{\event{r}{}} always returns true,
\ie, every branch is assumed dependent on preceding reads by the same thread.
Our implementation and evaluation make the same assumption, as explained later.
This assumption limits the capability of predictive analysis to predict different executions;
in other words, it limits the number of knowable data races from a single execution.

Three example traces are shown
in Figures~\ref{fig:CORE:has-predictable-race}, \ref{fig:CORE:valid-reordering}, and \ref{fig:CORE:no-predictable-race},
in which top-to-bottom order represents
trace order, and column placement denotes an event's executing thread.
We discuss these examples in detail later.

\notes{
A static source location of an event is a unique identifier for an instruction within the analyzed program, which when executed resulted in the event. Many events may share the same static source location as same instructions execute again, however the resulting events will not necessarily access the same memory or acquire/release the same lock.
The function \getDeps{\event{b}{}} returns the set of static source locations that perform
read events that may affect \event{b}{}. \getDeps{\event{b}{}} might be based on conservative static analysis,
or it could be the universe set of all read source locations (\eg, in the absence of static analysis).
Let \getStatic{\event{e}{}} be the static source location of read event \event{e}{}. 
Then \BrDepsOn{\event{b}{}}{\event{e}{}} holds if \event{b}{} is potentially dependent on \event{e}{}, \ie, a path from \getStatic{\event{e}} to \getStatic{\event{b}} exists in the system dependency graph.
In the absence of static analysis---and in this paper's example executions, implementation, and
evaluation---we assume conservatively that each branch is dependent on all prior reads,
\ie, \BrDepsOn{\event{b}}{\event{e}} holds for any \event{e}{}.
\kaan{Reviewer B says: ``this seems imprecise formulation: what is the connection between e and b?''. I tried addressing the comment by explaining that "depends" means "path on SDG", but I'm not sure if that is the question here.
\mike{I think that's in the right direction, except SDG isn't defined, and we don't actually use it in the future.}}
}

Two read or write events to the same variable are \emph{conflicting},
notated \conflicts{\event{e}{}}{\event{e}{'}},
if the events are executed by different threads and at least one is a write.

\emph{\PoFull (\PO)} is a partial order that orders events in the same thread:
\POOrdered{\event{e}{}}{\event{e}{'}} if \TROrdered{\event{e}{}}{\event{e}{'}} and the events are executed by the same thread.

The function
\CS{\event{e}{}} returns the set of events in the critical section started or ended by acquire or release event \event{e}{},
including the bounding acquire and release events.
\getRelease{\event{a}{}} returns the release event ending the critical section started by acquire event \event{a}{}, and
\getAcquire{r} returns the acquire event starting the critical section ended by release event \event{r}{}.
The function \lockset{\event{e}{}} returns the set of locks held at a read or write event
\event{e}{} by its executing thread.

\subsection{Predictable Traces and Predictable Races}
\label{subsec:execution-reordering}


\newsavebox{\firstlisting}
\begin{lrbox}{\firstlisting}
\begin{lstlisting}
int z = 0, y = 0; Object m = new Object();
new Thread(() -> { synchronized (m) {
                     int t = z;
                     y = 1;
                 } }).start();
new Thread(() -> { synchronized (m) {
                     z = 1;
                     x = 1;
                 } }).start();
new Thread(() -> { synchronized (m) {
                     int t = x;
                     if (t == 0) return;
                   }
                   int t = y;
                 } }).start();
\end{lstlisting}
\end{lrbox}

\newsavebox{\secondlisting}
\begin{lrbox}{\secondlisting}
\begin{lstlisting}
int z = 0, y = 0; Object m = new Object();
new Thread(() -> { synchronized (m) {
                     int t = z;
                     %*\textbf{\textcolor{blue}{if} (t == 0)}*)
                       y = 1;
                 } }).start();
new Thread(() -> { synchronized (m) {
                     z = 1;
                     x = 1;
                 } }).start();
new Thread(() -> { synchronized (m) {
                     int t = x;
                     if (t == 0) return;
                   }
                   int t = y;
                 } }).start();
\end{lstlisting}
\end{lrbox}

\begin{figure*}
\captionsetup{farskip=0pt} 
\footnotesize
\centering
\sf

\subfloat[Java code that could lead to the executions in (b) and (c).]{
\makebox[.35\linewidth]{\usebox{\firstlisting}}
\label{fig:CORE:sample-code}}
\hfill
\subfloat[Execution with a predictable race]{
\renewcommand\tikzmark[1]{\relax}
\makebox[.27\linewidth]{
\begin{tabular}{@{}lll@{}}
\textnormal{Thread 1} & \textnormal{Thread 2} & \textnormal{Thread 3}\\\hline
\Acquire{m} \\
\Read{z} \\
\Write{y} \\
\Release{m}\tikzmark{1}\\
                      & \Acquire{m} \\
                      & \Write{z} \\
                      & \Write{x} \\
                      & \Release{m}\tikzmark{7} \\
                      &             & \Acquire{m} \\
                      &						  & \tikzmark{8}\Read{x} \\
                      &						  & \tikzmark{9}\BrUniv \\
                      & 					  & \Release{m} \\
                      &						  & \Read{y} \\\\\\
\end{tabular}}
\label{fig:CORE:has-predictable-race}
}
\hfill
\subfloat[Predictable trace of (b)]{
\renewcommand\tikzmark[1]{\relax}
\makebox[.27\linewidth]{
\begin{tabular}{@{}lll@{}}
\textnormal{Thread 1} & \textnormal{Thread 2} & \textnormal{Thread 3}\\\hline
                      & \Acquire{m} \\
                      & \Write{z} \\
                      & \Write{x} \\
                      & \Release{m}\tikzmark{7}\\
                      &             & \Acquire{m} \\
                      &						  & \tikzmark{8}\Read{x} \\
                      &						  & \tikzmark{9}\BrUniv \\
                      & 					  & \Release{m} \\
\Acquire{m} \\
\Read{z} \\
\Write{y} \\
                      &						  & \Read{y} \\\\\\\\
\end{tabular}}
\label{fig:CORE:valid-reordering}
}\\\medskip
\subfloat[Java code that could lead to the execution in (e).]{
\makebox[.5\linewidth]{\usebox{\secondlisting}}
\label{fig:CORE:sample-code-norace}}
\subfloat[Execution with no predictable race]{
\renewcommand\tikzmark[1]{\relax}
\makebox[.5\linewidth]{
\begin{tabular}{@{}lll@{}}
\textnormal{Thread 1} & \textnormal{Thread 2} & \textnormal{Thread 3}\\\hline
\Acquire{m} \\
\Read{z} \\
\textbf{\BrUniv} \\
\Write{y} \\
\Release{m}\tikzmark{1} & \\
            & \Acquire{m} \\
					  & \tikzmark{2}\Write{z} \\
				 	  & \Write{x} \\
					  & \Release{m}\tikzmark{7}\\
                      &             & \Acquire{m} \\
                      &						  & \tikzmark{8}\Read{x} \\
                      &						  & \tikzmark{9}\BrUniv \\
                      & 					  & \Release{m} \\
                      &						  & \Read{y} \\\\\\
\end{tabular}}
\label{fig:CORE:no-predictable-race}
}

\caption{Two code examples, with potential executions they could lead to.
    The execution in (b) has a predictable race, as demonstrated by the predictable trace in (c).
The execution in (e) has no predictable race.
\label{fig:CORE}}
\end{figure*}

By observing one execution of a program, it is possible to predict data races in both
the observed execution and some \emph{other} executions of the program.
The information present in the observed execution implies the existence of other, different executions,
called \emph{predictable traces}.
To define what traces can be predicted from an observed trace, we first define several relevant concepts.

\begin{definition}[Last writer]
  Given a trace \tr, let \lastwr{\event{r}{}}{\tr} for a read event \event{r}{} be
  the last write event before \event{r}{} in \tr that accesses the same variable as \event{r}{},
  or \nolastwr if no such event exists.
\notes{
More formally,
\begin{align*}
\lastwr{\event{r}{}}{\tr} \coloneqq
\begin{cases}
  \event{w}{} & \textnormal{if } \, \exists \event{w}{} \mid \TROrdered{\event{w}{}}{\event{r}{}} \land \conflicts{\event{w}{}}{\event{r}{}} \land
  (\nexists \event{w}{'} \mid \TROrdered{\event{w}{}}{\TROrdered{\event{w}{'}}{}{\event{r}{}}} \land \conflicts{\event{w}{'}}{\event{r}{}}) \\
  \nolastwr & \textnormal{otherwise}
\end{cases}
\end{align*}
}%
\end{definition}

To ensure that a predictable trace is feasible,
each read in a predictable trace must have the same last writer as in the observed trace---with one exception:
a read can have a different last writer if the read cannot take the execution down a different control-flow path than the observed execution.
An example of such a read is Thread~1's \Read{z} event in Figure~\ref{fig:CORE:valid-reordering}.
Next, we introduce a concept that helps in identifying reads whose last writer must be
preserved in a predictable trace.

\newcommand{\essential}{\ensuremath{\mathit{S}}\xspace}

\begin{definition}[Causal events]
Given a trace \tr, set of events \essential, and event \event{e}, let
\isCausal{\tr}{\essential}{\event{e}} be a function that returns true if at
least one of the following properties holds, and false otherwise.
\begin{itemize}
\item \event{e} is a read, and
there exists a branch event \event{b} such that
$\event{b}\in\essential \land \POOrdered{\event{e}}{\event{b}} \land \BrDepsOn{\event{b}}{\event{e}}$.
\item \event{e} is a write, and
there exists a read event \event{e}{'} such that
$\event{e}{'} \in \essential \land \event{e} = \lastwr{\event{e}{'}}{\tr}$.
\item \event{e} is a read, and
there exists a write event \event{e}{'} such that
$\event{e}{'} \in \essential \land \POOrdered{\event{e}}{\event{e}{'}}$
(\event{e}{} and \event{e}{'} may access different variables).
\end{itemize}
\label{def:causal-events}
\end{definition}

\notes{
\jake{The set S feels disconnected from tr.
As in, S could be any set of events that might not even be in tr.
Could S be described as some subset of events in tr?
\mike{Hmm, technically it works because in each bullet above,
the event in \essential is always associated with \tr.}
\mike{Do we ever only use \isCausal{\tr}{\essential}{\event{e}}
where \essential is actually a trace \trPrime?
If so, should we just have the second parameter be a trace \trPrime?}}
}%
Intuitively, \isCausal{\tr}{\essential}{\event{e}} tells us whether an event \event{e}
could have affected some event \event{e}{'} in \essential directly.
For example, \isCausal{\tr}{\essential}{\event{e}}
if read \event{e} may affect a branch event in \essential;
if \event{e} writes a variable later read by an event in \essential; or
if \event{e} reads a value that may affect a later write by the same thread in \essential
(even if the read and write are to different variables, to account for intra-thread data flow).
We can now define a predictable trace of an observed trace, which
is a trace that is definitely a feasible execution of the program, given the existence of the observed execution.

\begin{definition}[Predictable trace]
  An execution trace \trPrime is a \emph{predictable trace} of trace \tr if \trPrime
  contains only events in \tr
  (\ie, $\forall \event{e}{} \mid \event{e}{} \in \trPrime \implies \event{e}{} \in \tr$)
  and all of the following rules hold:
  
  \smallskip
  \emph{Program order (\PO) rule:~}
  For any events \event{e}{_1} and \event{e}{_2}, if
  \POOrdered{\event{e}{_1}}{\event{e}{_2}}, then
  $\TRPrimeOrdered{\event{e}{_1}}{\event{e}{_2}} \lor \event{e}{_2} \notin \trPrime$.
  
  \smallskip
  \emph{\LwFull (\LW) rule:~} For every read event \event{e}{}
  such that $\isCausal{\tr}{\trPrime}{\event{e}}$,
  $\lastwr{\event{e}}{\trPrime}=\lastwr{\event{e}}{\tr}$.
  (In this context, \trPrime means the set of events in the trace \trPrime.)

  \smallskip
  \emph{\LsFull (\LS) rule:~} For acquire events \event{e}{_1} and
  \event{e}{_2} on the same lock, if \TRPrimeOrdered{\event{e}{_1}}{\event{e}{_2}} then
  \TRPrimeOrdered{\event{e}{_1}}{\TRPrimeOrdered{\getRelease{\event{e}{_1}}}{\event{e}{_2}}}.\label{def:valid-reordering}
\end{definition}

The \PO and \LW rules ensure key properties from \tr also hold in \trPrime,
while the \LS rule ensures that \trPrime is well formed.
The intuition behind the \LW rule is that any read that may (directly or indirectly) affect the control flow
of the program must have the same last writer in predictable trace \trPrime as in observed trace \tr.


Note that throughout the paper,
partial ordering notation such as \Ordered{\event{e}{}}{\prec}{\event{e}{'}} refers to the order of \event{e}{} and \event{e}{'}
in the \emph{observed} trace \tr (not a predictable trace \trPrime).

Predictable traces do not in general contain every event in the observed trace they are based on.
For the purposes of race detection, a predictable trace will \emph{conclude with} a pair of conflicting events,
which are preceded by events necessary according to the definition of predictable trace.
For example, consider Figure~\ref{fig:CORE:valid-reordering},
which is a predictable trace of Figure~\ref{fig:CORE:has-predictable-race} that
excludes Thread~1's event after \Write{y}.
The \PO rule is satisfied, and
the \LS rule is satisfied after reordering Thread~2 and 3's critical sections before Thread~1's.
The \LW rule is satisfied because \Read{z} is not a causal event in \trPrime.


\begin{definition}[Predictable race]
An execution \tr has a predictable race if
a predictable trace \trPrime of \tr
has two \emph{conflicting}, \emph{consecutive} events:
$\conflicts{\event{e}{_1}}{\event{e}{_2}} \land
\TRPrimeOrdered{\event{e}{_1}}{\event{e}{_2}} \land (\nexists \event{e}{} \mid \TRPrimeOrdered{\event{e}{_1}}{\TRPrimeOrdered{\event{e}{}}{\event{e}{_2}}})$.
\label{def:predictable-race}
\end{definition}

Figure~\ref{fig:CORE:has-predictable-race} has a predictable race,
as demonstrated by the predictable trace
in Figure~\ref{fig:CORE:valid-reordering}.
In contrast, Figure~\ref{fig:CORE:no-predictable-race} has no predictable race.
The difference between Figures~\ref{fig:CORE:has-predictable-race} and
\ref{fig:CORE:no-predictable-race} is the \BrUniv event in Thread~1.
Since no \BrUniv exists in Thread~1 in Figure~\ref{fig:CORE:has-predictable-race},
\Read{z} is not a causal event,
which in turn allows the critical sections in Threads 2 and 3 to be reordered above the critical section in Thread 1,
allowing \Write{y} and \Read{y} to be consecutive in the predictable trace.
In contrast, a \BrUniv event executes before \Write{y} in Figure~\ref{fig:CORE:no-predictable-race}.
No predictable trace of this example can exclude \BrUniv without excluding \Write{y};
otherwise the \PO rule would be violated. \Read{z} is a causal event in any predictable
trace where \Write{y} is included, which makes it impossible to reorder the critical sections.
As a result, no predictable trace exists in which \Write{y} and \Read{y} are consecutive.
Figures~\ref{fig:CORE:sample-code} and \ref{fig:CORE:sample-code-norace} show source code that could lead to
the executions in Figures~\ref{fig:CORE:has-predictable-race} and \ref{fig:CORE:no-predictable-race}, respectively.
The code in Figure~\ref{fig:CORE:sample-code-norace} has no race;
in fact, any deviation of critical section ordering from Figure~\ref{fig:CORE:no-predictable-race}'s
causes \Write{y} or \Read{y} \emph{not} to execute.


\subsection{Existing Predictive Partial Orders}



Here we overview three relations introduced in prior work,
called \emph{\hbFull (\HB)},
\emph{\wcpFull (\WCP)}, and \emph{\dcFull (\DC)},
that can be computed in time linearly proportional to the length of the execution trace~\cite{wcp,vindicator}.
Intuitively, each relation orders events that may not be legal to reorder in a predictable trace,
so that two unordered conflicting events represent a true or potential data race (depending on whether the relation is sound).
An execution trace has an \emph{\HB-race}, \emph{\WCP-race}, or \emph{\DC-race} if it contains two conflicting events that are unordered by
\HB, \WCP, or \DC, respectively.

\paragraph{Definitions of relations.}

Table~\ref{tab:partial-order-definitions-prior}
gives definitions of \HB, \WCP, and \DC by presenting their properties comparatively.
The first two rows of the table say how the relations order critical sections on the same lock.
\HB orders all critical sections on the same lock, and it orders the first critical section's \Release{m}
to the second critical section's \Acquire{m}.
\WCP and \DC order only \emph{conflicting} critical sections (critical sections on the same lock containing conflicting events),
and they order from the first critical section's \Release{m} to the second critical section's conflicting access event.
That is,
if \event{r}{_1} and \event{r}{_2} are release events on the same lock such that
\TROrdered{\event{r}{_1}}{\event{r}{_2}}, and
\event{e}{_1} and \event{e}{_2} are conflicting events (\conflicts{\event{e}{_1}}{\event{e}{_2}}) such that
$\event{e}{_1} \in \CS{\event{r}{_1}} \land\event{e}{_2} \in \CS{\event{r}{_2}}$,
then \WCPOrdered{\event{r}{_1}}{\event{e}{_2}} and \DCOrdered{\event{r}{_1}}{\event{e}{_2}}.
The intuition behind these properties of \WCP and \DC is that
non-conflicting critical sections can generally be reordered in a predictable trace;
and even in the case of conflicting critical sections, the second critical section can be ``reordered''
so that it executes only up to its conflicting access and the first critical section does not execute
at all in the predictable trace.

\begin{table}
\newcommand\yes{Yes} 
\newcommand\no{No}
\newcommand{\ltSO}{\ensuremath{\prec_\textsc{\tiny{SO}}}\xspace}
\small
\begin{tabular}{@{}l|lll@{}}
Property
 & \ltHB              & \ltWCP & \ltWDC \\\hline
Same-lock critical section ordering
 & All                & Confl.    & Confl.\\
Orders \code{rel} to\dots
 & \code{acq}         & \code{wr}/\code{rd} & \code{wr}/\code{rd} \\
 \hline
 Includes \ltPO? 
 & \yes               & \no    & \yes \\
Left-and-right composes with \ltHB?
 & \yes               & \yes   & \no \\
 \hline
$\Acquire{m}\prec\Release{m}$ implies $\Release{m}\prec\Release{m}$?
  & \yes               & \yes   & \yes \\
 Transitive?
  & \yes               & \yes   & \yes \\
 \end{tabular}


\caption{Definitions of three strict partial orders over events in an execution trace.
Each order is the minimum relation satisfying the listed properties.}
\label{tab:partial-order-definitions-prior}

\end{table}

The next two table rows show whether the relations include \PO or compose with \HB.
\HB and \DC include (\ie, are supersets of) \PO: if \POOrdered{\event{e}{_1}}{\event{e}{_2}},
then \HBOrdered{\event{e}{_1}}{\event{e}{_2}} and \DCOrdered{\event{e}{_1}}{\event{e}{_2}}.
In contrast, \WCP does not include \PO but instead \emph{composes with} the stronger \HB:
if \WCPOrdered{\event{e}{_1}}{\HBOrdered{\event{e}{_2}}{\event{e}{_3}}}
or \HBOrdered{\event{e}{_1}}{\WCPOrdered{\event{e}{_2}}{\event{e}{_3}}},
then \WCPOrdered{\event{e}{_1}}{\event{e}{_2}}.
(By virtue of being transitive, \HB composes with itself.)
The intuition behind including or composing with \PO (a subset of \HB)
is that \PO-ordered events cannot be reordered in a predictable trace.
The intuition behind \WCP composing with \HB, in essence, is to avoid predicting traces
that violate the \LS rule of predictable traces.
As a result, \WCP is sound while \DC is unsound, as we will see.

The last two rows show properties shared by all relations.
First, if two critical sections on the same lock
$\event{a}{_1} \ltPO \event{r}{_1} \ltTR \event{a}{_2} \ltPO \event{r}{_2}$
are ordered at all (meaning simply \Ordered{\event{a}{_1}}{\ensuremath{\prec_\ast}}{\event{r}{_2}}
because all relations minimally compose with \PO),
then their release events are ordered (\Ordered{\event{r}{_1}}{\ensuremath{\prec_\ast}}{\event{r}{_2}}).
Second, all of the relations are transitive. As a result of being transitive, antisymmetric, and irreflexive,
all of the relations are strict partial orders.

\paragraph{Example}


As an example of \WCP and \DC ordering, consider the execution in
Figure~\ref{fig:br-rd-wr-conflict:wbr-no-edge}.
Both relations order Thread~1's \Release{m} to Thread~2's \Write{x}
because the critical sections on \code{m} contain conflicting accesses to \code{x}.
By \WCP's composition with \HB (and thus \PO) and \DC's inclusion of \PO,
both \WCP and \DC transitively order \Read{x} to \Write{x}
and \Write{y} to \Read{y} (\WCPOrdered{\Write{y}}{\Read{y}} and \DCOrdered{\Write{y}}{\Read{y}}), so the execution has no \WCP- or \DC-races.

\later{
Figure~\ref{??} shows \WCP and \DC ordering (they are the same for this execution)
for the execution from Figure~\ref{fig:CORE:has-predictable-race}.
\WCP and \DC cannot predict the race because\dots
\mike{How about adding an example that shows why \WCP and \DC miss races; I'd suggest it be an example showing ordering for Figure~\ref{fig:CORE:has-predictable-race}?
I think the figures in the next subsection aren't really good for showing how \WCP and \DC work, or at least the \BR/\WBR ordering on them will confuse readers here.}
Figure~\ref{??} shows\dots
\mike{And next to that example, how about adding an example that differentiates \WCP and \DC?}
\mike{That might be overkill. Instead for now I've added an explanation using Figure~\ref{fig:br-rd-wr-conflict:wbr-no-edge}.}
}

\paragraph{Soundness and completeness.}

A relation or analysis is \emph{sound} if it detects a race only for an execution trace with a predictable race or
deadlock.\footnote{A trace has a \emph{predictable deadlock} if there exists a valid reordering with a deadlock.
We define soundness to include predictable deadlocks because prior work's
\WCP relation~\cite{wcp} and our \BR relation are sound in this way.}
%
A relation or analysis is \emph{complete} if it detects a race for every execution trace with a predictable race.
\label{subsec:completeness}

\WCP (and \HB) are sound: a \WCP-race (\HB-race) indicates a predictable race or deadlock.
\DC is unsound: an execution with a \DC-race may have no predictable race or deadlock.
However, prior work shows that \DC-races are generally true predictable races in practice,
and an efficient \emph{vindication} algorithm can
verify \DC-races as predictable races by computing additional constraints and building a predictable trace exposing the race~\cite{vindicator}.
Later in the paper, we provide more details about vindication, when introducing a new relation that (like \DC)
is unsound and makes use of a vindication algorithm.


\subsection{Limitations of Existing Predictive Partial Orders}
\label{sec:background:limitations}

\WCP and \DC analyses are the state of the art in detecting as many predictable races as possible using
online, unbounded analysis~\cite{wcp,vindicator}.
However, \WCP and \DC are \emph{incomplete},
failing to detect some predictable races.
\WCP and \DC are overly strict
because they order all conflicting accesses,
conservatively ruling out some predictable traces that still preserve the last writer of each causal read.
This strictness arises from imprecise handling of data and control dependence:

\paragraph{Data dependence:}

\WCP and \DC order all conflicting accesses, which is imprecise
because the order of a write--write or read--write conflict does not necessarily need to be preserved
to satisfy the last-writer (\LW) rule of predictable traces.

\begin{figure}
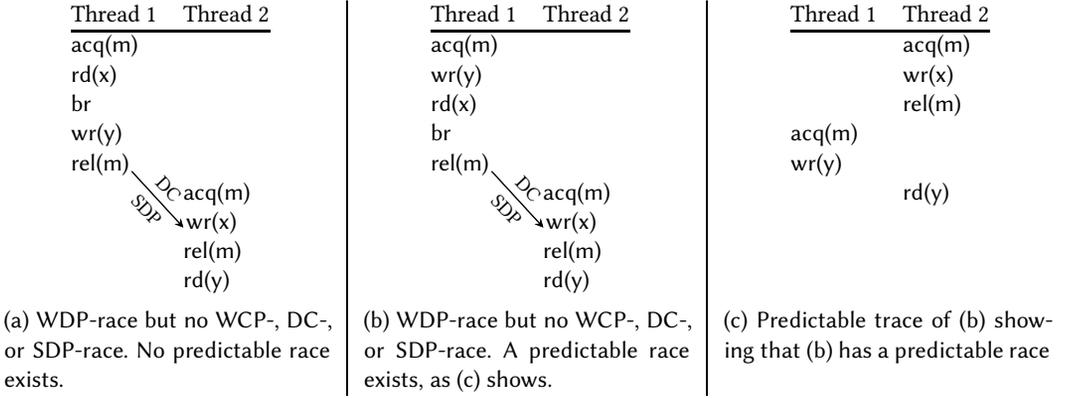

\captionsetup{farskip=0pt} 
\small
\centering

\subfloat[\WBR-race but no \WCP-, \DC-, or \BR-race. No predictable race exists.]{
\centering
\sf
\makebox[.3\linewidth]{
\begin{tabular}{@{}ll@{}}
\textnormal{Thread 1} & \textnormal{Thread 2} \\\hline 
acq(m)                & \\
rd(x)                 & \\
\BrUniv               & \\
wr(y)                 & \\
rel(m)\tikzmark{5}    & \\
                      & acq(m) \\
                      & \tikzmark{6}wr(x) \\
                      & rel(m) \\
                      & rd(y)
\end{tabular}}
\label{fig:br-rd-wr-conflict:wbr-false-race}}
\hfill\vrule\hfill
\subfloat[\WBR-race but no \WCP-, \mbox{\DC-,} or \BR-race. A predictable race exists, as (c) shows.]{
\centering
\sf
\makebox[.3\linewidth]{
\begin{tabular}{@{}ll@{}}
\textnormal{Thread 1} & \textnormal{Thread 2} \\\hline
acq(m)                & \\
wr(y)                 & \\ 
rd(x)                 & \\
\BrUniv                 \\
rel(m)\tikzmark{3}    & \\
                      & acq(m) \\
                      & \tikzmark{4}wr(x) \\
                      & rel(m) \\
                      & rd(y)
\end{tabular}}
\label{fig:br-rd-wr-conflict:wbr-no-edge}}
\hfill\vrule\hfill
\subfloat[Predictable trace of (b) showing that (b) has a predictable race]{
\centering
\sf
\makebox[.3\linewidth]{
\begin{tabular}{@{}ll@{}}
\textnormal{Thread 1} & \textnormal{Thread 2} \\\hline
                      & acq(m) \\
                      & wr(x) \\
                      & rel(m) \\
acq(m)                & \\
wr(y)                 & \\
                      & rd(y) \\
\\
\\
\\
\end{tabular}}
\label{fig:br-rd-wr-conflict:wbr-no-edge-reordered}}

\textlink{3}{4}{\DC}%
\textunderlink{3}{4}{\BR}%
\textlink{5}{6}{\DC}%
\textunderlink{5}{6}{\BR}%


\caption{Executions showing \WCP and \DC's overly strict handling of read--write dependencies.
Edges represent ordering, labeled using the weakest applicable relation(s) (and omitting ordering established by \HB alone),
implying ordering by strictly stronger relations (see Figure~\ref{fig:lattices:relations} for comparison of relations).%
\label{fig:br-rd-wr-conflict}}

\end{figure}
\begin{figure}
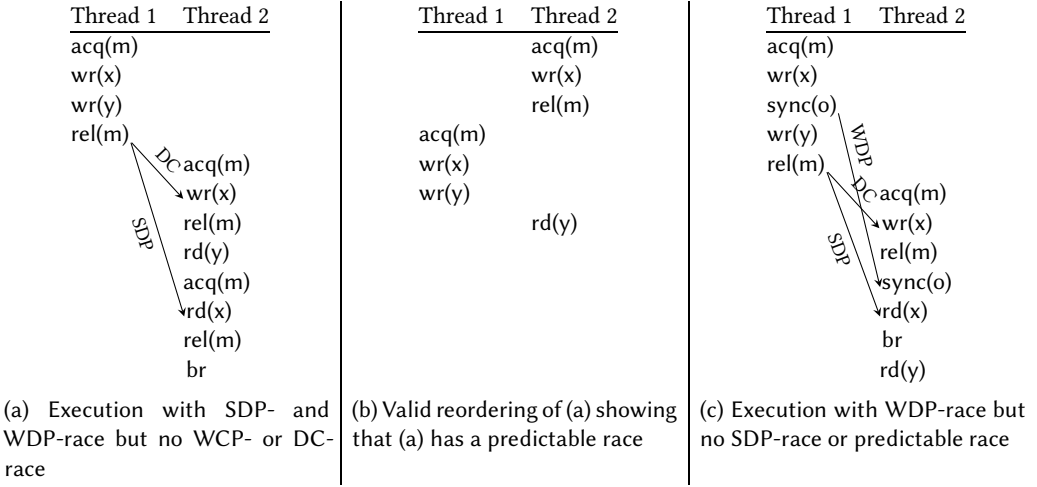


\small
\centering

\subfloat[Execution with \BR- and \WBR-race but no \WCP- or \DC-race]{
\centering
\sf
\makebox[.45\linewidth]{
\begin{tabular}{@{}ll@{}}
\textnormal{Thread 1} & \textnormal{Thread 2}\\\hline
\Acquire{m} \\
\Write{x} \\
\Write{y} \\
\Release{m}\tikzmark{3}\\
                      & acq(m) \\
                      & \tikzmark{4}wr(x) \\
                      & rel(m) \\
                      & rd(y) \\
                      & acq(m) \\
                      & \tikzmark{5}rd(x) \\
                      & rel(m) \\
                      & \tikzmark{2}\BrUniv \\
\end{tabular}}
\label{fig:wr-wr-conflict:original}}
\hfill
\subfloat[Valid reordering of (a) showing that (a) has a predictable race]{
\centering
\sf
\makebox[.45\linewidth]{
\begin{tabular}{@{}ll@{}}
\textnormal{Thread 1} & \textnormal{Thread 2}\\\hline
                      & acq(m) \\
                      &	wr(x) \\
                      &	rel(m) \\
\Acquire{m} \\
\Write{x} \\
\Write{y} \\
                      &	rd(y) \\
\\ \\ \\ \\ \\
\end{tabular}}
\label{fig:wr-wr-conflict:reordered}}
\hfill
\subfloat[Execution with \WBR-race but no \BR-race or predictable race]{
\centering
\sf
\makebox[.45\linewidth]{
\begin{tabular}{@{}ll@{}}
\textnormal{Thread 1} & \textnormal{Thread 2}\\\hline
\Acquire{m} \\
\Write{x} \\
\Sync{o}\tikzmark{6} \\
\Write{y} \\
\Release{m}\tikzmark{7}\\
                      & acq(m) \\
                      & \tikzmark{8}wr(x) \\
                      & rel(m) \\
                      & \tikzmark{9}\Sync{o} \\
                      & \tikzmark{10}rd(x) \\
                      & \tikzmark{11}\BrUniv \\
                      & rd(y) \\
\\ \\ \\ \\ \\ \\ \\ \\ \\ \\
\end{tabular}}
\label{fig:wr-wr-conflict:extra}}
\hfill
\subfloat[Execution with \WBR-race but no \BR-race or predictable race]{
\centering
\sf
\makebox[.45\linewidth]{
\begin{tabular}{@{}ll@{}}
\textnormal{Thread 1} & \textnormal{Thread 2}\\\hline
\Acquire{k} \\
\Acquire{p} \\
\Acquire{m} \\
\Write{x} \\
\Release{m} \\
\Release{p}\tikzmark{20} \\
\Acquire{n} \\
\Release{n} \\
\Write{y} \\
\Release{k}\tikzmark{22} \\
                      & \Acquire{p} \\
                      & \tikzmark{21}\Write{x} \\
                      & \Release{p} \\
                      & \Acquire{k} \\
                      & \tikzmark{23}\Release{k} \\
                      & \Acquire{n} \\
                      & \Acquire{m} \\
                      & \Read{x} \\
                      & \tikzmark{30}\BrUniv \\
                      & \Release{m} \\
                      & \Read{y} \\
                      & \Release{n} \\
\end{tabular}}
\label{fig:wr-wr-conflict:counterexample}}
\textlink{3}{4}{\DC}%
\textunderlink{3}{5}{\BR}%
\textlink{7}{8}{\DC}%
\textlink{6}{9}{\WBR\hspace*{5em}}%
\textunderlink{7}{10}{\BR}%

\textlink{20}{21}{\WCP}
\textunderlink{22}{23}{\BR\hspace*{4em}}
\textunderlink{20}{30}{\hspace*{12em}\WBR}

\caption{Executions showing \WCP and \DC's overly strict handling of write--write dependencies.
\Sync{o} is an abbreviation for
the sequence \Acquire{o}\code{;} \Read{oVar}\code{;} \BrUniv{}\code{;} \Write{oVar}\code{;} \Release{o}.
\label{fig:wr-wr-conflict}}

\end{figure}

Consider the executions in Figures~\ref{fig:br-rd-wr-conflict:wbr-false-race} and \ref{fig:br-rd-wr-conflict:wbr-no-edge},
in which \WCP and \DC order \Read{x} to \Write{x}.
\WCP and \DC's read--write ordering assumes that no predictable trace exists where a pair of conflicting accesses are reordered.
This rationale works for Figure~\ref{fig:br-rd-wr-conflict:wbr-false-race},
in which no predictable race exists.
However, conflicting accesses may be reordered as long as the \LW rule of predictable traces is satisfied.
Figure~\ref{fig:br-rd-wr-conflict:wbr-no-edge-reordered} is a predictable trace of
Figure~\ref{fig:br-rd-wr-conflict:wbr-no-edge} that reorders the critical sections and exposes a race on \code{y}.


Similarly for write--write conflicts, consider Figure~\ref{fig:wr-wr-conflict:original}, in which \WCP and \DC order the two \Write{x} events,
leading to no \WCP- or \DC-race on accesses to \code{y}.
However, the \Write{x} events can be reordered,
as the predictable trace in Figure~\ref{fig:wr-wr-conflict:reordered} shows,
exposing a race.
(The reader can ignore Figure~\ref{fig:wr-wr-conflict:extra} and~\ref{fig:wr-wr-conflict:counterexample} until Section~\ref{sec:BR-WBR-relations}.)

It is difficult to model read--write and write--write dependencies more precisely using a partial order.
In the case of a read--write dependency, the accesses can be reordered \emph{as long as the read cannot impact a branch's outcome
in the predictable trace} (\ie, the read is not a causal event in the predictable trace, or is not part of the predictable trace).
For a write--write dependency, the accesses can be reordered
\emph{as long as they do not change a causal read's last writer in the predictable trace.}
Incorporating either kind of constraint into a partial order is challenging
but also desirable because partial orders can be computed efficiently.

\paragraph{Control dependence:}

\begin{figure}
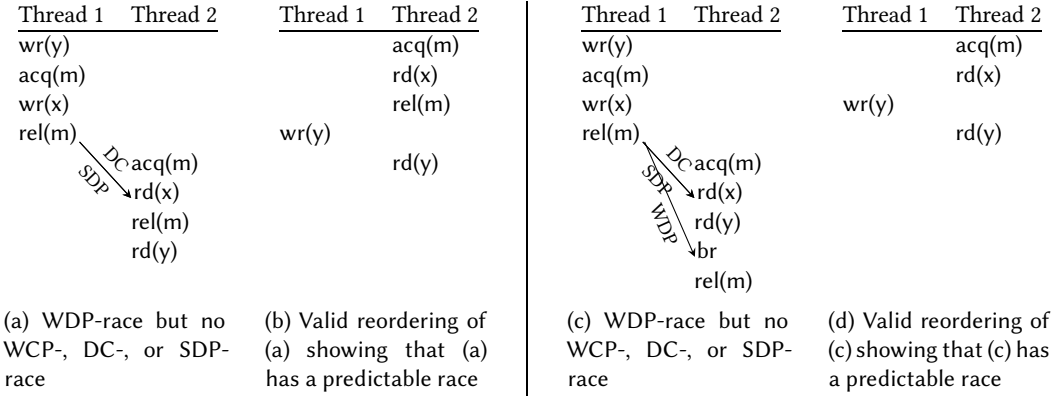

\captionsetup{farskip=0pt} 
\small
\centering

\subfloat[\WBR-race but no \WCP-, \DC-, or \BR-race]{
\centering
\sf
\makebox[.20\linewidth]{
\begin{tabular}{@{}ll@{}}
\textnormal{Thread 1} & \textnormal{Thread 2} \\\hline 
wr(y)                 & \\
acq(m)                & \\
wr(x)                 & \\
rel(m)\tikzmark{1}    \\
                      & acq(m) \\
                      & \tikzmark{2}rd(x) \\
                      & rel(m) \\
                      & rd(y) \\
\\
\end{tabular}}
\label{fig:simple-BR:original}}
\hfill
\subfloat[Valid reordering of (a) showing that (a) has a predictable race]{
\centering
\sf
\makebox[.20\linewidth]{
\begin{tabular}{@{}ll@{}}
\textnormal{Thread 1} & \textnormal{Thread 2} \\\hline
                      & acq(m) \\
                      & rd(x) \\
                      & rel(m) \\
wr(y)                 & \\
                      & rd(y) \\
				      & \\
				      & \\
\\ \\
\end{tabular}}
\label{fig:simple-BR:reordering}}
\hfill\vrule\hfill
\subfloat[\WBR-race but no \WCP-, \DC-, or \BR-race]{
\centering
\sf
\makebox[.20\linewidth]{
\begin{tabular}{@{}ll@{}}
\textnormal{Thread 1} & \textnormal{Thread 2} \\\hline 
wr(y)                 & \\
acq(m)                & \\
wr(x)                 & \\
rel(m)\tikzmark{5}\tikzmark{3}\\
                      & acq(m) \\
                      & \tikzmark{6}rd(x) \\
                      & \Read{y} \\
                      & \tikzmark{4}\BrUniv \\
                      & rel(m)
\end{tabular}}
\label{fig:simple-BR:branch:original}}
\hfill
\subfloat[Valid reordering of (c) showing that (c) has a predictable race]{
\centering
\sf
\makebox[.20\linewidth]{
\begin{tabular}{@{}ll@{}}
\textnormal{Thread 1} & \textnormal{Thread 2} \\\hline
                      & acq(m) \\
                      & rd(x) \\
wr(y)                 & \\
                      & rd(y) \\
		              & \\
				      & \\
				      & \\
                      & \\
\\
\end{tabular}}
\label{fig:simple-BR:branch:reordering}}

\textlink{1}{2}{\DC}%
\textunderlink{1}{2}{\BR}%
\textlink{5}{6}{\DC}%
\textunderlink{5}{6}{\BR}%
\textunderlink{3}{4}{\hspace*{2em}\WBR}%

\caption{Executions showing \WCP and \DC's overly strict handling of control dependencies.
\label{fig:simple-BR}}

\end{figure}

\WCP and \DC order true (write--read) dependencies
even when the read may not affect a branch outcome that affects whether a race happens.

Figure~\ref{fig:simple-BR:original} shows an execution with a predictable race,
as the predictable trace in Figure~\ref{fig:simple-BR:reordering} demonstrates.
Note that in Figure~\ref{fig:simple-BR:reordering}, \Read{x} has a different last writer than in Figure~\ref{fig:simple-BR:original},
but the lack of a following \BrUniv event means that \Read{y} is still guaranteed to happen (\ie, \Read{x} is not a causal event).
A variant of this example is in
Figure~\ref{fig:simple-BR:branch:original}, which has a branch event dependent on a read outcome,
but the branch can be absent from a predictable trace demonstrating a predictable race (Figure~\ref{fig:simple-BR:branch:reordering}).


\begin{figure}
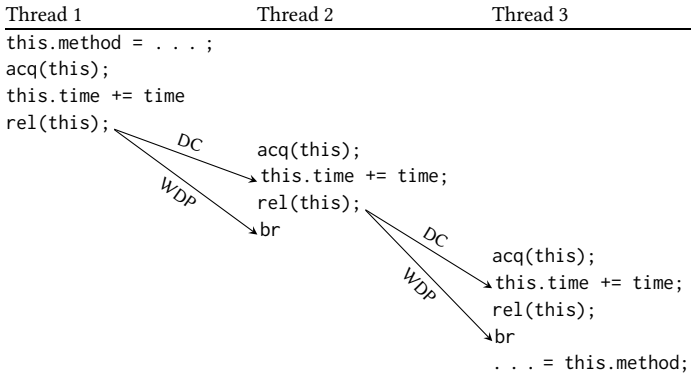






\footnotesize
\centering

\centering
\sf
{\setlength{\tabcolsep}{1em}
\begin{tabular}{@{}lll@{}}
  \textnormal{Thread 1}                         & \textnormal{Thread 2}              & \textnormal{Thread 3} \\\hline
  \texttt{this.method = \ldots;}                & \\
  \texttt{acq(this);}                           & \\
  \texttt{this.time += time} & \\
  \texttt{rel(this);}\tikzmark{1} & \\
                                                & \texttt{acq(this);}  \\
                                                & \tikzmark{2}\texttt{this.time += time;} \\
                                                & \texttt{rel(this);}\tikzmark{4} \\
                                                & \tikzmark{3}\texttt{br} \\
                                                &                                      & \texttt{acq(this);}  \\
                                                &                                      & \tikzmark{5}\texttt{this.time += time;} \\
                                                &                                      & \texttt{rel(this);} \\
                                                &                                      & \tikzmark{6}\texttt{br} \\
                                                &                                      & \texttt{\ldots = this.method;}  \\
\end{tabular}
}

\textlink{1}{2}{\DC}
\textlink{4}{5}{\DC}
\textunderlink{1}{3}{\WBR}
\textunderlink{4}{6}{\WBR}



\caption{A predictable race in the Java program \bench{pmd} that was detected by \WBR, but not \WCP,
  \WDC, or \BR. Note the transitive edges formed by \DC, which \WBR avoids as
  the branches are outside the critical sections. The code has been simplified
  and abbreviated.}
\label{fig:source-pmd}

\end{figure}

\WCP and \DC miss the predictable races in Figures~\ref{fig:simple-BR:original} and \ref{fig:simple-BR:branch:original}
by conservatively assuming that any event after a \Read{x} may be control dependent on the read value.
Similarly, \WCP and \DC miss the predictable race in Figure~\ref{fig:source-pmd},
which our implementation found in the Java program \bench{pmd} (Section~\ref{sec:evaluation}).
Essentially, \WCP and \DC conservatively assume that a dependent branch immediately follows each read.
This limitation is unsurprising considering the challenge of modeling control dependencies using a partial order.
In particular, it is difficult for a partial order to model
the fact that \emph{a read must have the same last writer only if the read may affect a branch in the predictable trace}.


\medskip
\noindent
This work develops partial orders that are weaker than \WCP and \DC and thus predict more races.
At the same time, these new partial orders retain key properties of the existing relations:
\WCP's soundness and \DC's amenability to a vindication algorithm that ensures soundness, respectively.


\section{Overview}

The previous section introduced prior work's
\wcpFull (\WCP)~\cite{wcp} and \dcFull (\DC)~\cite{vindicator},
and explained their limitations that lead to missing predictable races.
The next two sections introduce new relations and analyses that overcome these limitations.
Section~\ref{sec:BR-WBR-relations} introduces the \emph{\brFull (\BR)} and \emph{\wbrFull (\WBR)} relations,
which are weaker than \WCP and \DC, respectively.
Section~\ref{sec:wbr-analysis} presents online dynamic analyses for computing \BR and \WBR and detecting \BR- and \WBR-races.

\section{New Dependence-Aware Predictive Relations}
\label{sec:BR-WBR-relations}

This section introduces new partial orders called
\emph{\brFull (\BR)} and
\emph{\wbrFull (\WBR)} that overcome the limitations of
prior work's predictive relations~\cite{vindicator,wcp} (Section~\ref{sec:background:limitations})
by incorporating more precise notions of data and control dependence.

\subsection{The \BR and \WBR Partial Orders}

\BR is weaker than \WCP\footnote{It may seem confusing that \underline{S}DP is \emph{weaker} than \underline{W}CP.
SDP is so named because it is stronger than \emph{WDP}, while WCP is so named because it is weaker than prior work's \emph{causally-precedes (CP)}~\cite{causally-precedes}.}
by not ordering write--write conflicts, based on
the insight that writes can be unordered unless they can affect the outcome of a read,
but write--read and read--write ordering already handles that ordering soundly.
\WBR only orders the last writer of a read to a branch that depends on that read,
which is the only reordering constraint that does not lead to missing predictable races.

\begin{table}
\newcommand\yes{Yes} 
\newcommand\no{No}
\newcommand{\ltSO}{\ensuremath{\prec_\textsc{\tiny{SO}}}\xspace}
\small
\begin{tabular}{@{}l|lll|ll@{}}
 & \mc{3}{c|}{Prior work}               & \mc{2}{c}{This paper} \\
Property
 & \ltHB              & \ltWCP & \ltWDC & \ltBR & \ltWBR \\\hline
Same-lock critical section ordering
 & All                & Confl.    & Confl.     & Confl. & Last \code{wr}--\code{rd} only \\
Orders \code{rel} to\dots
 & \code{acq}         & \code{wr}/\code{rd} & \code{wr}/\code{rd} & \code{wr}/\code{rd} or next \code{rd}* & Next \code{br} \\
 \hline
 Includes \ltPO? 
 & \yes               & \no    & \yes   & \no & \yes \\
Left-and-right composes with \ltHB?
 & \yes               & \yes   & \no    & \yes & \no \\
 \hline
$\Acquire{m}\prec_*\Release{m}$ implies $\Release{m}\prec_*\Release{m}$?
  & \yes               & \yes   & \yes   & \yes** & \yes \\
  Transitive?
  & \yes               & \yes   & \yes   & \yes & \yes \\
 \end{tabular}

\caption{Definitions of five strict partial orders over events in an execution trace.
Each order is the minimum relation satisfying the listed properties.
This table adds columns \ltBR and \ltWBR to Table~\ref{tab:partial-order-definitions-prior} (page~\pageref{tab:partial-order-definitions-prior}).\\
*\,As the text explains, \BR adds release--access ordering for write--read and read--write conflicts,
and adds ordering from the release to the next read for write--write conflicts.\\
**\,As the text explains, \BR adds release--release ordering for two
critical sections on the same lock if the critical sections
are \HB ordered by a write--write conflict but do not both contain the conflicting writes.}
\label{tab:partial-order-definitions}

\end{table}

Table~\ref{tab:partial-order-definitions} defines \BR and \WBR.
The table shows that \BR is like \WCP and \WBR is like \DC
except in how they order conflicting critical sections (first two table rows).
Furthermore, \BR adds release--release ordering
under certain conditions because of write--write conflicts (next-to-last row).


\paragraph{The \BR relation.}

\BR only orders conflicting critical sections
when one critical section contains a read. Like \WCP, \BR orders the first critical section's \Release{m}
to the second critical section's access.
The intuition behind this property is that write--write conflicts generally do not impose any limitations on what traces can be predicted.
Figure~\ref{fig:wr-wr-conflict:original} shows an example in which two conflicting writes can be safely reordered in a predictable trace.

However, ignoring write--write conflicts altogether would be unsound,
as Figure~\ref{fig:wr-wr-conflict:extra} shows: the execution has no predictable race.
To ensure soundness, \BR handles write--write conflicts by ordering the first critical section to the second thread's next \emph{read} to the same variable. Figure~\ref{fig:wr-wr-conflict:counterexample} shows a similar case where the execution has no predictable race. To handle this case as well, \BR also orders critical sections that are ordered by, but do not contain, write--write conflicts.

More formally, \BR is the minimal relation that satisfies the following rules:
\begin{itemize}

\item {[Unique to \BR{}]}
\BROrdered{\event{r}{_1}}{\event{e}{_2}} if
\event{r}{_1} and \event{r}{_2} are release events on the same lock,
\TROrdered{\event{r}{_1}}{\event{r}{_2}},
$\event{e}{_1} \in \CS{\event{r}{_1}}$, $\event{e}{_2} \in \CS{\event{r}{_2}}$,
\conflicts{\event{e}{_1}}{\event{e}{_2}}, and
\event{e}{_1} or \event{e}{_2} is a read.

\item {[Unique to \BR{}]}
\BROrdered{\event{r}{_1}}{\event{e}{_3}} if
 \event{r}{_1} and \event{r}{_2} are release events on the same lock,
 \TROrdered{\event{r}{_1}}{\event{r}{_2}},
 \event{e}{_1} and \event{e}{_2} are write events and \event{e}{_3} is a read event,
 \conflicts{\event{e}{_1}}{\event{e}{_2}},
 \conflicts{\event{e}{_1}}{\event{e}{_3}},
 \POOrdered{\event{e}{_2}}{\event{e}{_3}},
 \inCS{\event{e}{_1}}{\event{r}{_1}}, and
 \inCS{\event{e}{_2}}{\event{r}{_2}}.

\item {[Borrowed from \WCP{}]}
\BROrdered{\event{r}{_1}}{\event{r}{_2}} if
\event{r}{_1} and \event{r}{_2} are release events on the same lock and
\BROrdered{\getAcquire{\event{r}{_1}}}{\event{r}{_2}}.

\item {[Unique to \BR{}]} $\event{r}{_1}\ltBR\event{r}{_2}$ if
\conflicts{\event{e}{_1}}{\event{e}{_2}} are write events,
\event{r}{_1} and \event{r}{_2} are release events on the same lock,
\event{s}{_1} and \event{s}{_2} are release events on the same lock,
$\getAcquire{\event{r}{_1}} \ltHB
\event{s}{_1} \ltTR \event{e}{_2} \ltHB \event{r}{_2}$,
$\event{e}{_1} \in\CS{\event{s}{_1}}$, $\event{e}{_2} \in\CS{\event{s}{_2}}$,
and $\event{e}{_1} \not\in\CS{\event{r}{_1}} \lor \event{e}{_2} \not\in\CS{\event{r}{_2}}$.
\notes{In other words, \BR orders two release events if there is a
write--write conflict HB-ordering parts of the critical sections, but not if both writes are within these critical sections.
\mike{Redundant with above informal explanation.}}

\item {[Borrowed from \WCP{}]}
\BROrdered{\event{e}{_1}}{\event{e}{_2}}
if \event{e}{_3} is an event and $\BROrdered{\event{e}{_1}}{\HBOrdered{\event{e}{_3}}{\event{e}{_2}}}
\lor \HBOrdered{\event{e}{_1}}{\BROrdered{\event{e}{_3}}{\event{e}{_2}}}$.

\label{br-rel-rel-rules}
\end{itemize}

\notes{
\mike{Assuming the execution is \BR-race free up to this point, either
(i) \event{e}{_1} and \event{e}{_2} are in critical sections on the same lock ending in
\event{s}{_1} and \event{s}{_2}
(where $\event{s}{_1} \ne \event{r}{_1} \lor \event{s}{_2} \ne \event{r}{_2}$), or
(2) \BROrdered{\event{e}{_1}}{\event{e}{_2}} (which has no effect).
\mike{What if
$\getAcquire{\event{r}{_1}} \not\ltHB \event{e}{_1} \ltHB \event{e}{_2} \ltHB \event{r}{_2}$
but
$\getAcquire{\event{r}{_1}} \ltHB \event{s}{_1} \ltHB \event{e}{_2} \ltHB
\event{r}{_2}$?
\kaan{We should adjust the definition to cover this, no changes to the algorithm
are required.}
\mike{Done}}
\mike{Another, unrelated question: In case (2), can \event{r}{_1} and \event{r}{_2} and
\event{s}{_1} and \event{s}{_2} all be on the same lock?}}
\kaan{What about the following example: Consider release events \event{l},
\event{l}{'}, \event{m}, \event{m}{'}; and write events \event{w}, \event{w}{'}.
$\getAcquire{\event{m}} \ltPO \event{w} \ltPO \event{l} \ltPO \event{m}$,
$\getAcquire{\event{l}{'}} \ltPO \event{w}{'} \ltPO \getAcquire{\event{m}{'}}
\ltPO \event{m}{'} \ltPO \event{l}{'}$, and $\event{l} \ltHB
\getAcquire{\event{l}{'}}$. \event{l}, \event{l}{'} and \event{m}, \event{m}{'}
release the same lock, and \conflicts{\event{w}}{\event{w}{'}}. We don't draw a
rule-b edge here with the current definition, but this might be a case we can't prove.
\mike{Assuming critical sections are properly nested, then
$\getAcquire{\event{m}{}} \ltPO \getAcquire{\event{l}{}} \ltPO \event{w}$
or
$\event{w} \ltPO \getAcquire{\event{l}{}} \ltPO \event{l}{}$.
In the first case, I think the current definition already provides
$\event{m}{} \ltBR \event{m'}{}$, right? Or did you have something else in mind?
In the second case, I think there's a \BR-race?
\medskip\\
However, in any case, I think the execution has a predictable \textbf{deadlock}.}}
}


\notes{Like WCP and other partial orders, \BR orders $\event{r}{_1}\ltBR\event{r}{_2}$ if
$\getAcquire{\event{r}{_1}}\ltBR\event{r}{_2}$. That is,
\BR orders two release events if the corresponding acquire of the
earlier release is \BR ordered to the later release.}

In essence, \BR addresses a limitation of \WCP via more precise handling of data dependencies.
\BR certainly does not address all imprecise data dependencies (\eg, read--write dependencies),
and it does not address control dependence.
\BR is the weakest known sound partial order.

\paragraph{The \WBR relation}

A separate but worthwhile goal is to develop a partial order that is weaker than \DC
but produces few false positives so that it is practical to vindicate potential races.
\WBR achieves this goal and is in fact complete, detecting all predictable races.
\WBR orders the last writer of each read to the earliest branch that depends on that read (and orders no other conflicting critical sections).
The intuition behind this behavior is that the only constraint that is universally true
for all predictable traces is that the last writer of a read must not occur
after the read if there is a branch that depends on the read.

More formally,
if \event{r}{_1} and \event{r}{_2} are releases on the same lock,
$\event{e}{_1} \in \CS{\event{r}{_1}}$,
$\event{e}{_2} \in \CS{\event{r}{_2}}$,
$\event{e}{_1} = \lastwr{\event{e}{_2}}{tr}$,
\POOrdered{\event{e}{_2}}{\event{b}{}},
and \BrDepsOn{\event{b}{}}{\event{e}{_2}}, then \WBROrdered{\event{r}{_1}}{\event{b}{}}.

%
%
%
%

\later{
\mike{Something to consider regarding \WBR being incomplete if it composes with write--read edges:
we could keep the composition and instead weaken \WBR by
\begin{itemize}
  \item changing rule (a) so it adds wr--br edges instead of rel--br edges and
  \item eliminating rule (b).
\end{itemize}
The result might be complete and might actually be less unsound in practice than the complete version of \WBR we're currently testing.}
\jake{Do you have an example in where change rule (a) would result in a true race
that \WBR currently misses?
\mike{Yeah, (b) in example/WBR-incomplete.tex.}}
}

Unlike \DC, \WBR
integrates control dependence by
ordering the write's critical section to the first branch dependent on the read.
\WBR does not model \emph{local} data dependencies, where a read affects the value written by a write in the same thread.
As a result, \WBR may find false races, but Section~\ref{sec:vindication-summary} describes a method for ruling out such false races.
These properties make \WBR complete (as we show).
\WBR is the strongest known complete partial order.

\paragraph{\BR- and \WBR-races.}
\label{sec:br-wbr-races}

Unlike \WCP and \DC, \BR and \WBR do not inherently order all conflicting accesses that hold a common lock.
Thus the following definition of \BR- and \WBR-races explicitly excludes conflicting accesses holding a common lock.

  A trace has a \emph{\BR-race} (or \emph{\WBR-race}) if
  it has two conflicting events unordered by \BR (\WBR) that hold no lock in common.
  That is, \tr has an \BR-race (\WBR-race) on events \event{e}{} and \event{e}{'} if
  \TROrdered{\event{e}{}}{\event{e}{'}},
  \conflicts{\event{e}{}}{\event{e}{'}},
  \nBROrdered{\event{e}{}}{\event{e}{'}} (\nWBROrdered{\event{e}{}}{\event{e}{'}}), and
  $\commonLocks{\event{e}{}}{\event{e}{'}} = \emptyset$.


\notes{
Recall that \lockset{\event{e}{}} is the set of locks held at \event{e}{} by the thread executing \event{e}{}
(Section~\ref{subsec:execution-trace}).
}

\paragraph{Examples.}

To illustrate \BR and \WBR, we refer back to examples from pages~\pageref{fig:br-rd-wr-conflict}--\pageref{fig:source-pmd}.

The executions in Figures~\ref{fig:br-rd-wr-conflict:wbr-false-race} and \ref{fig:br-rd-wr-conflict:wbr-no-edge}
have no \BR-races: \BR orders the read--write conflicts.
In contrast, these executions have \WBR-races: there is no cross-thread \WBR ordering because the executions have no lock-protected write--read conflicts.


Figure~\ref{fig:wr-wr-conflict:original} has \BR- and \WBR-races. \BR and \WBR do not order
the write--write conflict on \code{x}. Nor does \WBR order events on
the write--read conflict on \code{x},
since Thread~1's \Write{x} is not the last writer of \Read{x}.

Figure~\ref{fig:wr-wr-conflict:extra},
which has a \WBR-race but no \BR-race or predictable race,
shows the need for \BR's release--read ordering for write--write conflicts.
Figure~\ref{fig:wr-wr-conflict:counterexample} also has a \WBR-race but no \BR-race or predictable race, showing the need for \BR's release--release
ordering for write--write conflicts.

The executions in Figures~\ref{fig:simple-BR:original} and \ref{fig:simple-BR:branch:original} have no \BR-race
since \BR does not take branches into account.
On the other hand, both executions have \WBR-races:
Figure~\ref{fig:simple-BR:original} has no branch dependent on the read,
and Figure~\ref{fig:simple-BR:branch:original} has a branch, but it occurs after \Read{y}.

\WBR analysis discovers the predictable race in Figure~\ref{fig:source-pmd}.
In this case, the fact that there is no branch within the critical section allows \WBR to avoid creating an unnecessary transitive edge that otherwise would hide the race.

\begin{table}
\vspace*{-1em}
\captionsetup{type=figure} 
\small
\hfill
\subfloat[Weaker above stronger relations]{
\makebox[5cm][c]{
\begin{tikzpicture}
    \node (HB) at (0,0) {\ltHB};
    \node (WCP) at (0,0.8) {$(\ltWCP \cup \ltPO)$};
    \node (DC) at (-1, 1.6) {\ltWDC};
    \node (BR) at (1, 1.6) {$(\ltBR \cup \ltPO)$};
    \node (WBR) at (0, 2.4) {\ltWBR};

    \draw [thick] (HB) -- (WCP);
    \draw [thick] (WCP) -- (DC);
    \draw [thick] (WCP) -- (BR);
    \draw [thick] (DC) -- (WBR);
    \draw [thick] (BR) -- (WBR);
\end{tikzpicture}
\label{fig:lattices:relations}
}
}
\hfill
\subfloat[Supersets above subsets]{
\begin{tikzpicture}
    \node (HB) at (0,0) {\HB-races};
    \node (WCP) at (0, .8) {\WCP-races};
    \node (DC) at (-1, 1.6) {\DC-races};
    \node (BR) at (1, 1.6) {\BR-races};
    \node (WBR) at (0, 2.4) {\WBR-races};

    \draw [thick] (HB) -- (WCP);
    \draw [thick] (WCP) -- (DC);
    \draw [thick] (WCP) -- (BR);
    \draw [thick] (DC) -- (WBR);
    \draw [thick] (BR) -- (WBR);
\end{tikzpicture}
\label{fig:lattices:races}
}
\hspace*{\fill}

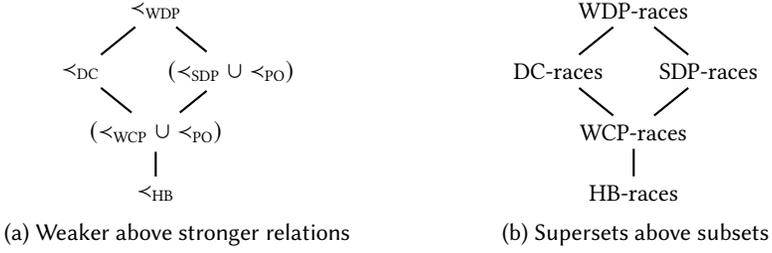
\caption{Lattices showing the relationships among the relations and corresponding kinds of races.
Only \WCP and \BR do not include \PO and thus do not in general order events within the same thread,
but this property is irrelevant for comparing relation strength because same-thread accesses cannot race in any case,
so the relation lattice uses $(\ltWCP \cup \ltPO)$ and $(\ltBR \cup \ltPO)$ to make the relations more directly comparable.
(\WCP and \BR analyses in fact detect races by comparing access events using $(\ltWCP \cup \ltPO)$ and $(\ltBR \cup \ltPO)$, respectively.)}
\label{fig:lattices}
\bigskip
\captionsetup{type=table} 
\small
\begin{tabular}{@{}l|ll@{}}
     & Sound? & Complete? \\\hline
\HB  & Yes~\cite{happens-before} & No (Fig.~\ref{fig:br-rd-wr-conflict:wbr-no-edge}, \ref{fig:wr-wr-conflict:original}, \ref{fig:simple-BR:original}, \ref{fig:simple-BR:branch:original}) \\
\WCP & Yes~\cite{wcp} & No (Fig.~\ref{fig:br-rd-wr-conflict:wbr-no-edge}, \ref{fig:wr-wr-conflict:original}, \ref{fig:simple-BR:original}, \ref{fig:simple-BR:branch:original}) \\
\DC  & No~\cite{vindicator} & No (Fig.~\ref{fig:br-rd-wr-conflict:wbr-no-edge}, \ref{fig:wr-wr-conflict:original}, \ref{fig:simple-BR:original}, \ref{fig:simple-BR:branch:original}) \\\hline
\BR  & Yes (Section~\ref{sec:br-soundness-proof}) & No (Fig.~\ref{fig:br-rd-wr-conflict:wbr-no-edge}, \ref{fig:simple-BR:original}, \ref{fig:simple-BR:branch:original}) \\
\WBR & No (Fig.~\ref{fig:br-rd-wr-conflict:wbr-false-race} and \ref{fig:wr-wr-conflict:extra}) & Yes (Section~\ref{sec:wbr-completeness}) \\
\end{tabular}
\caption{Soundness and completeness of each relation.}
\label{tab:soundness-completeness}

\end{table}

\subsection{Soundness and Completeness}

Figure~\ref{fig:lattices} and Table~\ref{tab:soundness-completeness} illustrate the relationships among the different relations
and corresponding race types.
\BR never misses a race that \WCP finds, and \WBR never misses a race that \DC or \BR finds.
\BR is sound but incomplete (never reports a false race but may miss predictable races),
while \WBR is unsound but complete (may report false races but never misses a predictable race).

\later{
\mike{This reminds me that I spoke with William Mansky (UIC) at FCRC,
and he was interested in the possibility of machine-verifying correctness properties of our relations/analyses.}
}

Here we prove that \BR is sound and \WBR is complete.
The proofs are manual and have not been verified by a theorem prover.

\label{sec:br-soundness-proof}

\begin{theorem*}[\BR soundness]
  If an execution trace has a \BR-race,
  then it has a predictable race or a predictable deadlock.
\end{theorem*}

\begin{proof}

\newcommand{\BRn}[1]{\BR{}(\ensuremath{#1})}
\newcommand{\ltBRn}[1]{\ensuremath{\prec_{\textsc{\tiny{\BR}}(#1)}}}
\newcommand{\nltBRn}[1]{\ensuremath{\not\prec_{\textsc{\tiny{\BR}}(#1)}}}
\newcommand{\BRnOrdered}[3]{\Ordered{#1}{\ltBRn{#3}}{#2}}
\newcommand{\nBRnOrdered}[3]{\Ordered{#1}{\nltBRn{#3}}{#2}}

We define \ltBRn{i} to be a variant of \ltWCP and \ltBR that handles conflicting writes like \BR for the first $i$ conflicting writes in \tr,
and orders conflicting writes like \WCP otherwise (Table~\ref{tab:partial-order-definitions}).
Formally,
\ltBRn{i} is the minimal relation that has all of the properties that \ltBR has
(Table~\ref{tab:partial-order-definitions};
\eg, $\event{r}{_1} \ltBRn{i} \event{r}{_2}$
if \event{r}{_1} and \event{r}{_2} are release events on the same lock such that
$\getAcquire{\event{r}{_1}} \ltBRn{i} \event{r}{_2}$)
and also the following property:
\begin{itemize}
\item[] \BRnOrdered{\event{l}{}}{\event{w}{'}}{i} if
\event{w}{} and \event{w'} are conflicting write events,
\event{l}{} and \event{l}{'} are release events of critical sections on the same lock
such that $\event{w}{} \in \CS{\event{l}{}}$ and $\event{w}{'} \in \CS{\event{l}{'}}$,
and there are at least $i$ conflicting pairs of write events
(\ie, two conflicting write events without an intervening conflicting write event)
before\footnote{An event pair $(e,e')$ is \emph{before} an event pair $(w,w')$ if $e' <_\tr w' \lor (e' = w \land e <_\tr w)$.} the write pair $(w,w')$.
\end{itemize}
Note that $\ltBRn{0} \equiv \ltWCP$
(because \BRn{0} adds back the write--write conflict rule that \BR omits from WCP)
and $\ltBRn{\infty} \equiv \ltBR$.


The rest of the proof proceeds by induction to show that \BRn{i} is sound for all $i$,
\ie, if an execution trace has an \BRn{i}-race, then it has a predictable race or deadlock.

\paragraph{Base case:}

Since $\ltBRn{0} \equiv \ltWCP$ and \WCP is sound~\cite{wcp}, \BRn{0} is sound.

\paragraph{Inductive step:}

Let $\sigma$ be an execution trace whose first \BRn{i}-race is between events \event{e}{_1} and \event{e}{_2},
where \emph{first} means that \event{e}{_2} is as early as possible in $\sigma$, and among \BRn{i}-races whose second event is \event{e}{_2},
\event{e}{_1} is as late as possible in $\sigma$.
Proceeding with proof by contradiction,
suppose $\sigma$ has no predictable race or deadlock.

\newcommand{\SigmaOrdered}[2]{\Ordered{#1}{<_\sigma}{#2}}

Now let \tr be a trace equivalent to $\sigma$ that moves
all events between \event{e}{_1} and \event{e}{_2} that are not \HB ordered with both events,
to outside of \event{e}{_1} and \event{e}{_2},
and additionally removes all events after \event{e}{_2}.
Specifically:
\begin{itemize}
\item if $\SigmaOrdered{\event{e}{_1}}{\event{e}{}} \land \nHBOrdered{\event{e}{_1}}{\HBOrdered{\event{e}{}}{\event{e}{_2}}}$,
move \event{e}{} before \event{e}{_1} in \tr;
\item if $\SigmaOrdered{\event{e}{_1}}{\SigmaOrdered{\event{e}{}}}{\event{e}{_2}} \land {\nHBOrdered{\event{e}{}}{\event{e}{_2}}}$, omit \event{e}{} from \tr;
\item if \SigmaOrdered{\event{e}{_2}}{\event{e}{}}, omit \event{e}{} from \tr.
\end{itemize}
Thus the last event in \tr is \event{e}{_2}.
Like $\sigma$, \tr's first \BRn{i}-race is between events \event{e}{_1} and \event{e}{_2},
and \tr has no predictable race or deadlock.
That is,
\TROrdered{\event{e}{_1}}{\event{e}{_2}},
\conflicts{\event{e}{_1}}{\event{e}{_2}},
\nBROrdered{\event{e}{_1}}{\event{e}{_2}} in \tr,
and \event{e}{_1} and \event{e}{_2} are not in critical sections on the same lock.

Because \tr has no predictable race or deadlock,
\BRnOrdered{\event{e}{_1}}{\event{e}{_2}}{i-1} by the induction hypothesis.
Because of the disparity between \BRnOrdered{\event{e}{_1}}{\event{e}{_2}}{i-1} and
\nBRnOrdered{\event{e}{_1}}{\event{e}{_2}}{i},
it must be that \BRnOrdered{\event{l}{}}{\event{w}{'}}{i-1} and
\nBRnOrdered{\event{l}{}}{\event{w}{'}}{i},
where \event{w}{'} is the $i$th conflicting write in \tr,
\event{w}{} is the latest write before \event{w}{'} such that
$\conflicts{\event{w}{}}{\event{w}{'}}$, and
\event{l}{} is the outermost release event of a critical section containing
\event{w}{} such that a release event \event{l}{'} on the same lock that ends a critical
section containing \event{w}{'} also exists.



\notes{
%
%
Furthermore, there must exist some release \event{m} such that $\event{e}{_1}
\ltHB \event{m} \ltHB \event{e}{_2}$ and $\getAcquire{\event{m}} \ltHB
\event{l}$. We can prove the existence of this event as follows: Because
$\event{e}{_1}, \event{e}{_2}$ is not a race (assumed above), \HB is sound, and
\HB only orders acquire and release events, there must exist at least some
release event \event{x} and some acquire event \event{y} such that
$\event{e}{_1} \ltHB \event{x} \ltHB \event{y} \ltHB \event{e}{_2}$. Next,
because of the difference between $\event{e}{_1} \ltBRn{i-1} \event{e}{_2}$ and
$\event{e}{_1} \nltBRn{i} \event{e}{_2}$, there must exist some \event{x},
\event{y} such that $\event{e}{_1} \ltHB \event{x} \ltBRn{i-1} \event{y} \ltHB
\event{e}{_2}$ and $\event{x} \nltBRn{i} \event{y}$. 

\mike{How do we know that's true? (To be clear: I'm not convinced it's true.
Even if it's true, we need to justify it.) \kaan{Tried to justify this but got
stuck. $\event{e}{_1} \ltHB \event{m} \ltHB \event{e}{_2}$ part is easy, but I
am not sure about $\getAcquire{\event{m}} \ltHB \event{l}$ part (Starting with
"Next") since it is very disconnected from everything else. Also, I think this
could be made stricter with $\event{m} \ltHB \event{l}$ but avoided doing so
before we know it is true as it is currently.}
\mike{Intuitively I don't see why such an \event{m}{} must exist.
In particular, if
\BRnOrdered{\event{l}{}}{\event{w}{'}}{i} ``causes''
\BRnOrdered{\event{e}{_1}}{\event{e}{_2}}{i}
through a series of rule~(b) applications.}}

Note that
\event{l} and \event{m} may be the same event, and \event{m} and \event{l} are
not necessarily on the same lock. This implies either
$\inCS{\event{e}{_1}}{\event{m}}$, or $\event{e}{_1} \ltHB
\getAcquire{\event{m}} \ltPO \event{m}$, or $\getAcquire{\event{m}} \ltHB
\event{e}{_1} \ltPO \event{m}$.

Based on this, let \event{m}{'} be a release event on the same lock as \event{m}
such that $\event{m}{'} \ltHB \event{e}{_2}$. Either $\event{l} = \event{m}
\land \event{l}{'} = \event{m}{'}$, or no \event{m}{'} exists such that
$\event{l}{'} \ltHB \event{m}{'}$.

\kaan{We guarantee here that there can be no release-release
ordering caused by a write-write conflict that orders $\event{e}{_1} \ltBRn{i-1}
\event{e}{_2}$, unless it is an "immediate conflict", in which case our old
cases may be sufficient to handle it. Making this statement may require us to
fix the counterexample I proposed below the new rule, see "What about the
following example:..." above.}

}%

Note that \BRnOrdered{\event{l}{}}{\event{w}{'}}{i-1}
``causes''
\BRnOrdered{\event{e}{_1}}{\event{e}{_2}}{i-1}: the only difference between
$\ltBRn{i-1}$ and $\ltBRn{i}$ is the order between $\event{l}$ and $\event{w}{'}$,
along with transitive and release--release orders implied by it. 
Thus, by the definition of \ltBRn{i-i}, either
\begin{enumerate}
\item
\HBOrdered{\event{e}{_1}}{\BRnOrdered{\event{l}{}}{\HBOrdered{\event{w}{'}}{\event{e}{_2}}}{i-1}};
or
\item 
\HBOrdered{\event{e}{_\mathit{acq}}}{\BRnOrdered{\event{l}{}}{\HBOrdered{\event{w}{'}}{\event{e}{_\mathit{rel}}}}{i-1}},
where \event{e}{_\mathit{acq}} and \event{e}{_\mathit{rel}} are the start and end events,
respectively, of two critical sections on the same lock.
One or more applications of the acq--rel$\Rightarrow$rel--rel property
(Table~\ref{tab:partial-order-definitions}) on
\BRnOrdered{\event{e}_\mathit{acq}}{\event{e}{_\mathit{rel}}}{i-1}
results in
\BRnOrdered{\event{e}{_1}}{\event{e}{_2}}{i-1}.
\end{enumerate}
In the second case,
\BRnOrdered{\event{e}{_\mathit{acq}}}{\event{e}{_\mathit{rel}}}{i}
by the ``critical sections ordered by a write--write conflict'' property
of \BR (Table~\ref{tab:partial-order-definitions}),
and thus \BRnOrdered{\event{e}{_1}}{\event{e}{_2}}{i},
a contradiction.
\notes{
\mike{OR: Instead of using the new rule of \BR,
which is pretty strong and kind of tricky to compute,
can we prove the second case above using cases similar to the cases below?
We have to be careful because although \TROrdered{\event{r}{}}{\event{e}{_2}},
it's not necessarily true that \TROrdered{\event{r}{}}{\event{e}{_\mathit{rel}}}.
\mike{Update after our 22 December meeting:
I think a rule that combines rule~(b) and write--write--read conflicts
could work, but I don't see how it would be computable with a linear-time algorithm.
\medskip\\
But I also take back some of my concern about the new \BR rule
(\ie, rule~(b) for write--write conflicts)
likely being too strong in practice. I really don't know,
and it's worth trying it out, especially if this proof actually works.}}
}%
Thus the first case must hold.

Suppose there is a read event $r$ that reads the same variable as \event{w}{} and \event{w}{'}
such that $\event{w}{'} \ltTR \event{r} \ltTR \event{e}{_2}$.
It must be that $\event{w}{'} \ltHB \event{r} \ltHB \event{e}{_2}$: if
$\event{r} \nltHB \event{e}{_2}$ then $\event{r} \not\in \tr$, and if $\event{w}{'}
\nltHB \event{r}$ then (\event{w}{'}, \event{r}) are a race earlier than
(\event{e}{_1}, \event{e}{_2}), both contradictions.
\notes{
Then either
\mike{I think two of the following cases now work.
The other two involve the read being
between the two writes, so the writes are ordered by \BRn{i}.
To complete the proof, I think it would be sufficient to make a rule
that if two writes $w,w'$ are ordered by \BR, then \BROrdered{\event{l}{}}{\event{w}{'}},
where $\event{w}{} \in \CS{\event{l}{}}$ and \event{w}{'} is in a critical section
on the same lock as \event{l}{}.
Or make the rule slightly more strict, by requiring that there be a read
to the same variable between the two writes.}
\kaan{Question: If we define \event{r} such that $\event{e}{_1} \ltHB
\event{r}$, and proceed with the proof if such an \event{r} does not exist, does
the final part of the proof still work? It looks like it should to me.
\mike{I don't see how the final part of the proof would still work.}}
\begin{enumerate}[label=(\alph*)]
  \item $\conflicts{\event{w}{'}}{\event{r}{}} \land \TROrdered{\event{r}{}}{\event{w}{'}}$, in which case
  \HBOrdered{\event{e}{_1}}{\BRnOrdered{\event{r}{}}{\HBOrdered{\event{w}{'}}{\event{e}{_2}}}{i}};
  \mike{How do we know that
  \TROrdered{\event{e}{_1}}{\event{r}{}}?
  \kaan{Defined \event{r} to have that.}}
  
  \item $\conflicts{\event{w}{'}}{\event{r}{}} \land \TROrdered{\event{w}{'}}{\event{r}{}}$, in which case
  \HBOrdered{\event{e}{_1}}{\BRnOrdered{\event{w}{'}}{\HBOrdered{\event{r}{}}{\event{e}{_2}}}{i}};

  \item $\conflicts{\event{w}{}}{\event{r}{}} \land
  \POOrdered{\event{w}{'}}{\event{r}{}}$, in which case
  \HBOrdered{\event{e}{_1}}{\BRnOrdered{\event{l}{}}{\HBOrdered{\event{r}{}}{\event{e}{_2}}}{i}}; or
  \mike{I think this works now (since we've now established above that
  \HBOrdered{\event{e}{_1}}{\event{l}{}}).}
  \item $\conflicts{\event{w}}{\event{r}} \land \POOrdered{\event{r}}{\event{w}{'}}$, in which case
  $\event{e}{_1} \ltHB \event{l} \ltBRn{i} \event{l}{'}$.
  \textcolor{red}{\HBOrdered{\event{e}{_1}}{\BRnOrdered{\event{w}}{\HBOrdered{\event{r}}{\event{e}{_2}}}{i}}}.
  \mike{I think this conclusion was also incorrect for old \BR,
  because we don't know that \TROrdered{\event{e}{_1}}{\event{w}{}}.
  As for the previous case, I believe we need to be able to conclude that
  \HBOrdered{\event{e}{_1}}{\HBOrdered{\event{q}{}}{\HBOrdered{\event{w}{}}
  {\HBOrdered{\event{w}{'}}{\HBOrdered{\event{q}{'}}{\event{e}{_2}}}}}}
  and therefore \BRnOrdered{\event{e}{_1}}{\event{e}{_2}}{i},
  where \event{q}{} and \event{q}{'} are release events on the same lock.
  Can we conclude that?
  \medskip\\
  Or: Is this case actually easier because
  \BRnOrdered{\event{w}{}}{\event{r}{}}{i} and therefore
  \BRnOrdered{\event{w}{}}{\event{w}{'}}{i}?
  \kaan{Hmm, now that I look at both cases, the solution I wrote does not even
  use \event{r} which feels wrong. I am likely missing some cases that we need
  to consider.}}
\end{enumerate}
In each of these cases, by \BRn{i}'s composition with \HB,
\BRnOrdered{\event{e}{_1}}{\event{e}{_2}}{i}, a contradiction.

Therefore there is no read event $r$ that reads the same variable \event{w}{}
and \event{w}{'} such that \TROrdered{\event{w}{}}{\event{r}{}}.
Any read event that reads the same variable as \event{w}{} and \event{w}{'} must
occur \emph{before} \event{w}{}.}
Events \event{w}{'} and \event{r} are executed either by the same or different threads:


\begin{enumerate}[label=(\alph*)]
  \item If $\conflicts{\event{w}{'}}{\event{r}}$,
  then $\event{e}{_1} \ltHB \event{w}{'} \ltBRn{i}
  \event{r} \ltHB \event{e}{_2}$.

  \item If $\event{w}{'} \ltPO \event{r}$,
  then \conflicts{\event{w}}{\event{r}} and thus
  $\event{e}{_1} \ltHB \event{l} \ltBRn{i} \event{r} \ltHB \event{e}{_2}$
  by the ``conflicting writes imply release--read ordering'' property of \BR.
\end{enumerate}







\noindent
In all cases, $\event{e}{_1} \ltBRn{i} \event{e}{_2}$ which is a
contradiction. As a result, no such \event{r} exists.

Now consider the trace \trPrime that is equivalent to \tr except that

\begin{itemize}
  \item \event{w}{'} is replaced by a \Write{x} event, where \code{x} is a brand-new variable not used in \tr.
  \item For every read \event{r}{} in \tr that reads the same variable as \event{w'}{},
  an event \event{r}{'} is appended immediately after \event{r}{} such that
  \event{r}{'} is a \Read{x} event and \POOrdered{\event{r}{}}{\event{r}{'}} in \trPrime.
\end{itemize}
Note that the \BRn{i-1} ordering for \trPrime is the same as the \BRn{i}
ordering for \tr, because there exists no read for which \event{w}{'} was the
last writer, and the added reads \event{r}{'} ensure that for any read \event{r}
for which $\event{r} \ltTR \event{w}{'}$ we have $\event{r}{'} \ltTRPrime
\Write{x}$.
\notes{
\mike{I agree with the above logic (particularly when \TROrdered{\event{w}{}}{\event{r}{}}).
\medskip\\
I wonder if this simplification might enable having something weaker than the new \BR rule.}
}%
In other words,
the \Read{x}--\Write{x} conflicts introduce the same ordering in \trPrime
as the original read--write conflicts between \event{w}{'} and its prior reads in \tr, and
\trPrime does not contain the write--write conflict on \event{w}{} and
\event{w}{'} found in \tr.
Thus \underline{in \trPrime}, \nBRnOrdered{\event{e}{_1}}{\event{e}{_2}}{i-1}.
By the induction hypothesis,
\trPrime has a predictable race or deadlock.
Let \trDoublePrime be a predictable trace of \trPrime that exposes a race or deadlock.
However, if we modify \trDoublePrime by removing the \Read{x} events and replacing the \Write{x} event with \event{w}{'},
the resulting trace is a predictable trace of \tr that exposes a race or deadlock.
Thus \tr has a predictable race or deadlock, which is a contradiction.



Thus for all $i$, \BRn{i} is sound.
Since $\ltBRn{\infty} \equiv \ltBR$, therefore \BR is sound.
\end{proof}

\label{sec:wbr-completeness}

\begin{theorem}[\WBR completeness]
  If an execution trace has a predictable race,
  then it has a \WBR-race.%
  \label{thm:wbr-completeness}
\end{theorem}

To prove the theorem, we use the following helper lemma:

\begin{lemma}[\WBR-ordered events cannot be reordered]
  Given an execution trace \tr, for any events \event{e}{_1} and \event{e}{_2}
  in \tr such that \WBROrdered{\event{e}{_1}}{\event{e}{_2}}, let \trPrime be a
  reordering of \tr where \event{e}{_1} and \event{e}{_2} have been reordered:
  either \TRPrimeOrdered{\event{e}{_2}}{\event{e}{_1}} or $\event{e}{_2} \in \trPrime \land \event{e}{_1} \notin \trPrime$.
  Then, \trPrime must not be a valid predictable trace of \tr.
  \label{lem:wbr-necessary-ordering}
\end{lemma}

The overall proof strategy is analogous to a corresponding proof for \DC~\cite{vindicator},
so we have relegated the proof of Lemma~\ref{lem:wbr-necessary-ordering} to
\iftoggle{extended-version}{Appendix~\ref{sec:completeness-lemma}.}{the extended arXiv version~\cite{depaware-extended-arxiv}.}

\begin{proof}[Proof of Theorem~\ref{thm:wbr-completeness}]
  Let us prove this theorem by contradiction. Let \tr be a trace
  with a predictable race on conflicting events \event{e}{_1} and \event{e}{_2} such that \TROrdered{\event{e}{_1}}{\event{e}{_2}}, but no \WBR-race.
  Let \trPrime be a predictable trace of \tr in which \event{e}{_1} and \event{e}{_2} are consecutive:
  \TRPrimeOrdered{\event{e}{_1}}{\event{e}{_2}} and $\nexists e \mid \TRPrimeOrdered{\event{e}{_1}}{\TRPrimeOrdered{\event{e}{}}{\event{e}{_2}}}$.
  
  Applying the definition of a \WBR-race (Section~\ref{sec:BR-WBR-relations}),
  either
  \WBROrdered{\event{e}{_1}}{\event{e}{_2}} or
  $\commonLocks{\event{e}{_1}}{\event{e}{_2}} \neq \emptyset$.
  If $\commonLocks{\event{e}{_1}}{\event{e}{_2}} \neq \emptyset$,
  then \trPrime violates the \LS rule of predictable traces.



  Thus \WBROrdered{\event{e}{_1}}{\event{e}{_2}}.
  By the definition of a predictable race, \event{e}{_1} and \event{e}{_2} must be read or
  write events, and must be on different threads. As a result, the \WBR ordering
  between \event{e}{_1} and \event{e}{_2} cannot be established by \WBR
  conflicting critical section ordering or ``$\Acquire{m}\ltWBR\Release{m} \implies \Release{m}\ltWBR\Release{m}$,'' which require \event{e}{_1} to be a
  release, and not by \PO since the events are on different threads.
  Therefore, \WBROrdered{\event{e}{_1}}{\event{e}{_2}} by
  \WBR transitivity, so there must exist an event \event{e} such that
  \WBROrdered{\event{e}{_1}}{\WBROrdered{\event{e}}{\event{e}{_2}}}.
  
  Since \event{e}{_1} and \event{e}{_2} are consecutive in \trPrime, either
  \TRPrimeOrdered{\event{e}{}}{\event{e}{_1}},
  \TRPrimeOrdered{\event{e}{_2}}{\event{e}{}}, or
  $\event{e}{_2} \in \trPrime \land \event{e}{} \notin \trPrime$.
  By Lemma~\ref{lem:wbr-necessary-ordering},
  any of these possibilities implies \trPrime is an invalid predictable trace of \tr, a contradiction.
\end{proof}

\subsection{Using Precise Dependence Information}


Up to this point, we have assumed that a branch event depends on every preceding read event in the same thread,
meaning that the condition \BrDepsOn{\event{b}{}}{\event{e}{_2}} in \WBR's handling of write--read critical sections
holds for every read \event{e}{_2} and branch \event{b}{}.
This assumption is needed unless static control dependence information is available
from conservative static program analysis (\eg,~\cite{mcr-s,Ferrante87theprogram}).
We tried out one kind of static analysis to compute static control dependencies
but found it provided no benefit, so our experiments do not use it (Section~\ref{sec:evaluation}).
Here we show some examples of how \WBR uses static control dependence information if it is available.

Figure~\ref{fig:sdg-needed-race} shows two executions that differ only in whether precise static control dependency information is available.
Figure~\ref{fig:sdg-needed-race:woSDG} has no control dependency information available, so each branch is conservatively dependent on all prior reads, or the information is available but the branch outcome \emph{may depend} on the prior read.
Figure~\ref{fig:sdg-needed-race:wSDG} has control dependency information that says that the branch outcome does \emph{not} depend on the prior read.
As a result, Figure~\ref{fig:sdg-needed-race:wSDG} has weaker \WBR ordering than
Figure~\ref{fig:sdg-needed-race:woSDG}, leading to a detected \WBR-race in Figure~\ref{fig:sdg-needed-race:wSDG} only.

\begin{figure}
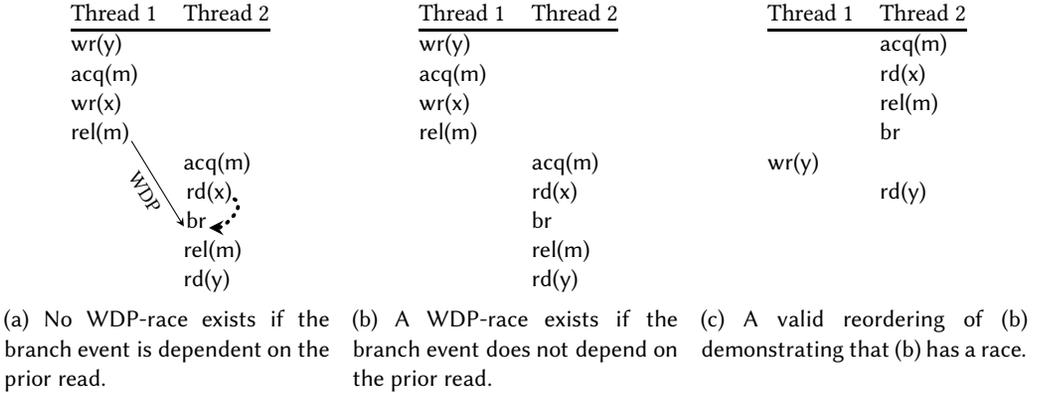

\captionsetup{farskip=0pt} 

\small
\centering

\hfill
\subfloat[No \WBR-race exists if the branch event is dependent on the prior read.]{
\centering
\sf
\makebox[.3\linewidth]{
\begin{tabular}{@{}ll@{}}
\textnormal{Thread 1} & \textnormal{Thread 2}\\\hline
\Write{y} \\
\Acquire{m} \\
\Write{x}\tikzmark{5} \\
\Release{m}\tikzmark{1} \\
                      &						  \Acquire{m} \\
                      &						  \tikzmark{6}\Read{x}\tikzmark{3} \\
                      &						  \tikzmark{2}\BrUniv{}\tikzmark{4} \\
                      & 					  \Release{m} \\
                      &						  \Read{y} \\
\end{tabular}}
\label{fig:sdg-needed-race:woSDG}}
\hfill
\subfloat[A \WBR-race exists if the branch event does not depend on the prior read.]{
\centering
\sf
\makebox[.3\linewidth]{
\begin{tabular}{@{}ll@{}}
\textnormal{Thread 1} & \textnormal{Thread 2}\\\hline
\Write{y} \\
\Acquire{m} \\
\Write{x}\tikzmark{10} \\
\Release{m} \\
                      &						  \Acquire{m} \\
                      &						  \Read{x} \\
                      &						  \BrUniv{} \\
                      & 					  \Release{m} \\
                      &						  \Read{y} \\
\end{tabular}}
\label{fig:sdg-needed-race:wSDG}}
\hfill
\subfloat[A valid reordering of (b) demonstrating that (b) has a race.]{
\centering
\sf
\makebox[.3\linewidth]{
\begin{tabular}{@{}ll@{}}
\textnormal{Thread 1} & \textnormal{Thread 2}\\\hline
                      &						  \Acquire{m} \\
                      &						  \Read{x} \\
                      & 					  \Release{m} \\
                      &						  \BrUniv{} \\
\Write{y} \\
                      &						  \Read{y} \\
\\
\\
\\
\end{tabular}}
\label{fig:sdg-needed-race:reordering}}

\textunderlink{1}{2}{\WBR}%
\forwardcurved{3}{4}{0}{-30}%

\caption{Example executions that differ only in the static control dependencies between branches and reads.
Dotted edges indicate reads that a branch depends on, \ie, \BrDepsOn{\event{b}{}}{\event{e}{}}.
If precise control dependence information rules out read--branch dependencies, \WBR can find additional races, such as the race on \code{y} in (b).
\label{fig:sdg-needed-race}}
\end{figure}



\section{\BR and \WBR Analyses}
\label{sec:wbr-analysis}


\newcommand\CHt{\ensuremath{\amsbb{C}_t[t := \amsbb{H}_t(t)]}\xspace}
\newcommand\YlHt{\ensuremath{\amsbb{Y}_l[t := \amsbb{H}_t(t)]}\xspace}





\newcommand\AcqX{\ensuremath{\mathit{Acq}^{K}}}

\begin{algorithm}[t]
\caption{\hfill \BR analysis at each event type, with differences from \WCP analysis}
\small

\newcommand\XHt{\ensuremath{\amsbb{K}_t[t := \amsbb{H}_t(t)]}\xspace}

\begin{algorithmic}[1]

  \Procedure{acquire}{$t, l$}
  \State \update{\amsbb{H}_t}{\addvc \amsbb{H}_l}
  \State \update{\amsbb{C}_t}{\addvc \amsbb{C}_l}
  \State \hspace{-1.8em}$+$\hspace{0.85em}\xspace\update{\amsbb{K}_t}{\addvc \amsbb{K}_l}\label{line:br-analysis:xdp-hb-comp-1}
  \State \hspace{-1.8em}$+$\hspace{0.85em}\xspace\textbf{foreach}\xspace $m \in L \setminus \{l\}$ \xspace\textbf{do} \update{\amsbb{Y}_m}{\addvc \amsbb{K}_l}\label{line:br-analysis:xdp-acquire-outer-locks} \Comment{wr--wr conflict \HB ordered before $l$ but not $m$}
  \State \hspace{-1.8em}$+$\hspace{0.85em}\xspace$\amsbb{Y}_l \gets \amsbb{K}_t$\label{line:br-analysis:xdp-save} \Comment{wr--wr conflict where later write did not hold $l$}
  \State \textbf{foreach}\xspace $t' \neq t$\xspace\textbf{do} $Acq_l(t').Enque(\CHt)$
  \State \hspace{-1.8em}$+$\hspace{0.85em}\xspace\textbf{foreach}\xspace $t' \neq t$ \xspace\textbf{do} $\AcqX_l(t').Enque(\XHt)$
  \EndProcedure


  \Procedure{release}{$t, l, L, R, W$}
  \Statex \hspace{-0.38em}\textcolor{gray}{$-$\hspace{0.85em}\xspace\textbf{while} $Acq_l(t).Front() \lessvc \CHt$ \textbf{do}}
  \While{\hspace{-4.6em}$+$\hspace{3.8em}\xspace$Acq_l(t).Front() \lessvc \CHt \lor \AcqX_l(t).Front() \lessvc \YlHt$}\label{line:br-analysis:xdp-use}
  \State $Acq_l(t).Deque()$
  \State \hspace{-3.3em}$+$\hspace{2.35em}\xspace$\AcqX_l(t).Deque()$
  \State \update{\amsbb{C}_t}{\addvc Rel_l(t).Deque()}
  \EndWhile
  \State \Bfor{x}{R}{\update{\amsbb{L}^r_{l,x}}{\addvc \amsbb{H}_t}}
  \State \Bfor{x}{W}{\update{\amsbb{L}^w_{l,x}}{\addvc \amsbb{H}_t}}
  \State $\amsbb{H}_l \gets \amsbb{H}_t$
  \State $\amsbb{C}_l \gets \amsbb{C}_t$
  \State \hspace{-1.8em}$+$\hspace{0.85em}\xspace $\amsbb{K}_l \gets \amsbb{K}_t$\label{line:br-analysis:xdp-hb-comp-2}
  \State \textbf{foreach}\xspace $t' \neq t$ \xspace\textbf{do} $Rel_l(t').Enque(\amsbb{H}_t)$
  \State \update{\amsbb{H}_t(t)}{+ 1}
  \EndProcedure
  
  \Procedure{read}{$t, x, L$}
  \State \hspace{-1.8em}$+$\hspace{0.85em}\xspace\update{\amsbb{C}_t}{\addvc \amsbb{B}_{t,x}}\Comment{Apply prior write--write conflict}\label{line:br-analysis:wr-wr-apply}
  \State \update{\amsbb{C}_t}{\addvc_{l \in L} \amsbb{L}^w_{l,x}}\label{line:br-analysis-wr-rd-edge}
  \State \textbf{check} $\amsbb{W}_x \lessvc \CHt$\label{line:br-analysis:check-write-read-race}\Comment{Write--read race?}
  \State $\amsbb{R}_x(t) \gets \amsbb{H}_t(t)$
  \EndProcedure

  \Procedure{write}{$t, x, L$}
  \State \update{\amsbb{C}_t}{\addvc_{l \in L} \amsbb{L}^r_{l,x}}\label{line:br-analysis:rd-wr-edge}
  \Statex \hspace{-0.38em}\textcolor{gray}{$-$\hspace{0.85em}\xspace$\update{\amsbb{C}_t}{\addvc_{l \in L}\amsbb{L}^w_{l,x}}$}
  \State \textbf{check} $\amsbb{W}_x \lessvc \CHt \addvc_{l \in L} \amsbb{L}^w_{l,x}$\label{line:br-analysis:check-write-write-race}\Comment{Write--write race?}
  \State \hspace{-1.8em}$+$\hspace{0.85em}\xspace$\amsbb{B}_{t,x} \gets \addvc_{l \in L}\amsbb{L}^w_{l,x}$\Comment{Records write--write conflict for future read or release}\label{line:br-analysis:wr-wr-delayed}
  \State \hspace{-1.8em}$+$\hspace{0.85em}\xspace$\amsbb{K}_{t} \gets \amsbb{K}_t \addvc_{l \in L}\amsbb{L}^w_{l,x}$\Comment{Used for rel--rel through wr--wr conflicts}\label{line:br-analysis:xdp-conflict}
  \State \hspace{-1.8em}$+$\hspace{0.85em}\xspace\textbf{foreach}\xspace $m \in L$ \xspace\textbf{do} $\amsbb{Y}_{m} \gets \amsbb{Y}_m \addvc_{l \in L \setminus \{m\}}\amsbb{L}^w_{l,x}$\label{line:br-analysis:xdp-update-outer}\Comment{wr--wr conflicts where earlier write did not hold $m$}
  \State \textbf{check} $\amsbb{R}_x \lessvc \CHt$\label{line:br-analysis:check-read-write-race}\Comment{Read--write race?}
  \State $\amsbb{W}_x(t) \gets \amsbb{H}_t(t)$
  \EndProcedure

  \Procedure{branch}{$t, L$}
  \State \textbf{skip} \Comment{No analysis at branch events}
  \EndProcedure

\end{algorithmic}
\label{alg:br-analysis}
\end{algorithm}

\begin{algorithm}
\caption{\hfill \WBR analysis at each event type, with differences from \DC analysis}
\small

\begin{algorithmic}[1]

\Procedure{acquire}{$t, l$}
\State \textbf{foreach}\xspace $t' \neq t$ \xspace\textbf{do} $Acq_{l,t'}(t).Enque(\amsbb{C}_t)$
\EndProcedure


\Procedure{release}{$t, l, R, W$}
\ForEach{$t' \neq t$}
\While{$Acq_{l,t}(t').Front() \lessvc \amsbb{C}_t$}
\State $Acq_{l,t}(t').Deque()$
\State \update{\amsbb{C}_t}{\addvc Rel_{l,t}(t').Deque()}
\EndWhile
\EndFor
\State \Bfor{x}{W}{$\amsbb{L}_{l,x}\gets \amsbb{C}_t$} \Comment{Record release time for writes in critical section} \label{line:wbr-analysis:record-release-time}
\Statex \hspace{-0.25em}\textcolor{gray}{$-$\hspace{0.85em}\xspace\Bfor{x}{R}{$\amsbb{L}^r_{l,x}\gets \amsbb{C}_t$}}
\State \textbf{foreach}\xspace $t' \neq t$ \xspace\textbf{do} $Rel_{l,t'}(t).Enque(\amsbb{C}_t)$
\State \update{\amsbb{C}_t(t)}{+ 1} \label{line:wbr-analysis:release-increment}
\EndProcedure

\Procedure{read}{$t, x, e, L$}
\Statex \hspace{-0.25em}\textcolor{gray}{$-$\hspace{0.85em}\xspace\update{\amsbb{C}_t}{\bigsqcup_{l \in (L \cap L^w_{x,t'})} \amsbb{L}_{l,x}}}
\Statex \hspace{-0.25em}\textcolor{gray}{$-$\hspace{0.85em}\xspace\textbf{check} $\amsbb{W}_x \lessvc \amsbb{C}_t$}
\State \hspace{-1.8em}$+$\hspace{0.85em}\xspace\textbf{foreach} thread $t' \ne t$ \textbf{check} $\amsbb{W}_x(t') \le \amsbb{C}_t(t') \lor L^w_{x,t'}\cap L \neq \emptyset$ \Comment{Check write--read race} \label{line:wbr-analysis:check-wr-rd-race}
\State \hspace{-1.8em}$+$\hspace{0.85em}\xspace\textbf{let} $t' \gets T_x$ \label{line:wbr-analysis:get-last-writer-thread} \Comment{Get last writer thread of $x$}
\If{\hspace{-2.8em}$+$\hspace{1.85em}\xspace$t' \notin \{ \nolastwr, t \} \land L \cap L^w_{x,t'} \neq \emptyset$} \Comment{Write--read conflict}\label{line:wbr-analysis:wr-rd-detected}
\State \hspace{-3.3em}$+$\hspace{2.35em}\xspace$\amsbb{B}_{t,x} \gets \bigsqcup_{l \in (L \cap L^w_{x,t'})} \amsbb{L}_{l,x}$ \Comment{Record time of writer thread's release for later use} \label{line:wbr-analysis:defer-release-time}
\lIf{\hspace{-4.3em}$+$\hspace{3.45em}\xspace$\amsbb{B}_{t,x} \nsqsubseteq \amsbb{C}_t$}{$\update{D_t}{\cup \{\langle x, e \rangle\}}$} \Comment{Record read} \label{line:wbr-analysis:add-to-d_t}
\EndIf
\State $\amsbb{R}_x(t) \gets \amsbb{C}_t(t)$
\State \hspace{-1.8em}$+$\hspace{0.85em}\xspace$L^r_{x,t} \gets L$
\EndProcedure

\Procedure{write}{$t, x, L$}
\Statex \hspace{-0.25em}\textcolor{gray}{$-$\hspace{0.85em}\xspace \update{\amsbb{C}_t}{\addvc \bigsqcup_{l \in (L \cap L^w_{x,t'})} \amsbb{L}_{l,x} \addvc \bigsqcup_{l \in (L \cap L^r_{x,t'})} \amsbb{L}^r_{l,x}}}
\Statex \hspace{-0.25em}\textcolor{gray}{$-$\hspace{0.85em}\xspace\textbf{check} $\amsbb{W}_x \lessvc \amsbb{C}_t$}
\State \hspace{-1.8em}$+$\hspace{0.85em}\xspace\textbf{foreach} thread $t' \ne t$ \textbf{check} $\amsbb{W}_x(t') \le \amsbb{C}_t(t') \lor L^w_{x,t'} \cap L \neq \emptyset$ \Comment{Check write--write race} \label{line:wbr-analysis:check-wr-wr-race}
\Statex \hspace{-0.25em}\textcolor{gray}{$-$\hspace{0.85em}\xspace\textbf{check} $\amsbb{R}_x \lessvc \amsbb{C}_t$}
\State \hspace{-1.8em}$+$\hspace{0.85em}\xspace\textbf{foreach} thread $t' \ne t$ \textbf{check} $\amsbb{R}_x(t') \le \amsbb{C}_t(t') \lor L^r_{x,t'} \cap L \neq \emptyset$ \Comment{Check read--write race} \label{line:wbr-analysis:check-rd-wr-race}
\State $\amsbb{W}_x(t) \gets \amsbb{C}_t(t)$
\State \hspace{-1.8em}$+$\hspace{0.85em}\xspace$L^w_{x,t} \gets L$
\State \hspace{-1.8em}$+$\hspace{0.85em}\xspace$T_x \gets t$ \Comment{Set last writer thread of $x$} \label{line:wbr-analysis:update-last-writer}
\EndProcedure

\Procedure{branch}{$t, e, L$}
\Statex \hspace{-0.25em}\textcolor{gray}{$-$\hspace{0.85em}\xspace\textbf{skip}}
\State \hspace{-1.8em}$+$\hspace{0.85em}\xspace$\textbf{foreach} \langle x,r \rangle \in D_t \mid \BrDepsOn{e}{r} \xspace\textbf{ do }\xspace \update{\amsbb{C}_t}{\addvc \amsbb{B}_{t,x}}$ \Comment{Add release--branch ordering} \label{line:wbr-analysis:use-release-time}
\State \hspace{-1.8em}$+$\hspace{0.85em}\xspace\update{D_t}{\setminus \{ \langle x,r \rangle \in D_t \mid \BrDepsOn{e}{r} \}} \label{line:wbr-analysis:remove-dependencies} \Comment{Remove dependencies that were applied}
\EndProcedure
\end{algorithmic}
\label{alg:wbr-analysis}
\end{algorithm}

\emph{\BR analysis} and
\emph{\WBR analysis} are new online dynamic program analyses that compute \BR and \WBR and detect \BR- and \WBR-races, respectively.
Algorithms~\ref{alg:br-analysis} and \ref{alg:wbr-analysis} show \BR and \WBR analyses, respectively, for each kind of event.
This section's notation and terminology follow the \WCP and \DC papers' to some extent~\cite{vindicator,wcp}.

Both algorithms show the differences relative to prior analyses
(\BR versus \WCP and \WBR versus \DC) by labeling lines with
``$+$'' to show logic added by our analyses and
``$-$'' with grayed-out text to show lines removed by our analyses.
Algorithm~\ref{alg:br-analysis} shows that \BR analysis requires several changes to \WCP analysis.
These changes are for tracking write--write conflicts to add ordering (1) when a future read is detected on the second write's thread
and (2) when two critical sections are HB ordered by the conflicting writes.
In addition, \BR analysis avoids reporting write--write races for writes in critical sections on the same lock.

Algorithm~\ref{alg:wbr-analysis} shows that \WBR analysis makes several significant changes to \DC analysis.
These changes are primarily to deal with branches, by recording information about write--read dependencies
at read events (lines~\ref{line:wbr-analysis:get-last-writer-thread}--\ref{line:wbr-analysis:add-to-d_t})
and using the recorded information at branch events (lines~\ref{line:wbr-analysis:use-release-time}--\ref{line:wbr-analysis:remove-dependencies}).
Unlike \DC analysis, \WBR analysis does not establish ordering at any conflicting accesses, and it never needs to track
ordering from a read to another access (since it detects only write--read conflicts).
In addition, \WBR analysis ensures it does not report races on accesses in critical sections on the same lock.


\paragraph{Analysis details}

In both \BR and \WBR analyses,
the procedural parameters $t$ and $l$ are the current thread and lock;
$L$ is the set of locks held by the thread performing the current event;
$R$ and $W$ are the sets of variables that were read and written in the ending critical section on $l$;
and $e$ represents the current read or branch event (for detecting branch dependencies).


The analysis uses \emph{vector clocks}~\cite{vector-clocks} to represent logical \BR or \WBR time. 
A vector clock $C : \mathit{Tid} \mapsto \mathcal{N}$ maps
each thread to a nonnegative integer. Operations on vector clocks are
pointwise comparison ($C_1 \sqsubseteq C_2 \iff \forall t . C_1(t) \leq C_2(t)$) and
pointwise join ($C_1 \sqcup C_2 \equiv \lambda t . \mathit{\max(C_1(t), C_2(t))}$):
\notes{
Intuitively, for events \event{e}{_1} and \event{e}{_2} executed by different threads at times $C_1$ and $C_2$, respectively,
$C_1 \sqsubseteq C_2$ if and only if $\BROrdered{\event{e}{_1}}{\event{e}{_2}}$ (Algorithm~\ref{alg:br-analysis})
or $\WBROrdered{\event{e}{_1}}{\event{e}{_2}}$ (Algorithm~\ref{alg:wbr-analysis}).
}%

Both analyses maintain the following state:

\begin{itemize}

\item
$\amsbb{C}_t$ is a vector clock that represents the current \BR or \WBR time for thread $t$.

\item
$\amsbb{R}_{x}$ and $\amsbb{W}_{x}$ are vector clocks that represent the \BR or \WBR time of the last reads and writes to $x$.

\item
$\mathit{Acq}_{l}(t)$ and $\mathit{Rel}_{l}(t)$ (\BR) and
$\mathit{Acq}_{l,t'}(t)$ and $\mathit{Rel}_{l,t'}(t)$ (\WBR) are queues of vector clocks
that help compute the ``$\Acquire{m}\prec\Release{m}$ implies $\Release{m}\prec\Release{m}$'' property (Table~\ref{tab:partial-order-definitions}).

\end{itemize}

In addition,
\BR analysis maintains the following state:

\begin{itemize}

\item
$\amsbb{H}_t$ is a vector clock that represents the current \HB time for thread $t$.
\CHt, which evaluates to a vector clock with every element
equal to $\amsbb{C}_t$ except that element $t$ is equal to $\amsbb{H}_t(t)$, represents $\ltBR \cup \ltPO$.

\item
$\amsbb{L}^r_{l,x}$ and $\amsbb{L}^w_{l,x}$ are vector clocks that represent the \HB times of
critical sections on $l$ that read and wrote $x$, respectively.

\item
$\amsbb{B}_{t,x}$ is a vector clock that represents the \BR time of a release event
of a critical section on lock $m$ containing a write event \event{w}{} to $x$ such that
a later write event \event{w}{'} to $x$ by $t$ conflicts with the write and $m \in \commonLocks{\event{w}{}}{\event{w}{'}}$.

\item
$\amsbb{K}_t$ is a vector clock that represents the time of a partial order
that orders conflicting writes
(\ie, $r_1 \prec w_2$ if $w_1 \ltTR w_2$ are conflicting write events,
$r_1$ is a release and $a_2$ is an acquire on the same lock,
$w_1 \in \CS{r_1}$, and $w_2 \in \CS{a_2}$)
and composes with HB, for thread $t$.

\item
$\amsbb{Y}_l$ is a vector clock that represents the time of a partial order
that orders conflicting writes
\emph{except} those involving the current critical section on lock $l$
(\ie, $r_1 \prec w_2$ if $w_1 \ltTR w_2$ are conflicting write events,
$r_1$ is a release and $a_2$ is an acquire on the same lock,
$a_2$ is \emph{not} the most recent acquire of lock $l$,
$w_1 \in \CS{r_1}$, and $w_2 \in \CS{a_2}$)
and composes with HB, for the thread currently holding lock $l$.
(While $l$ is unheld, $\amsbb{Y}_l$ is undefined.)

\item
$\AcqX_l(t)$ is a queue of vector clocks
akin to $\mathit{Acq}_{l}(t)$, but instead of representing \BR times,
each vector clock represents the time of a partial order that orders conflicting writes
and composes with HB.

\end{itemize}

\WBR analysis maintains the following additional state:

\begin{itemize}

\item
$\amsbb{L}_{l,x}$ is a vector clock that represents the \WBR time of critical sections on $l$ that wrote $x$.

\item
$\amsbb{B}_{t,x}$ is a vector clock that represents the \WBR time of the last write event \event{e}{} to $x$ such that
a later read event \event{e}{'} to $x$ by $t$ conflicts with \event{e}{} and $\commonLocks{\event{e}{}}{\event{e'}{}} \ne \emptyset$.

\item
$T_x$ is the last thread to write to $x$, or \nolastwr if no thread has yet written $x$.

\item
$D_t$ is a set of pairs $\langle x,\event{e}{} \rangle$
such that event \event{e}{} is a read to $x$ by a thread that has not (yet) executed a branch $b$ such that \BrDepsOn{b}{e}.

\item
$L^r_{x,t}$ and $L^w_{x,t}$ are sets of locks that were held by thread $t$ when it last read and wrote variable $x$, respectively.

\end{itemize}

Initially, every vector clock maps every thread to 0,
except $\forall t.\amsbb{H}_t(t) = 1$ for \BR analysis, and $\forall t.\amsbb{C}_t(t) = 1$ for \WBR analysis.
Every queue and set is initially empty.

The analyses compute the ``$\Acquire{m}\prec\Release{m}$ implies $\Release{m}\prec\Release{m}$'' property (Table~\ref{tab:partial-order-definitions}) similarly to
\WCP and \DC analyses, respectively~\cite{wcp,vindicator}.
Briefly, $\mathit{Acq}_{l}(t)$ and $\mathit{Rel}_{l}(t)$ contain times of \Acquire{l} and \Release{l} operations (by any thread other than $t$) 
such that the \Acquire{l} operation is not yet \BR ordered to a following \Release{l} by thread $t$.
$\mathit{Acq}_{l,t'}(t)$ and $\mathit{Rel}_{l,t'}(t)$ contain times of \Acquire{l} and \Release{l} operations by thread $t$
such that the \Acquire{l} operation is not yet \WBR ordered to a following \Release{l} by thread $t'$.
\BR additionally uses $\AcqX_l(t)$ for \Acquire{l} operations for write--write
conflicts that are \HB ordered before the acquire, which is required for
soundness since \BR does not directly order write--write conflicts like \WCP.

The analyses provide \BR and \WBR's handling of conflicting critical sections
by detecting some kinds of conflicts on accesses holding a common lock.
\BR analysis orders the earlier release of a common lock to the current event for write--read and read--write conflicts
using $\amsbb{L}^r_{l,x}$ and $\amsbb{L}^w_{l,x}$
(lines~\ref{line:br-analysis-wr-rd-edge} and \ref{line:br-analysis:rd-wr-edge} in Algorithm~\ref{alg:br-analysis}).
For write--write conflicts, \BR analysis stores the time of the earlier release of a common lock
in $\amsbb{B}_{t,x}$ (line~\ref{line:br-analysis:wr-wr-delayed})
to order the earlier write to a later read of $x$ by the current thread (line~\ref{line:br-analysis:wr-wr-apply}).

\BR analysis detects release--release ordering by write--write conflicts
using $\amsbb{K}_t$, $\amsbb{Y}_l$, and $Acq^K_l(t)$.
In particular, the analysis tracks write--write ordering using $\amsbb{K}_t$,
and it tracks write--write ordering
\emph{excluding the current critical section on $l$} using $\amsbb{Y}_l$,
in accordance with the \BR definition.
Line~\ref{line:br-analysis:xdp-use} compares $\amsbb{Y}_l$ and $Acq^K_l(t)$ to check whether two critical sections are
ordered by write--write ordering excluding the current critical section on $l$.
Effectively, $\amsbb{Y}_l$, $\amsbb{K}_t$, and $\AcqX_l(t)$ represent relations that
\BR can use to check its release--release rules (page~\pageref{br-rel-rel-rules}).

\notes{
\mike{Here's some text that was here:}
\BR analysis extends WCP's check for release--release edges to also consider the
elided write--write conflicts. The analysis saves elided write--write conflicts as they
occur (line~\ref{line:br-analysis:xdp-conflict}). These elided conflicts are
propagated with \HB (lines~\ref{line:br-analysis:xdp-hb-comp-1}
and~\ref{line:br-analysis:xdp-hb-comp-2}). At an acquire, the analysis
captures the current state of these elided conflicts
(line~\ref{line:br-analysis:xdp-save}). When computing release--release ordering
(line~\ref{line:br-analysis:xdp-use}), the analysis checks if the acquire is \BR ordered
to the release ($Acq_l(t).Front() \lessvc \CHt$), or if the acquire
is before an elided write--write conflict ($Acq^K_l(t).Front() \lessvc \YlHt$).

To see how \BR analysis detects release--release ordering from write--write conflicts,
first note that either the earlier or the later write in the conflict (or both) must
not have held the relevant lock. Let us consider the case
where the later write did not hold the lock. Then, that write must have occurred
before the later critical section started, in which case it would have been captured
in $\amsbb{Y}_l$ at line~\ref{line:br-analysis:xdp-save}
\textcolor{red}{if the later write occurred before the later critical section started and line~\ref{line:br-analysis:xdp-acquire-outer-locks} otherwise.}
\mike{Huh? Isn't this sentence starting with the supposition that the later write was outside the critical section.}
If the earlier
write did not hold the lock, then
line~\ref{line:br-analysis:xdp-update-outer} would have updated $\amsbb{Y}_l$.
Either way, $\YlHt$ updated with the write--write conflict will be greater than
$Acq^K_l(t).Front()$, triggering \BR ordering of the release events. Note here
that the analysis checks $\amsbb{Y}_l$ against $Acq^K_l(t)$, which is a queue that tracks
the write--write conflict times at acquires, because the write--write conflict
clocks are more restricted than \BR and cannot be compared with the existing
$Acq_l(t)$ queues.
}

\WBR analysis orders the release on the writer's executing thread to a later branch dependent on the read.
The analysis does so by recording the time of the last writer's release in $\amsbb{L}_{l,x}$
(line~\ref{line:wbr-analysis:record-release-time} in Algorithm~\ref{alg:wbr-analysis}).
Later, when a conflicting read occurs on thread $t$ holding $l$, the analysis uses
$\amsbb{L}_{l,x}$ to get the time for the last conflicting writer $T_x$'s release,
and stores this time in $\amsbb{B}_{t,x}$ (line~\ref{line:wbr-analysis:defer-release-time}).
When $t$ executes a branch dependent on the prior conflicting read,
\WBR adds ordering from the release to the current branch (line~\ref{line:wbr-analysis:use-release-time}).
The analysis detects the dependent branch using $D_t$, which contains a set of $\langle x, \event{e}{} \rangle$ pairs
for which a branch dependent on read event \event{e}{} has not yet executed (line~\ref{line:wbr-analysis:update-last-writer}).
\later{\kaan{Reviewer B says that ``I found this hard to follow. An example could help.''
\mike{How about referring to Figures~\ref{fig:sdg-needed-race:woSDG} and \ref{fig:sdg-needed-race:wSDG}?}}}%
The exact representation of \event{e}{} and behavior of \BrDepsOn{b}{e} are implementation dependent.

\BR analysis provides composition with \HB
using $\amsbb{C}_t$, $\amsbb{H}_t$, and
\CHt.\footnote{In Algorithm~\ref{alg:br-analysis},
$\amsbb{C}_t$ and \CHt are analogous to the \WCP paper's $\amsbb{P}_t$ and $\amsbb{C}_t$, respectively~\cite{wcp}.}
\WBR analysis includes \PO with the increment of $\amsbb{C}_t(t)$ at line~\ref{line:wbr-analysis:release-increment}.

The analyses check the conditions for a \BR- or \WBR-race by
using $\amsbb{R}_{x}$ and $\amsbb{W}_{x}$. 
Since the analyses do not order all pairs of conflicting accesses,
unordered conflicting accesses are not sufficient to report a race.
\BR analysis uses the vector clock $\amsbb{L}^w_{l,x}$
and \WBR analysis uses the locksets $L^r_{x,t}$ and $L^w_{x,t}$ to check if the current
and prior conflicting accesses' held locks overlap
(lines~\ref{line:br-analysis:check-write-read-race},
\ref{line:br-analysis:check-write-write-race}, and
\ref{line:br-analysis:check-read-write-race} in Algorithm~\ref{alg:br-analysis};
lines~\ref{line:wbr-analysis:check-wr-rd-race},
\ref{line:wbr-analysis:check-wr-wr-race}, and
\ref{line:wbr-analysis:check-rd-wr-race} in Algorithm~\ref{alg:wbr-analysis}).

\paragraph{Atomic accesses and operations.}
\label{subsec:atomic-accesses}

We extend \BR and \WBR analyses to handle accesses that have ordering or atomicity
semantics:
\emph{atomic accesses} that introduce ordering such as Java \code{volatile} and C++ \code{atomic} accesses, and
\emph{atomic read-modify-write operations} such as atomic test-and-set.
\notes{
Atomic accesses and operations are the building blocks for implementing lock-free data structures such as those
in the \code{java.util.concurrent} package.
\mike{We don't want to oversell this, since as we mention in Impl, we don't actually handle any of these.}
}%

The following pseudocode shows how we extend \WBR analysis (Algorithm~\ref{alg:wbr-analysis})
to handle atomic reads and writes and atomic operations.
(The extensions to \BR analysis are similar but also add conflicting read--write and write--write--read ordering.)

\begin{algorithmic}[1]
\small

\Procedure{atomicRead}{$t, x, e$}
\State \textbf{let} $t' \gets T_x$ \Comment{Get last writer thread of $x$}
\If{$t' \notin \{ \nolastwr, t \} \land \amsbb{W}_x \nsqsubseteq \amsbb{C}_t$} \Comment{Write--read conflict}\label{line:wbr-analysis:wr-rd-detected-volatile}
\State $\amsbb{B}_{t,x} \gets \amsbb{W}_x$ \Comment{Record the write} \label{line:wbr-analysis:defer-release-time-volatile}
\lIf{$\amsbb{B}_{t,x} \nsqsubseteq \amsbb{C}_t$}{$\update{D_t}{\cup \{\langle x, e \rangle\}}$} \label{line:wbr-analysis:add-to-d_t-volatile}
\EndIf
\EndProcedure

\Procedure{atomicWrite}{$t, x$}
\State $\amsbb{W}_x \gets \amsbb{C}_t$
\State $T_x \gets t$ \Comment{Set last writer thread of $x$} \label{line:wbr-analysis:update-last-writer-volatile}
\EndProcedure

\Procedure{atomicReadModifyWrite}{$t, x, e$}
  \State \textsc{atomicRead}($t$, $x$, $e$)
  \State \textsc{atomicWrite}($t$, $x$)
\EndProcedure

\end{algorithmic}

In essence, the analysis handles atomic accesses
like regular accesses contained in single-access critical sections on a unique lock to the accessed variable.
The analysis treats an atomic operation as an atomic read followed by an atomic write.

\paragraph{Handling races.}
\label{sec:analyses:handling-races}


The behavior of programs with data races is unreliable~\cite{memory-models-cacm-2010,dolan-bounding-races},
but our analyses' instrumentation performs synchronization operations before accesses, which generally ensures sequential consistency (SC)
for all program executions.
A different problem is that if an analysis continues detecting races after the first race,
then additional detected races are not necessarily real races because they may depend on an earlier race
(\ie, if the earlier race were ordered, the later race would not exist).
Our implementation (Section~\ref{subsec:implementation}) addresses this issue by
treating racing accesses as if they were contained in
single-access critical sections on the same lock.
Specifically, \BR analysis orders one racing event to the other for write--read and read--write races,
and \WBR analysis orders write--read races to a branch that depends on the read if the write is the last writer of the read.
For example, after detecting a race in line~\ref{line:wbr-analysis:check-wr-rd-race}, 
\WBR analysis performs the following: $\textbf{if} \: t' = T_x \: \textbf{then} \; \amsbb{B}_{t,x} \gets \amsbb{W}_x(t');\: D_t \gets D_t \cup \{ \langle x, e \rangle \}$.

\paragraph{Time and space complexity.}

The run times of
\BR and \WBR analyses are each linear in the number of events.
For \WCP and \DC analyses~\cite{wcp,vindicator}, for $N$ events, $L$ locks held at once by a thread, and $T$ threads,
time complexity for an entire execution trace is $O(N \times (L \times T + T^2))$.
For \BR analysis, time complexity is generally the same as for \WCP analysis,
except that the run time of
line~\ref{line:br-analysis:xdp-update-outer} in Algorithm~\ref{alg:br-analysis}
is $O(L^2 \times T)$.
As a result, \BR's run time is
$O(N \times (L^2 \times T + T^2))$.
Since $L$ is the number of locks \emph{simultaneously} held by a single thread,
we expect it to be small in practice.
For \WBR analysis, time complexity for each event type is the same as for \DC analysis,
and the run time for branch events can be amortized over read events.

%
%
%
%

\notes{
Fortunately, space usage is sub-linear in $N$ in practice.
\mike{We probably shouldn't claim that since we don't show space results. :)}}

\section{Verifying \WBR-Races}
\label{sec:vindication-summary}

\WBR analysis is unsound, so a \WBR-race may not indicate a predictable race
(\ie, there may be no race exposed in any predictable trace).
To avoid reporting false races,
our approach post-processes each detected \WBR-race with an algorithm called \checkWBRRace.
Here we overview \checkWBRRace;
\iftoggle{extended-version}{Appendix~\ref{sec:vindication-full}}{the extended arXiv version~\cite{depaware-extended-arxiv}}
presents \checkWBRRace in detail with an algorithm and examples.

To support performing \checkWBRRace on \WBR-races,
\WBR analysis builds a \emph{constraint graph} in which execution events are nodes,
and initially edges represent \WBR ordering.
\CheckWBRRace discovers and adds additional constraints to the graph
that enforce lock semantics (\LS) and last-writer (\LW) rules.
\CheckWBRRace uses the constraint graph to attempt to construct a predictable trace
that exposes the \WBR-race as a predictable race.


\CheckWBRRace extends prior work's \checkDCRace algorithm for checking \DC-races~\cite{vindicator}.
\CheckWBRRace differs from \checkDCRace primarily in the following way.
\CheckWBRRace computes and adds \LW constraints to the constraint graph
for all reads that must be causal events in the predictable trace.
Importantly, \CheckWBRRace computes causal reads and adds \LW constraints at each iteration of adding constraints
and at each attempt at building a predictable trace.

Algorithm~\ref{alg:abbrv-vindication} shows \checkWBRRace at a high level;
\iftoggle{extended-version}{Appendix~\ref{sec:vindication-full}}{the extended arXiv version~\cite{depaware-extended-arxiv}}
presents \checkWBRRace in detail.
\CheckWBRRace takes the initial constraint graph (\Gv) and a \WBR-race ($\event{e}{_1}, \event{e}{_2}$) as input (line~\ref{line:abbrv:vindicate-inputs}).
It first calls \textsc{AddConstraints} (line~\ref{line:abbrv:call-addconstraints}),
which adds necessary constraints to \Gv and returns an updated constraint graph.
\textsc{AddConstraints} first adds \emph{consecutive-event} constraints (\ie, edges) to \Gv
to enforce that any predictable trace must execute \event{e}{_1} and \event{e}{_2} consecutively (line~\ref{line:abbrv:consecutive-event}).
\textsc{AddConstraints} then computes the set of causal events for any predictable trace constrained by \Gv,
which it uses to add \LW constraints to \Gv, ensuring that every causal read
in a predictable trace can have the same last writer as in the original trace (line~\ref{line:abbrv:last-writer}).
Next, \textsc{AddConstraints} adds \LS constraints to \Gv,
by identifying critical sections on the same lock that are partially ordered
and thus must be fully ordered to obey \LS rules (line~\ref{line:abbrv:lock-semantics}).
Since added \LW and \LS constraints may lead to new \LS and \LW constraints being detected,
respectively, \textsc{AddConstraints} iterates until it finds no new constraints to add (lines~\ref{line:abbrv:do}--\ref{line:abbrv:while}).


\begin{algorithm*}[t]
\algdef{SE}[DOWHILE]{Do}{doWhile}{\algorithmicdo}[1]{\algorithmicwhile\ #1}%
\renewcommand{\CORE}{causal\xspace}
\renewcommand\GetCOREReads{\textsc{GetCausalReads}\xspace}
\caption{\hfill Check if \WBR-race is a true predictable race (high-level version of algorithm)}
\begin{small}
\end{small}
\small
\begin{algorithmic}[1]
\Procedure{\CheckWBRRace}{$\Gv, \event{e}{_1}, \event{e}{_2}$} \Comment{Inputs: constraint graph and \WBR-race events}\label{line:abbrv:vindicate-inputs}
	\State $\Gv \gets \textsc{AddConstraints}(G, \event{e}{_1},\event{e}{_2})$\label{line:abbrv:call-addconstraints}
	\lIf {\Gv has a cycle reaching \event{e}{_1} or \event{e}{_2}}{\textbf{return} \underline{No predictable race}}\label{line:abbrv:cycle}
	\State \algorithmicelse\
	\Indent
		\State $\trPrime \gets \textsc{ConstructReorderedTrace}(\Gv, \event{e}{_1}, \event{e}{_2})$ \Comment{Non-empty iff predictable trace constructed}\label{line:abbrv:call-constructreorderedtrace}
		\lIf {$\trPrime \ne \langle \, \rangle\,$}{\textbf{return} \underline{Predictable race witnessed by \trPrime}} \Comment{Check for non-empty trace}\label{line:abbrv:predictable-race}
		\lElse{\textbf{return} \underline{Don't know}}\label{line:abbrv:dont-know}
	\EndIndent
\EndProcedure

\medskip

\Procedure{AddConstraints}{$\Gv, \event{e}{_1}, \event{e}{_2}$}
	\State Add consecutive-event constraints to \Gv\label{line:abbrv:consecutive-event}
	\Do\label{line:abbrv:do}
		\State Compute causal reads and add last-writer (\LW) constraints to \Gv\label{line:abbrv:last-writer}
		\State Add lock-semantics (\LS) constraints to \Gv\label{line:abbrv:lock-semantics}
	\doWhile {$G$ has changed}\label{line:abbrv:while}
	\State \textbf{return} $\Gv$
\EndProcedure
\end{algorithmic}
\label{alg:abbrv-vindication}
\end{algorithm*}

The constraints added by \textsc{AddConstraints} are necessary but insufficient constraints on
any trace exposing a predictable race on \event{e}{_1} and \event{e}{_2}.
Thus if \Gv has a cycle that must be part of any predictable trace, then
the original trace has no predictable race on \event{e}{_1} and \event{e}{_2}
(line~\ref{line:abbrv:cycle}).
Otherwise, \textsc{AddConstraints} calls \textsc{ConstructReorderedTrace} (line~\ref{line:abbrv:call-constructreorderedtrace}),
which attempts to construct a legal predictable trace \trPrime.
\textsc{ConstructReorderedTrace} is a greedy algorithm that starts from \event{e}{_1} and \event{e}{_2}
and works backward, adding events in reverse order that satisfy \Gv's constraints
and also conform to \LS and \LW rules (\Gv's constraints are necessary but insufficient).
If \textsc{ConstructReorderedTrace} returns a (non-empty) trace \trPrime,
it is a legal predictable trace exposing a race on \event{e}{_1} and \event{e}{_2} (line~\ref{line:abbrv:predictable-race}).
Otherwise, \textsc{ConstructReorderedTrace} returns an empty trace,
which means that it could not find a predictable race, although one may exist (line~\ref{line:abbrv:dont-know}).

\section{Evaluation}
\label{sec:evaluation}


This section evaluates the predictive race detection effectiveness and
run-time performance of this paper's approaches.


\subsection{Implementation}
\label{subsec:implementation}

We implemented \BR and \WBR analyses and \checkWBRRace by extending the publicly available \emph{Vindicator} implementation,
which includes \HB, \WCP, and \DC analyses and
\checkDCRace~\cite{vindicator}.\footnote{\url{https://github.com/PLaSSticity/Vindicator}}
Vindicator is built on top of \emph{RoadRunner}, a dynamic analysis framework for concurrent Java
programs~\cite{roadrunner}.\footnote{\url{https://github.com/stephenfreund/RoadRunner/releases/tag/v0.5}}
We extended RoadRunner to instrument branches to enable \WBR analysis at program branches.
RoadRunner operates on the Java bytecode of analyzed programs, so analysis properties such as \BR soundness and \WBR completeness hold with respect to the execution of the bytecode, even if the JVM compiler optimizes away control or data dependencies.

Our implementation of \BR and \WBR analyses and \checkWBRRace is publicly
available.\footnote{\url{https://github.com/PLaSSticity/SDP-WDP-implementation}}
\mike{TODO: update publicly available implementation with fixed \BR and other changes}

We evaluated \emph{Joana} to perform static analysis for detecting whether a
branch depends on prior reads or not~\cite{joana}, following the system
dependency graphs used in \mbox{MCR-S}~\cite{mcr-s}. We found no practical advantages to
using Joana. In most programs, for the vast majority of the write--read branch dependencies executed, the next branch after the read is dependent on the read according to Joana.
In \bench{pmd} and \bench{sunflow}, static analysis reported many write--read dependencies where the following branch did not depend on the read, but this did not lead to any additional \WBR-races being detected.
It is unclear whether these results are mainly due to properties of the evaluated programs
(\ie, if almost all branches do depend on prior reads) or imprecision of Joana's static analysis.
Our implementation and evaluation do not use static analysis, and instead assume that branches always depend on prior reads.

\paragraph{\BR and \WBR analyses.}

We implemented a single analysis tool within RoadRunner
that can perform \HB, \WCP, \DC, \BR, and \WBR analyses on a single observed execution.
The implementation of \HB, \WCP, and \DC analyses are taken from the Vindicator implementation,
and implementation of \BR and \WBR analyses follows Algorithms~\ref{alg:br-analysis} and \ref{alg:wbr-analysis}.
For thread fork and join (including implicitly forked threads~\cite{vindicator})
and static class initializer edges~\cite{jvm-spec},
each analysis adds appropriate ordering between the two synchronizing events.
The analyses treat calls to \code{m.wait()} as a release of \code{m} followed by an acquire of \code{m}.
The analyses instrument \code{volatile} variable accesses as \emph{atomic accesses} as
described in Section~\ref{subsec:atomic-accesses}.
The analyses can in theory handle lock-free data structures,
such as data structures in \code{java.\allowbreak util.\allowbreak concurrent},
by handling atomic operations as in Section~\ref{subsec:atomic-accesses}.
However, RoadRunner instruments only application code, not Java library code,
and it does not intercept underlying atomic operations (\eg, by instrumenting calls to atomic \code{sun.\allowbreak misc.\allowbreak Unsafe} methods).
The analyses may thus miss some synchronization in the evaluated programs.

The analyses can determine that some observed events are ``redundant'' and cannot affect
the analysis results.
For a read or write event, if the same thread has performed a read or write,
respectively, to the same variable without an intervening synchronization operation, then the access is redundant.
For a branch event, if the same thread has not performed a read event since the last branch event,
then the branch is redundant (since our implementation
assumes that a branch is dependent on all prior reads).
The implementation ``fast path'' detects and filters redundant events,
and does not perform analysis for them.

The implementation is naturally parallel because application threads running in parallel perform analysis.
The implementation uses fine-grained synchronization on metadata to ensure atomicity of the analysis for an event.
For \WBR analysis, to obtain an approximation of \ltTR
(needed by vindication; see
\iftoggle{extended-version}{line~\ref{line:backReorder:pick-latest-event} of Algorithm~\ref{alg:verify-race} in Appendix~\ref{sec:vindication-full}),}
{the extended arXiv version~\cite{depaware-extended-arxiv}),}
the implementation assigns each event node in the constraint graph a Lamport timestamp~\cite{happens-before} that respects \HB order:
$\HBOrdered{\event{e}{}}{\event{e}{'}} \implies \mathit{ts}(\event{e}{}) < \mathit{ts}(\event{e}{'})$.

\paragraph{Handling races.}

To keep finding real races after the first detected race,
whenever an analysis detects a race, it updates vector clocks (and \WBR's constraint graph)
so that the execution so far is race free.
\BR and \WBR analyses treat racing accesses as though minimal critical sections on the same lock protected them,
as described in Section~\ref{sec:analyses:handling-races}.
\HB, \WCP, and \DC analyses handle detected races
by adding ordering between all accesses.

If an analysis detects multiple races involving the current access,
it reports only one of the races---the best-estimate ``shortest'' race according to the Lamport timestamps---but adds ordering to eliminate all of the races.

\notes{
The \WBR analysis algorithm, in line~\ref{line:wbr-analysis:wr-rd-detected}, computes the
sum of release times for every common lock that was held for the write and read
accesses. However, note that all these locks were held by the same thread during
the last write. As a result, the release time of the outermost lock will be
greater than the release time of any inner locks. Thus it is sufficient to find
the set of common locks, and then pick the release time of the one that was the
outermost during the write.
}

\later{
\paragraph{Detecting dependencies}

Explain how we use the system dependency graph.
}

\paragraph{Vindication.}

\WBR analysis constructs a constraint graph representing the observed execution's \WBR ordering.
When the execution completes,
the implementation calls \checkWBRRace on a configurable subset of the \WBR-races,
\eg, each \WBR-race that is not also a \BR-race.

\notes{
We implemented \checkWBRRace based on the algorithm presented in Figure~\ref{alg:verify-race}
\iftoggle{extended-version}{in Appendix~\ref{sec:vindication-full}.}{in the extended arXiv version~\cite{dp-extended-arxiv}.}
The \checkWBRRace extends the implementation of \checkDCRace,
which uses optimizations including limiting computation of constraints to a window of events
that initially includes events between the racing accesses and may grow as constraints are added~\cite{vindicator}.
\CheckWBRRace soundly extends this optimization to construct a predictable trace only for events in the window,
since the original execution before the window is guaranteed to be a well-formed prefix for the reordered execution.

\notes{
\kaan{Removed the old explanation since we no longer explain these details of the algorithm to begin with.}
We implemented \checkWBRRace based on Algorithm~\ref{alg:verify-race}.
When computing \LW and \LS constraints,
\textsc{AddConstraints} can generally safely consider only events within a window of events between \event{e}{_1} and \event{e}{_2} in \ltTR order
(conservatively approximated using events' Lamport timestamps).
The implementation grows the window as it adds constraints outside of the current window.
Intuitively, it it is sufficient to find \LW constraints
within the window because, for any earlier event, the latest-first ordering heuristic
(line~\ref{line:backReorder:pick-latest-event} of Algorithm~\ref{alg:verify-race})
ensures the algorithm will not attempt to violate \LW outside of the window.
}

Although \checkWBRRace is sound by design, as a sanity check
we developed an \emph{independent checker} that verifies that \trPrime is a valid predictable trace.
This checker operates independently of \WBR analysis and \checkWBRRace and does not rely on the correctness of the constraint graph to check
that \trPrime is a predictable trace of \tr.
The checker---which was useful for finding bugs during implementation---did not report any invalid traces in the following experiments.
}

\subsection{Methodology}
\label{subsec:methodology}



\later{
\kaan{Reviewer B says ``The evaluation uses only an old, non-representative set of programs. While DaCapo has been commonly used in prior work, and it's good to compare the results on it, the evaluation would be much more convincing if the new and old analyses were run also on some more recent programs with real inputs, ideally finding races in the current program versions.''
\mike{Yeah that'd be nice to pick one program, and RoadRunner might work out of the box.
(I think DaCapo's custom class loading is the main reason it needs special handling in RoadRunner.)}
\mike{TODO: I think this is worth doing. Pick a widely used Java program (we can discuss) and run it with your analyses?}}
}

The experiments execute large, real Java programs harnessed as
the DaCapo benchmarks~\cite{dacapo-benchmarks-conf}, version 9.12-bach.
We use a version of the DaCapo programs that the RoadRunner authors have modified to work with
RoadRunner;\footnote{\url{https://github.com/stephenfreund/RoadRunner/releases/tag/v0.5}}
the resulting workloads are approximately equal to DaCapo's default workload.
The experiments exclude DaCapo programs \bench{eclipse}, \bench{tradebeans}, and \bench{tradesoap},
which the RoadRunner authors have not modified to run with RoadRunner;
\bench{jython}, which failed to run with RoadRunner in our environment;
and the single-threaded program \bench{fop}.

The experiments execute on a quiet Intel Xeon
E5-4620 with four 8-core processors with hyperthreading disabled and
256 GB of main memory, running Linux 3.10.0.
We execute RoadRunner with the HotSpot 1.8.0 JVM and set the maximum heap size to 245 GB.

We run various combinations of the analyses to collect race results and statistics and measure performance.
To account for run-to-run variation,
each reported result is the
mean of five trials.

Each \WBR-race in an execution is a \emph{dynamic} \WBR-race (similarly for \BR-, \DC-, \WCP-, and \HB-races).
Among dynamic \WBR-races,
some may be detected at the same static accesses.
If two dynamic \WBR-races have the same two static source location regardless of order,
then they are the same \emph{static} \WBR-race (similarly for \BR-, \DC-, \WCP-, and \HB-races).

\subsection{Dynamic Characteristics}

\newcommand{\M}[1]{#1~M}
\newcommand\colname[1]{\emph{#1}}
\newcommand{\slowdown}[1]{#1$\;\!\times$}

Table~\ref{tab:events} shows properties of the analyzed programs.
The
\colname{\#Thr} column reports total threads created by an execution and,
in parentheses, threads active at termination.
The rest of the columns count events from \WBR analysis; other analyses are similar but exclude branch events.
\colname{Total events} are all executed events instrumented by the analysis.

\begin{table}
\newcommand{\slowpath}[1]{\M{#1}}
\newcommand{\maxlive}[1]{(#1)}

\small
\centering
\begin{tabular}{@{}l|rr|r|r@{\quad(}rrrr@{\,)}H}
        & &                   & Total  & \mc{5}{c}{Analyzed events} & \\
& \mc{2}{c|}{\#Thr} & events & All & \code{acq}/\code{rel} & \code{wr} & \code{rd} & \BrUniv & Other \\\hline
\bench{avrora} & 7 & (7) & \M{2,400} & \M{260}&1.2\% & 17.7\% & 42.2\% & 38.4\% \\
\bench{batik} & 7 & (7) & \M{490} & \M{17}&0.6\% & 26.3\% & 38.4\% & 34.0\% \\
\bench{h2} & 34 & (33) & \M{9,368} & \M{768}&0.5\% & 17.1\% & 43.1\% & 39.1\% \\
\bench{luindex} & 3 & (3) & \M{910} & \M{72}&0.6\% & 20.1\% & 42.4\% & 36.9\% \\
\bench{lusearch} & 34 & (34) & \M{2,746} & \M{301}&0.9\% & 19.5\% & 43.6\% & 35.6\% \\
\bench{pmd} & 33 & (33) & \M{403} & \M{41}&$<\,$0.1\% & 28.5\% & 37.3\% & 34.2\% \\
\bench{sunflow} & 65 & (33) & \M{14,452} & \M{887}&$<\,$0.1\% & 44.7\% & 41.4\% & 13.8\% \\
\bench{tomcat} & 106 & (67) & \M{113} & \M{29}&2.8\% & 18.7\% & 42.1\% & 36.1\% \\
\bench{xalan} & 33 & (33) & \M{1,306} & \M{436}&2.1\% & 12.0\% & 48.8\% & 37.1\% \\
\end{tabular}

\later{
\mike{The maxlive number for tomcat (10) is different from the Vindicator paper (22). Weird.}
}



\caption{Dynamic characteristics of the analyzed programs.
Event counts (shown in millions) and percentages are collected from \WBR analysis; other analyses do not analyze branch events.}

\label{tab:events}

\end{table}

\colname{Analyzed events} are the events \emph{not} filtered by the fast path that detects redundant events.
The rest of the columns show the breakdown of analyzed events
by event type.
The percentages do not add up to 100\% because they do not include other events
(\eg, fork, join, \code{wait}, \code{volatile} access, and static class initializer events),
which are always less than 1\% of analyzed events.
Unsurprisingly, most analyzed events are memory accesses or branches.

\subsection{Race Detection Effectiveness}



\later{
\kaan{Reviewer C asks ``What were the characteristics of tomcat and xalan that caused such
a large spike between DC and DP (unlike the other benchmarks which
seemed to be very comparable)?''
\mike{TODO: Instead of trying to answer this directly,
(1) perhaps we should discuss the \code{pmd} race here (not just in the motivation), and/or
(2) try out a non-benchmarked program and (assuming it has ``new'' races), explain the new races.}}
}

\later{
\kaan{Reviewer A says ``Evaluation methodology compares against other data race detectors,
   but does not compare against schedule exploration paired with race detection.''
\mike{It's worth thinking about trying to do this in the future.
A simpler evaluation for now would be to measure the efficacy of multiple runs with \HB compared with one run with \WBR.}
\mike{Or: do this investigation with the evalauted non-benchmarked program.
If it's not benchmarked/harnessed, then perhaps some kind of evaluation of analyses across multiple runs would be more interesting?}}
}


Table~\ref{tab:races-counts} reports detected races for two different experiments that each run a combination of analyses
on the same executions.
Table~\ref{tab:races-counts:br}'s results are from an experiment that runs
\HB, \WCP, and \BR analyses together on the same executions, to compare these analyses' race detection capabilities directly.
Likewise, a separate experiment runs \DC and \WBR analyses together on the same executions to make them directly comparable,
resulting in Table~\ref{tab:races-counts:wbr}'s results.

For each race count, the first value is static races,
followed by dynamic races in parentheses.
For example, on average over the five trials,
the analysis detects about 406,000 \WBR-races for \bench{avrora},
which each correspond to one of 5 different unordered pairs of static program locations.

\begin{table*}
\centering
\small
\newcommand{\none}{\mc{1}{r@{$\;\;\textcolor{white}{\rightarrow}\;\;$}}{0}}
\newcommand{\K}[2]{#2$\;\!$K}
\newcommand{\MM}[2]{#2$\;\!$M}

\captionsetup{position=bottom}

\subfloat[\HB, \WCP, and \BR analyses on the same executions.]{
\begin{tabular}{@{}l|rr|rr|rr@{}ZZZZ}
Program & \mc{2}{c|}{\HB-races} & \mc{2}{c|}{\WCP-races} & \mc{2}{c@{}}{\BR-races} & \mc{2}{Z}{\WDC-races} & \mc{2}{Z}{\WBR-races} \\
\hline
\bench{avrora} & 5 & (\K{202034}{202}) & 5 & (\K{202960}{203}) & 5 & (\K{202963}{203}) &  \\
\bench{batik} & 0 & (0) & 0 & (0) & 0 & (0) &  \\
\bench{h2} & 12 & (\K{52704}{53}) & 12 & (\K{52788}{53}) & 12 & (\K{53538}{54}) &  \\
\bench{luindex} & 1 & (1) & 1 & (1) & 1 & (1) &  \\
\bench{lusearch} & 0 & (0) & 0 & (0) & 0 & (0) &  \\
\bench{pmd} & 8 & (436) & 8 & (443) & 10 & (651) &  \\
\bench{sunflow} & 2 & (20) & 2 & (26) & 2 & (26) &  \\
\bench{tomcat} & 98 & (\K{35935}{36}) & 99 & (\K{35956}{36}) & 103 & (\K{38752}{39}) &  \\
\bench{xalan} & 7 & (208) & 15 & (\K{531236}{531}) & 37 & (\K{2072070}{2072}) &  \\

\end{tabular}%
\label{tab:races-counts:br}}
\hfill
\subfloat[\DC and \WBR analyses on the same executions.]{
\begin{tabular}{@{}l|HHHHHHrr|rr@{}}
Program & \mc{2}{Z}{\HB-races} & \mc{2}{Z}{\WCP-races} & \mc{2}{Z}{\BR-races} & \mc{2}{c|}{\WDC-races} & \mc{2}{c@{}}{\WBR-races} \\
\hline
\bench{avrora} & 5 & (\K{204514}{205}) & 5 & (\K{206374}{206}) & 5 & (\K{206375}{206}) & 5 & (\K{203497}{203}) & 5 & (\K{406047}{406}) \\
\bench{batik} & 0 & (0) & 0 & (0) & 0 & (0) & 0 & (0) & 0 & (0) \\
\bench{h2} & 9 & (\K{51683}{52}) & 9 & (\K{51782}{52}) & 9 & (\K{53237}{53}) & 11 & (\K{53768}{54}) & 12 & (\K{62526}{63}) \\
\bench{luindex} & 1 & (1) & 1 & (1) & 1 & (1) & 1 & (1) & 1 & (1) \\
\bench{lusearch} & 0 & (0) & 0 & (0) & 0 & (0) & 0 & (0) & 1 & (30) \\
\bench{pmd} & 6 & (351) & 6 & (354) & 8 & (562) & 9 & (\K{1678}{2}) & 10 & (\K{3183}{3}) \\
\bench{sunflow} & 2 & (19) & 2 & (25) & 2 & (25) & 2 & (49) & 2 & (100) \\
\bench{tomcat} & 85 & (\K{33885}{34}) & 86 & (\K{33911}{34}) & 91 & (\K{37812}{38}) & 94 & (\K{35875}{36}) & 284 & (\K{125078}{125}) \\
\bench{xalan} & 6 & (203) & 21 & (\K{520406}{520}) & 52 & (\MM{2187865}{2.2}) & 17 & (\K{648507}{649}) & 170 & (\MM{15324942}{15}) \\
\end{tabular}%
\label{tab:races-counts:wbr}}

\caption{Static and dynamic (in parentheses) race counts from two different experiments.\label{tab:races-counts}}

\end{table*}



Table~\ref{tab:races-counts:br} shows that
\BR analysis finds significantly more races than not only \HB analysis but also \WCP analysis---the state of the art in
unbounded sound predictive race detection~\cite{fasttrack2,wcp}.
These additional races are due to \BR incorporating data dependence more precisely than \WCP
by not ordering write--write conflicting critical sections,
essentially permitting predictable traces that swap writes without changing a causal read's last writer.


Likewise, Table~\ref{tab:races-counts:wbr} shows that
\WBR analysis finds more races than \DC analysis,
the state of the art in high-coverage unbounded predictive race detection~\cite{vindicator}.
These additional races result from \WBR being more precise with respect to both data and control dependence
than \DC, and in fact being complete.

The counts of \HB-, \WCP-, and \WDC-races we
report here are significantly different from those reported by
the Vindicator paper~\cite{vindicator}.
(While the counts are not directly comparable, both papers show similar trends between relations.)
The most significant cause of this effect is that
RoadRunner stops tracking a field after the field has 100 races,
a behavior that Vindicator used but that we disabled for these results to avoid artificially underreporting race counts.
Furthermore, our analyses do not use a Vindicator optimization
that merges events, reducing the number of races reported when there are multiple races between synchronization-free regions.
We disabled this optimization because \WBR analysis must track variable access information for each event, negating the advantages of this optimization.
Another difference is that the Vindicator experiments spawned fewer threads for some benchmarks by setting the number of available cores to 8.


Table~\ref{tab:races-vindication} reports results from an experiment that
runs \BR and \WBR analyses together and then performs \checkWBRRace on \emph{\WBR-only races},
which are \WBR-races that are not also \BR-races.
The \colname{\BR-races} and \colname{\WBR-races} columns report static and dynamic races,
as in Table~\ref{tab:races-counts}.
The \colname{\WBR-only} column is \emph{static} \WBR-only races, which are static \WBR-races that have no dynamic instances that are \BR-races.
The last column, \emph{\WBR-only $\rightarrow$ Verified}, reports
how many static \WBR-only races are detected
and how
many are successfully vindicated as true races by \checkWBRRace.
In this experiment,
the implementation tries to vindicate up to 10 dynamic instances of each static \WBR-only race.
The implementation first attempts to vindicate the five earliest dynamic instances of a static \WBR-only race,
then five random dynamic instances, stopping as soon as it verifies any dynamic instance of the static race.

\begin{table*}
\centering
\small
\newcommand{\none}{\mc{1}{r@{$\;\;\textcolor{white}{\rightarrow}\;\;$}}{0}}
\newcommand{\K}[2]{#2$\;\!$K}
\newcommand{\MM}[2]{#2$\;\!$M}

\begin{tabular}{@{}l|rr|rr|r@{$\;\;\rightarrow\;\;$}l@{}}
    Program & \mc{2}{c|}{\BR-races} & \mc{2}{c|}{\WBR-races} & \WBR-only & Verified \\
    \hline
    \bench{avrora} & 5 & (\K{201613}{202}) & 5 & (\K{406567}{407}) & \none \\
    \bench{batik} & 0 & (0) & 0 & (0) & \none \\
    \bench{h2} & 12 & (\K{53431}{53}) & 12 & (\K{62606}{63}) & 1 &  0  \\
    \bench{luindex} & 1 & (1) & 1 & (1) & \none \\
    \bench{lusearch} & 0 & (0) & 1 & (30) & 1 &  0  \\
    \bench{pmd} & 9 & (470) & 11 & (\K{2944}{3}) & 1 & 1 \\
    \bench{sunflow} & 2 & (38) & 2 & (100) & \none \\
    \bench{tomcat} & 99 & (\K{37642}{38}) & 321 & (\K{125743}{126}) & 222 & 53 \\
    \bench{xalan} & 33 & (\MM{1898266}{1.8}) & 170 & (\MM{15277795}{15.2}) & 137 & 135 \\
\end{tabular}
\caption{Static and dynamic (in parentheses) race counts from an experiment running
\BR and \WBR analyses together and vindicating dynamic instances of static \WBR-only races.
The \colname{\WBR-only $\rightarrow$ Verified} column reports static \WBR-only races,
followed by how many static \WBR-only races were verified as predictable races by \checkWBRRace.\label{tab:races-vindication}}

\newcommand{\stdev}[1]{#1}

\small
\begin{tabular}{@{}l|r@{$\;\;\pm\;$}r|r@{}}
& Mean & \stdev{Stdev} & Max \\
\hline
\bench{pmd} & 805,000 & \stdev{1,750,000} & 4,721,218 \\
\bench{tomcat} & 4,560,000 & \stdev{6,070,000} & 28,750,242 \\
\bench{xalan} & 50,900 & \stdev{96,800} & 1,170,850 \\
\end{tabular}

\caption{Characteristics of the distribution of event distances of \WBR-only races that are verified predictable races.
The table rounds the mean and standard deviation to three significant digits.}
\label{tab:race-distances}

\notes{
\kaan{I also added a commented-out version with window sizes (i.e. how many events were in the reordered trace for verified races, when window-only vindication is turned on). I'm not sure if it makes sense to add that table to the paper, but it better describes why h2 takes a large amount of time.
\mike{We mention this information in the Performance subsection.}}
\begin{tabular}{@{}l|r@{$\;\;\pm\;$}r|r||r@{$\;\;\pm\;$}r|r@{}}
\small
        & \mc{3}{c||}{Race Distance} & \mc{3}{c}{Window Size} \\
Program & Mean & \stdev{Stdev} & Max & Mean & \stdev{Stdev} & Max \\
\hline
pmd & 65873 & \stdev{73172} & 203713 & 2494068 & \stdev{519869} & 3163329 \\
tomcat & 5463200 & \stdev{5622417} & 29417813 & 12325413 & \stdev{7135659} & 29435885 \\
xalan & 52903 & \stdev{90384} & 625293 & 427008 & \stdev{774949} & 5508859 \\
\end{tabular}
\kaan{Out of date!}
}

\end{table*}
  
Over half of the static \WBR-only races are verified predictable races: out of 359 static \WBR-only races on average,
189 are verified predictable races.
As Section~\ref{Sec:eval:perf} shows, it can take a few minutes for \CheckWBRRace to check a \WBR-race.
Given the difficulty and importance of detecting unknown, hard-to-expose data races in mature software---and
the amount of time developers currently spend on testing and debugging---the time for \checkWBRRace is reasonable.

\notes{
\kaan{There are still some differences between both tables, although I'm not sure why.
\mike{A quick look at the results suggests that the run-to-run variation may be a plausible explanation.
For example, for \bench{tomcat}'s static \WBR-races between the two configurations, the confidence intervals overlap.}}
}



\notes{
\kaan{Removed that \WBR never fails for acyclic graphs, since we never discuss AddConstraints within the main body of the paper anyway. (This is still true if we want to add this back in.)
\mike{That seems reasonable.}}
}

We confirmed that in all of the experiments, \BR analysis detected
every race detected by \WCP analysis, and \WBR analysis detected every race detected by \DC or \BR analysis.
For sanity, we also successfully vindicated every reported \BR-only race.

\paragraph{Race characteristics.}

SMT-solver-based predictive race detectors can be as precise
as \BR and \WBR analyses, but cannot scale to unbounded program
executions~\cite{rvpredict-pldi-2014,said-nfm-2011,rdit-oopsla-2016,jpredictor,maximal-causal-models,ipa}
(Section~\ref{sec:related}).
These approaches typically analyze bounded windows of an execution trace,
missing races involving ``far apart'' events.
We can estimate whether SMT-based approaches would miss a predictable race
by computing the race's \emph{event distance},
which is the number of events in the execution trace between the race's two events.
Since our implementation does not compute a total order of events,
it approximates event distance using Lamport timestamps:
event distance is the number of events \event{e}{} such that
$\mathit{ts}(\event{e}{_1}) < \mathit{ts}(\event{e}{}) < \mathit{ts}(\event{e}{_2})$.

Table~\ref{tab:race-distances} reports the distribution of event distances between accesses in
each successfully vindicated \WBR-only race (\ie, the last column of Table~\ref{tab:races-vindication}). The average
distance and standard deviation are across all trials.
The last column reports the greatest distance found among all trials.

\subsection{Performance}
\label{Sec:eval:perf}


\later{
\kaan{Reviewer A says ``Whether using explicit schedule exploration or not, how does working
harder to extrapolate from *one* trace compare with running a weaker
tool (HB, or WCP, or DC), but running it on traces sampling multiple
schedules?''
\mike{Related to the previous comment about schedule exploration,
we could answer this question with our existing experiments,
by looking at the average difference between one trial of HB (or WCP or DC) and multiple trials of HB in terms of races found.}}
}

Table~\ref{tab:performance} reports the run-time performance of various combinations of analyses.
\colname{Base} is native execution time without any instrumentation.
\notes{but using RoadRunner's harness functionality to ensure that the workloads match
the workloads used in the analyses.}%
Other columns (excluding \colname{Failed} and \colname{Verified}) are
slowdowns relative to \colname{Base}.

The \colname{Instr.\ only} columns are
RoadRunner configurations that instrument events
(excluding or including branches) but perform no analysis in the instrumentation.

The \colname{Analyses w/o constraint graph} show configurations that do not construct a constraint graph.
Only configurations including \WBR analysis instrument branch events.
The \colname{\WCP} and \colname{\BR} columns show the slowdowns
from running the \WCP and \BR analyses independently;
the performance difference between them is usually modest,
suggesting that there is minimal performance penalty from using \BR analysis over \WCP analysis.

\begin{table*}
\newcommand{\seconds}[1]{#1~s}

\small

\begin{tabular}{@{}l|r|rr|rrrr|rrr@{}}
        &      & \mc{2}{c|}{Instr. only} & \mc{4}{c|}{Analyses w/o constraint graph} & \mc{3}{c}{\BR{}+\WBR{}+graph} \\
        & Base & w/o br & w/ br & \WCP & \BR & \BR{}+\WDC & \BR{}+\WBR & Slowdown & Failed & Verified \\
\hline
\bench{avrora} & \seconds{7.7} & \slowdown{2.1} & \slowdown{2.6} & \slowdown{17} & \slowdown{20} & \slowdown{28} & \slowdown{34} & \slowdown{42} & -  & -  \\
\bench{batik} & \seconds{3.3} & \slowdown{3.9} & \slowdown{5.5} & \slowdown{19} & \slowdown{20} & \slowdown{24} & \slowdown{33} & \slowdown{32} & -  & -  \\
\bench{h2} & \seconds{9.4} & \slowdown{6.6} & \slowdown{9.3} & \slowdown{129} & \slowdown{139} & \slowdown{202} & \slowdown{212} & \slowdown{280}& \seconds{233} & -  \\
\bench{luindex} & \seconds{1.4} & \slowdown{5.9} & \slowdown{11} & \slowdown{91} & \slowdown{93} & \slowdown{118} & \slowdown{139} & \slowdown{154} & -  & -  \\
\bench{lusearch} & \seconds{3.8} & \slowdown{4.4} & \slowdown{5.2} & \slowdown{20} & \slowdown{23} & \slowdown{29} & \slowdown{32} & \slowdown{42}& \seconds{< 0.1} & -  \\
\bench{pmd} & \seconds{2.6} & \slowdown{7.6} & \slowdown{10.0} & \slowdown{27} & \slowdown{29} & \slowdown{32} & \slowdown{36} & \slowdown{36} & - & \seconds{8.6} \\
\bench{sunflow} & \seconds{2.3} & \slowdown{11} & \slowdown{13} & \slowdown{142} & \slowdown{196} & \slowdown{210} & \slowdown{221} & \slowdown{264} & -  & -  \\
\bench{tomcat} & \seconds{1.6} & \slowdown{5.9} & \slowdown{6.2} & \slowdown{30} & \slowdown{37} & \slowdown{55} & \slowdown{59} & \slowdown{56}& \seconds{2.0}& \seconds{55} \\
\bench{xalan} & \seconds{5.3} & \slowdown{2.6} & \slowdown{3.3} & \slowdown{40} & \slowdown{46} & \slowdown{63} & \slowdown{86} & \slowdown{124}& \seconds{30}& \seconds{0.1} \\

\end{tabular}

\caption{Slowdowns of program instrumentation and various analyses over uninstrumented execution, and the average time taken to vindicate \WBR-only races.}
\label{tab:performance}

\end{table*}

\colname{\BR{}+\DC{}} represents performing \BR and \DC analyses together.
We run \DC analysis with \BR analysis to minimize \DC-races that need vindication.
(Vindicator combined \DC analysis with \WCP analysis for this purpose~\cite{vindicator}, but \BR analysis is more powerful.)


\colname{\BR{}+\WBR{}+graph} represents the canonical use case for \WBR analysis.
This configuration performs \BR and \WBR analyses
and constructs the constraint graph to enable vindication.
It uses \BR to reduce how many \WBR-races need vindication.
For comparison purposes, \colname{\BR{}+\WBR{}} forgoes constructing the constraint graph,
showing the cost of constructing the graph, which we have not optimized.
\colname{\BR{}+\WBR{}} is slower than \colname{\BR{}+\DC{}}
because \WBR analysis is generally more complex than \DC analysis.

Finally, \colname{Failed} and \colname{Verified} are the average times taken for each dynamic race
that \checkWBRRace fails to verify or successfully verifies, respectively.
Vindication times vary significantly across programs;
vindication is particularly slow for
\bench{tomcat} because most of its racing accesses are separated by millions of events (Table~\ref{tab:race-distances}).
Vindication is slow for \bench{h2}, even though its races are not far apart,
because \checkWBRRace discovers new critical section constraints
that require analyzing over 500 million events.


\subsection{Summary and Discussion}

Our \BR- and \WBR-based approaches are slower than other predictive approaches,
but they find more races, some of which are millions of events apart.
SMT-based approaches would not be able to find these far-apart races because they
cannot scale past analyzing bounded windows of
executions~\cite{rvpredict-pldi-2014,said-nfm-2011,rdit-oopsla-2016,jpredictor,maximal-causal-models,ipa}
(Section~\ref{sec:related}).
Notably, \emph{RVPredict}, which (like \WBR) incorporates precise control and data dependence,
uses an analysis window of 10,000 events~\cite{rvpredict-pldi-2014},
meaning it would miss many of the predictable races detected and verified by our approach.

Our evaluation demonstrates the power of \BR and \WBR analyses to find more races than prior approaches \emph{in a single execution}.
A potential limitation of the evaluation is that it does not compare our analyses
with approaches that perform \HB analysis on multiple executions
(\eg, using one of the many schedule exploration approaches; Section~\ref{sec:related}).
Our work's aim is to push the limit on what can be found in a single execution,
which is essentially complementary to approaches that explore multiple schedules.
No other known sound technique could have predicted all of these races from the observed executions.
\notes{\CheckWBRRace is more expensive on average than \checkDCRace
because only \checkWBRRace must detect causal events and check last writers,
which requires maintaining the variable accessed by each read and write,
precluding optimizations such as merging event nodes~\cite{vindicator}.}%

\section{Related Work}
\label{sec:related}

This section describes and compares with prior work,
starting with the most closely related work.

\paragraph{Unbounded predictive analysis.}

Prior work introduces unbounded predictive analyses, \wcpFull (\WCP) and \dcFull (\DC) analyses~\cite{wcp,vindicator},
which Sections~\ref{sec:background} and \ref{sec:evaluation} covered and evaluated in detail.
\BR and \WBR analyses predict more races in real programs than \WCP and \DC analyses, respectively (Section~\ref{sec:evaluation}).



The \WCP relation is weaker (\ie, detects more races) than Smaragdakis \etal's earlier \emph{\cpFull (CP)} relation~\cite{causally-precedes}.
Smaragdakis \etal's implementation detects races within bounded windows of 500 events
because of the difficulty of developing an efficient unbounded analysis for \CP~\cite{causally-precedes,raptor}.

Recent work introduces the \emph{afterward-confirm (AC)} relation and an approach called \emph{DigHR} for computing AC~\cite{dighr}.
AC is the same as \emph{\CP} except that it removes write--write conflicting critical section ordering.
Despite this similarity with our work, the contributions differ significantly, and
the DigHR work has major correctness issues.
Foremost, the DigHR paper claims incorrectly that AC is sound.
AC is, to the best of our understanding, unsound: its removal of write--write ordering leads to detecting
false races, including for the execution in Figure~\ref{fig:wr-wr-conflict:extra}
(with the \BrUniv event omitted; DigHR's event model does not include branches).
The DigHR paper provides a soundness proof, which we believe
is incorrect as a result of informality leading to not covering cases such as Figure~\ref{fig:wr-wr-conflict:extra}.
In contrast with DigHR, our work introduces a sound relaxation of \WCP (\BR).
Additionally, our work introduces a complete relation (\WBR),
handles control dependencies (\BrUniv events), and
presents linear-time analyses for \BR and \WBR
(DigHR is superlinear, like existing \CP analyses~\cite{causally-precedes,raptor}).

Concurrently with our work,
Pavlogiannis introduces a predictive race detection approach called \emph{M2} that is
related to vindication~\cite{pavlogiannis-2019}.
Pavlogiannis uses lockset analysis as an imprecise filter for potential races checked by M2,
while our work introduces linear-time \WBR analysis as a less-imprecise filter for potential races checked by \checkWBRRace.
Although Pavlogiannis reports performance that is sometimes competitive with
the performance of \HB, \WCP, and \DC analyses, Pavlogiannis's implementations of these
analyses perform extra passes over execution traces in addition to the
efficient single-pass vector-clock-based analyses from prior work~\cite{wcp,vindicator}.
It is unclear to us how M2 and \checkWBRRace would compare in terms of detection capability
(aside from the fact that only \checkWBRRace takes branches and control dependence into account).
In addition to these differences, our work incorporates branches and control dependence sensitivity,
while Pavlogiannis's work does not and thus would miss races
such as Figures~\ref{fig:CORE:has-predictable-race} and~\ref{fig:source-pmd}; and
our work introduces a sound partial order and linear-time analysis (\BR and \BR analysis).

Our concurrent work introduces the \emph{SmartTrack} algorithm,
which optimizes the performance of \WCP and \DC analyses~\cite{smarttrack}.
SmartTrack's optimizations apply to analyses that compute predictive relations that order all pairs of conflicting accesses---a
property that \BR and \WBR do not conform to.
In any case, optimizing \WBR analysis would have limited benefit because the analysis still must construct a constraint
graph in order to perform vindication (a necessity considering that so many \WBR-races fail vindication in practice).
The SmartTrack paper also introduces a new relation \emph{weak-doesn't-commute (WDC)} that is a weak variant of \DC~\cite{smarttrack}.
Unlike \BR and \WBR, WDC does \emph{not} help find more races than \DC, but rather serves to improve analysis performance.

\paragraph{Bounded predictive approaches.}

Other approaches predict data races by generating and solving satisfiability modulo theories (SMT)
constraints~\cite{rvpredict-pldi-2014,said-nfm-2011,rdit-oopsla-2016,jpredictor,maximal-causal-models,ipa}.
These SMT-based approaches cannot analyze long execution traces in reasonable time,
because constraints are quadratic or cubic in trace length, and constraint-solving time
often grows superlinearly with constraints.
These approaches thus break traces into bounded ``windows'' of traces (\eg, 500--10,000 events each),
missing predictable races involving accesses that occur further apart in the trace.

One advantage of SMT-based approaches is that they can be both sound and complete (within a bounded window) by encoding precise constraints.
Notably, \emph{RVPredict} includes branches in its execution model~\cite{rvpredict-pldi-2014}. (RVPredict
also incorporates \emph{values} into its event model, so a read in a predictable trace can have a different last
writer as long as it writes the same value~\cite{rvpredict-pldi-2014}.)
RVPredict is thus complete by Section~\ref{subsec:completeness}'s definition,
except that it is practically limited to bounded windows of execution.
\WBR analysis, on the other hand, is complete without the windowing limitation,
but \checkWBRRace is not guaranteed to vindicate a \WBR-race even when a predictable race exists.

\paragraph{Schedule exploration.}

In contrast to predictive analysis,
\emph{schedule exploration} approaches
execute the program multiple times to explore more program
behaviors~\cite{mcr,mcr-s,chess,randomized-scheduler,racageddon,racefuzzer,drfinder,racedet-model-checking-2004}.
These approaches may be systematic (often called \emph{model checking}) or be based on randomness or heuristics.
Schedule exploration is complementary to
predictive analysis, which aims to glean as much as possible from a given execution.
\emph{Maximal causality reduction (MCR)} combines schedule exploration with
predictive analysis~\cite{mcr,mcr-s}.
\emph{MCR-S} incorporates static control flow information
to reduce the complexity of MCR's generated SMT constraints~\cite{mcr-s}.

\paragraph{Other analyses.}

Widely used data race detectors typically use dynamic \emph{\hbFull (\HB)}
analysis~\cite{happens-before,fasttrack,multirace,goldilocks-pldi-2007,google-tsan-v1,google-tsan-v2,intel-inspector}.
\HB analysis
cannot predict races involving reordered critical sections on the same lock.
The detected races thus depend heavily on the scheduling of the analyzed program.
Other analyses find a subset of \HB-races by detecting simultaneously executing conflicting accesses or
regions~\cite{frost,valor-oopsla-2015,ifrit,datacollider,racefuzzer,racechaser-caper-cc-2017}.

\emph{Lockset} analysis checks a locking discipline, ensuring that all
pairs of conflicting accesses hold some common lock~\cite{eraser,
racedet-escape-nishiyama-2004,
choi-racedet-2002, object-racedet-2001}.
Lockset analysis is predictive but unsound,
reporting false races for synchronization patterns other than the locking discipline.
\emph{Hybrid} lockset--\HB analyses generally incur disadvantages of one or both kinds of
analysis~\cite{dinning-schonberg,ocallahan-hybrid-racedet-2003,racetrack,multirace}.
  %
  %
  %
  %

\emph{Sampling-based} analysis trades coverage for performance
(opposite of predictive analysis)
in order to detect data races in
production~\cite{pacer-2010,datacollider,literace,racemob,racechaser-caper-cc-2017,racez,prorace}.
Custom hardware support can detect data races with low performance cost
but has not been implemented~\cite{radish,zhou-hard,lard,clean-isca-2015,parsnip,conflict-exceptions,drfx-2010,drfx-2011}.

Dynamic analysis can estimate the likely harm of a data
race~\cite{boehm-miscompile-hotpar-11,portend-toplas15,benign-races-2007,prescient-memory,adversarial-memory,relaxer},
which is orthogonal to detection.
All data races are erroneous under
language memory models that ascribe them undefined
semantics~\cite{memory-models-cacm-2010,c++-memory-model-2008,java-memory-model,
out-of-thin-air-MSPC14,data-races-are-pure-evil,you-dont-know-jack,jmm-broken}. Java's
memory model defines weak semantics for data races~\cite{java-memory-model},
but inadvertently prohibits common JVM optimizations~\cite{jmm-broken,out-of-thin-air-MSPC14}.


\emph{Static} program analysis can detect all races across all feasible executions of a
program~\cite{praun-gross-pldi-2003, naik-static-racedet-2006, naik-static-racedet-2007, locksmith, racerx, relay-2007}, but
it reports thousands of false races for real programs~\cite{racechaser-caper-cc-2017,chimera}.

\paragraph{Avoiding or tolerating data races.}

New languages and type systems can ensure data race freedom,
but require significant programmer effort~\cite{dpj,
boyapati-ownership-types-2002,jade,rust,abadi-type-racedet-2006,flanagan-inference-racedet-2007}.
Compilers and hardware can provide well-defined semantics for data races,
but incur high run-time costs or hardware
complexity~\cite{clean-isca-2015,conflict-exceptions,drfx-2010,drfx-2011,SC-preserving-compiler,endtoendSC,
Sura-high-performance-SC,enforser-asplos-2015,bulk-compiler-micro-2009,region-serializability-for-all}.

\section{Conclusion}

\BR and \WBR analyses improve over existing predictive analyses
by incorporating precise notions of data and control dependence, finding more races both in theory and in practice
while retaining linear (in trace length) run time and thus unbounded analysis.
\BR analysis maintains \WCP analysis's soundness while increasing race coverage.
\WBR analysis finds all data races that can be predicted from an observed execution;
not all \WBR-races are predictable races, but \checkWBRRace can efficiently filter false races.
Experiments show that our new approaches find many
predictable races in real programs that prior approaches are unable to find.
These properties and results suggest that our contributions advance the state of the art in
predictive race detection analysis.

\begin{acks}

We thank Rob LaTour for early help with this project.
Thanks to Steve Freund for making RoadRunner publicly available and providing help with using and modifying it.
Thanks to Andreas Pavlogiannis, Umang Mathur, and Mahesh Viswanathan
for providing the execution in Figure~\ref{fig:wr-wr-conflict:counterexample},
which is a soundness counterexample for a previous version of \BR.
Thanks to the anonymous reviewers for their thorough and insightful feedback.

\end{acks}



\newcommand{\showDOI}[1]{\unskip}
\bibliography{bib/conf-abbrv,bib/plass}

\iftoggle{extended-version}{

\appendix

\section{Proof of \WBR Completeness Helper Lemma}
\label{sec:completeness-lemma}

\begin{proof}[Proof of Lemma~\ref{lem:wbr-necessary-ordering}]
  We proceed by induction on the \emph{\WBRdist} of two events,
  \dist{\event{e}{}}{\event{e}{'}},
  defined as follows (the \WBR properties refer to Table~\ref{tab:partial-order-definitions}):
\begin{align*}
\dist{\event{e}{}}{\event{e}{'}} = \min
\begin{cases}
0 & \textnormal{ if } \WBROrdered{\event{e}{}}{\event{e}{'}} \textnormal{ by \WBR conflicting critical section ordering} \\
1 + \dist{\getAcquire{\event{e}{}}}{\event{e}{'}} & \textnormal{ if } \WBROrdered{\event{e}{}}{\event{e}{'}} \textnormal{ by ``$\Acquire{m}\ltWBR\Release{m} \implies \Release{m}\ltWBR\Release{m}$''} \\
0 & \textnormal{ if } \WBROrdered{\event{e}{}}{\event{e}{'}} \textnormal{ by \PO} \\
\multicolumn{2}{@{}l}{1 + \min_{\event{e}{''}}(\dist{\event{e}{}}{\event{e}{''}} + \dist{\event{e}{''}}{\event{e}{'}})} \\
& \textnormal{ if } \exists \event{e}{''} \mid \WBROrdered{\event{e}{}}{\event{e}{''}} \land \WBROrdered{\event{e}{''}}{\event{e}{'}} \textnormal{ by \WBR transitivity} \\
\infty & \textnormal{ otherwise}
\end{cases}
\end{align*}

  \paragraph{Base case}

  Let \event{e}{_1} and \event{e}{_2} be events in a trace \tr such that
  $\dist{\event{e}{_1}}{\event{e}{_2}} = 0$ and
  \WBROrdered{\event{e}{_1}}{\event{e}{_2}}.
  Since $\dist{\event{e}{_1}}{\event{e}{_2}} = 0$, \event{e}{_1} and \event{e}{_2} are ordered directly,
  using \PO or \WBR's conflicting critical section ordering:
  
  If \WBROrdered{\event{e}{_1}}{\event{e}{_2}} by \PO,
  then \POOrdered{\event{e}{_1}}{\event{e}{_2}}. Let \trPrime be a predictable trace of \tr
  in which \TRPrimeOrdered{\event{e}{_2}}{\event{e}{_1}} or $\event{e}{_2} \in \trPrime \land \event{e}{_1} \notin \trPrime$,
  either of which violates the \PO rule of a predictable traces.
  
  If \WBROrdered{\event{e}{_1}}{\event{e}{_2}} by \WBR's conflicting critical section ordering,
  then \event{e}{_1} is a release event, \event{e}{_2} is a branch event, there
  exists a release event \event{r}{_2} over the same lock as \event{e}{_1}, and
  there exists write event \event{e} and read event \event{e}{'} such that
  $\event{e}\in\CS{\event{e}{_1}}$, $\event{e}{'}\in\CS{\event{r}{_2}}$,
  \conflicts{\event{e}}{\event{e}{'}}, $\lastwr{\event{e}{'}}{\tr} = \event{e}{}$, and
  \BrDepsOn{\event{e}{_2}}{\event{e}{'}}. Let \trPrime be a predictable trace of \tr
  where \TRPrimeOrdered{\event{e}{_2}}{\event{e}{_1}} or $\event{e}{_2} \in \trPrime \land \event{e}{_1} \notin \trPrime$.
  Then in either case,
  \TRPrimeOrdered{\event{r}{_2}}{\event{a}} or $\event{r}{_2} \in \trPrime \land \event{a} \notin \trPrime$,
  where \event{a} is the matching acquire of \event{e}{_1}; otherwise \trPrime would be an invalid predictable trace
  due to the \LS rule of predictable traces. As a result, \TRPrimeOrdered{\event{e}{'}}{\event{e}} or
  $\event{e}{'} \in \trPrime \land \event{e}{} \notin \trPrime$,
  which means $\lastwr{\event{e}{'}}{\trPrime} \neq \lastwr{\event{e}{'}}{\tr}$.
  Furthermore, \event{e}{'} is a causal read in \trPrime as
  \BrDepsOn{\event{e}{_2}}{\event{e}{'}} and $\event{e}{_2}\in\trPrime$. As a
  result, \trPrime violates the \LW rule of predictable traces.

  \paragraph{Inductive step}

  Given some $k > 0$,
  suppose that the lemma statement holds for
  $\dist{\event{e}{_1}}{\event{e}{_2}} < k$ (induction hypothesis).

  Let \event{e}{_1} and \event{e}{_2} be events in a trace \tr such that
  $\dist{\event{e}{_1}}{\event{e}{_2}} = k$ and
  \WBROrdered{\event{e}{_1}}{\event{e}{_2}}.
  Since $\dist{\event{e}{_1}}{\event{e}{_2}} > 0$, \event{e}{_1} and \event{e}{_2} are ordered indirectly,
  using \WBR transitivity or ``$\Acquire{m}\ltWBR\Release{m} \implies \Release{m}\ltWBR\Release{m}$'':


  If \WBROrdered{\event{e}{_1}}{\event{e}{_2}} by \WBR transitivity,
  then there must be an event \event{e} such that
  $\WBROrdered{\event{e}{_1}}{\event{e}} \land
  \WBROrdered{\event{e}}{\event{e}{_2}}$ and
  $\dist{\event{e}{_1}}{\event{e}{_2}} = 1 + \dist{\event{e}{_1}}{\event{e}} + \dist{\event{e}}{\event{e}{_2}}$.
  Let \trPrime be a predictable trace
  of \tr in which \TRPrimeOrdered{\event{e}{_2}}{\event{e}{_1}} or $\event{e}{_2} \in \trPrime \land \event{e}{_1} \notin \trPrime$.
  Then either
  \TRPrimeOrdered{\event{e}{}}{\event{e}{_1}},
  \TRPrimeOrdered{\event{e}{_2}}{\event{e}{}},
  $\event{e}{_2} \in \trPrime \land \event{e}{} \notin \trPrime$, or
  $\event{e}{} \in \trPrime \land \event{e}{_1} \notin \trPrime$. Any of these possibilities
  makes \trPrime an invalid predictable trace of \tr
  according to the induction hypothesis on \WBROrdered{\event{e}{_1}}{\event{e}} and \WBROrdered{\event{e}}{\event{e}{_2}}.

  If \WBROrdered{\event{e}{_1}}{\event{e}{_2}} by ``$\Acquire{m}\ltWBR\Release{m} \implies \Release{m}\ltWBR\Release{m}$'',
  then there exist acquire events \event{a}{_1} and \event{a}{_2} on the
  same lock such that \event{e}{_1} and \event{e}{_2} are their corresponding
  releases, \WBROrdered{\event{a}{_1}}{\event{e}{_2}}, and
  $\dist{\event{e}{_1}}{\event{e}{_2}} = 1 + \dist{\getAcquire{\event{e}{}}}{\event{e}{'}}$.
  Let \trPrime be a
  predictable trace of \tr where \TRPrimeOrdered{\event{e}{_2}}{\event{e}{_1}} or $\event{e}{_2} \in \trPrime \land \event{e}{_1} \notin \trPrime$.
  Then because of the \LS rule of predictable traces,
  either \TRPrimeOrdered{\event{e}{_2}}{\event{a}{_1}} or $\event{e}{_2} \in \trPrime \land \event{a}{_1} \notin \trPrime$.
  Either of these possibilities makes \trPrime invalid predictable trace of \tr
  according to the induction hypothesis on \WBROrdered{\event{a}{_1}}{\event{e}{_2}}.
\end{proof}

\section{Vindicating \WBR-Races}
\label{sec:vindication-full}

\notes{
\mike{Fixed ``\CORE'' in this section by redefining the macro as ``causal.''}
}

\renewcommand{\CORE}{causal\xspace}
\newcommand{\CapitalCORE}{Causal\xspace}
\renewcommand\GetCOREReads{\textsc{Get\CapitalCORE{}Reads}\xspace}

\WBR analysis is unsound: not every \WBR-race indicates a true predictable race 
(Section~\ref{sec:BR-WBR-relations}).
While \WBR-races often indicate true predictable races in practice,
true and false races alike take hours or days for programmers to understand and 
fix~\cite{literace, benign-races-2007, fasttrack, microsoft-exploratory-survey, billion-lines-later}.
We thus introduce an algorithm called \emph{\checkWBRRace}
that determines whether a \WBR-race is a true predictable race.
As shown in Algorithm~\ref{alg:verify-race},
\checkWBRRace takes as input a single \WBR-race
and a constraint graph corresponding to the execution trace with \WBR ordering,
and either returns a valid reordering exposing a predictable race, or reports failure.

\CheckWBRRace extends prior work's algorithm for vindicating \DC-races~\cite{vindicator}, 
referred to here as \emph{\checkDCRace}.
The main challenge in developing \checkWBRRace is that unlike \DC, \WBR does not in general order a read after its last writer,
since only \emph{\CORE} reads in a predictable trace \trPrime need to be ordered after their last writers from \tr.
\CheckWBRRace thus computes reads that must be \CORE reads in \trPrime,
both when computing constraints and when constructing the reordered trace \trPrime.

\notes{
The differences between \checkWBRRace and \checkDCRace
are related to the fact that \WBR is (by design) weaker than \DC,
so the initial \WBR constraint graph is weaker than a \DC constraint graph.
In addition, \checkWBRRace enforces a weaker (but still correct) definition of valid reordering
than \checkDCRace, incorporating control and data dependence precisely.
}

The rest of this section first defines the constraint graph 
and then explains \checkWBRRace and its component procedures.

\subsection{Constraint Graph}
\label{subsec:constraint-graph}

The \emph{constraint graph}, \Gv, is a directed graph in which
the nodes are the events in \tr and the edges represent reordering constraints on events in any reordered trace \trPrime.
Notationally, we represent \Gv as a set of edges, \eg, $\edge{\event{e}}{\event{e'}} \in \Gv$.

We use the notation \Gpath{\event{e}}{\event{e}{'}}{\Gv} 
to indicate that \event{e}{'} is reachable from \event{e} in \Gv:
\[
\Gpath{\event{e}}{\event{e}{'}}{\Gv} \;\; \equiv \;\; \edge{\event{e}}{\event{e}{'}} \in \Gv \lor \exists \event{e}{''} \mid \Gpath{\event{e}}{\event{e}{''}}{\Gv} \land \Gpath{\event{e}{''}}{\event{e}{'}}{\Gv} 
\]
The initial constraint graph is constructed by \WBR analysis,
which creates a node for each event and adds edges that correspond to \WBR ordering between events.
That is, initially the following property holds:
\[
\forall \event{e}, \event{e}{'} \in \tr \; \big( \WBROrdered{\event{e}{}}{\event{e}{'}} \iff \Gpath{\event{e}{}}{\event{e}{'}}{\Gv} \big)
\]

Figures~\ref{fig:br-rd-wr-conflict:wbr-false-race}, \ref{fig:br-rd-wr-conflict:wbr-no-edge},
\ref{fig:wr-wr-conflict:original},
\ref{fig:wr-wr-conflict:extra},
\ref{fig:simple-BR:original}, and
\ref{fig:simple-BR:branch:original}
(pages~\pageref{fig:br-rd-wr-conflict:wbr-false-race}--\pageref{fig:simple-BR:original})
essentially show initial constraint graphs
for traces that have a \WBR-race but no \DC-race.
The reader should consider only the edges labeled ``\WBR'' to be part of the constraint graph,
ignoring edges labeled ``\DC'' and ``\BR.''
The arrows in the figures represent edges corresponding to \WBR conflicting critical section ordering.\footnote{These
figures do not explicitly depict any
``$\Acquire{m}\ltWBR\Release{m} \implies \Release{m}\ltWBR\Release{m}$'' edges because
all such edges are already implied by other edges in the examples,
\eg, Figure~\ref{fig:simple-BR:branch:original}'s 
\Release{m} events are already ordered by a conflicting critical section edge composed with \PO.}
The figures omit showing \PO edges that exist between events by the same thread.

\subsection{\textsc{AddConstraints}} 

\algdef{SE}[DOWHILE]{Do}{doWhile}{\algorithmicdo}[1]{\algorithmicwhile\ #1}%

\begin{algorithm*}
\caption{\hfill Check if \WBR-race is a true predictable race}
\begin{small}
\raggedright An execution trace is an ordered list of events: $\langle e, \dots, e' \rangle$.\\
\raggedright The operator $\oplus$ concatenates two traces. \\
\end{small}
\small
\begin{algorithmic}[1]
\Procedure{\CheckWBRRace}{$\Gv, \event{e}{_1}, \event{e}{_2}$} \Comment{Inputs: constraint graph and \WBR-race events}
	\State $\Gv \gets \textsc{AddConstraints}(G, \event{e}{_1},\event{e}{_2})$\label{line:approach:call-wdc-b}
	\lIf {$\Gv = \emptyset$}{\textbf{return} \underline{No predictable race}}\label{line:approach:falseRace}
	\State \algorithmicelse\
	\Indent
		\State $\trPrime \gets \textsc{ConstructReorderedTrace}(\Gv, \event{e}{_1}, \event{e}{_2})$ \label{line:approach:backReorder}
		\lIf {$\trPrime \ne \langle \, \rangle\,$}{\textbf{return} \underline{Predictable race witnessed by \trPrime}} \label{line:approach:trueRace} \Comment{Check for non-empty trace}
		\lElse{\textbf{return} \underline{Don't know}} \label{line:approach:questionableRace}
	\EndIndent
\EndProcedure

\medskip

\Procedure{AddConstraints}{$\Gv, \event{e}{_1}, \event{e}{_2}$} \label{alg:addConstraints}
    \State $C \gets \{\edge{\mathit{src}}{\event{e}{_2}} \mid \edge{\mathit{src}}{\event{e}{_1}} \in G\} \cup \{\edge{\mathit{src}}{\event{e}{_1}} \mid \edge{\mathit{src}}{\event{e}{_2}} \in G\}$ \label{line:addConstraints:start-initial-constraints}
    \Comment{Consecutive-event constraints}
	\State $\Gv \gets \Gv \cup C$ \label{line:addConstraints:end-initial-constraints}
	\Do \label{line:addConstraints:start-adding-LSConstraints}
		\LineCommentxx{Add \LW constraints}
		\State $M \gets \GetCOREReads(\Gv, \event{e}{_1}, \event{e}{_2})$ \Comment{Compute \CORE reads} \label{line:addConstraints:find-core-reads}
		\ForEach{$\event{r}{} \in M$} \label{line:addConstraints:begin-adding-last-writer-edges}
			\State \textbf{let} $w \gets \lastwr{\event{r}{}}{}$
			\lIf {$w \ne \nolastwr$}{$\Gv \gets \Gv \cup \{\edge{\event{w}{}}{\event{r}{}}\}$} \label{line:addConstraints:add-wr-rd-to-G}
		\EndFor \label{line:addConstraints:end-adding-last-writer-edges}
		\LineCommentxx{Add \LS constraints}
		\ForEach {$\edge{\mathit{src}}{\mathit{snk}} \in C$}\label {line:addConstraints:outer-foreach-acq-rel}
			\ForEach {$(\event{a}{}, \event{r}{}) \mid 
			\event{a}{} \textnormal{ is an \code{acq}} \land
			\event{r}{} \textnormal{ is a  \code{rel}} \land 
			\Gpath{\event{a}{}}{\mathit{src}}{G} \land
			\Gpath{\mathit{snk}}{\event{r}{}}{G} \land
			\getLock{\event{a}{}} = \getLock{\event{r}{}}$} \label{line:addConstraints:foreach-acq-rel}
				\If {$(\Gpath{A(\event{r}{})}{\event{e}{_1}}{G} \lor \Gpath{A(\event{r}{})}{\event{e}{_2}}{G}) \land (\Gpath{\event{a}{}}{\event{e}{_1}}{G} \lor \Gpath{\event{a}{}}{\event{e}{_2}}{G})$} \label{line:addConstraints:acq-rel-reach-race}
					\State $C \gets C \cup \{\edge{R(\event{a}{})}{A(\event{r}{})}\}$ \label{line:addConstraints:add-acq-rel-to-C}
					\State $G \gets G \cup \{\edge{R(\event{a}{})}{A(\event{r}{})}\}$ \label{line:addConstraints:add-acq-rel-to-G}
				\EndIf
			\EndFor
		\EndFor
		\lIf {$\exists \event{e}{} \mid \Gpath{\event{e}{}}{\event{e}{}}{\Gv} \land (\Gpath{\event{e}{}}{\event{e}{_1}}{\Gv} \lor \Gpath{\event{e}{}}{\event{e}{_2}}{\Gv})$}{\textbf{return} $\emptyset$} \label{line:addConstraints:detect-cycle} \Comment{Cycle detected; no predictable race}
	\doWhile {$G$ has changed} \label{line:addConstraints:end-adding-LSConstraints}
	\State \textbf{return} $\Gv$
\EndProcedure

\medskip

\newcommand\constructHelper{\textsc{AttemptToConstructTrace}\xspace}
\newcommand\ConstructHelper{\constructHelper}
\Procedure{ConstructReorderedTrace}{$\Gv, \event{e}{_1}, \event{e}{_2}$}
	\Do
		\State $\trPrime \gets \constructHelper(G, \event{e}{_1}, \event{e}{_2})$\label{line:fullConstruct:construct-call}
		\lIf {$\trPrime = \langle \event{r}{} \rangle$}{$G \gets G \cup \{ \edge{\event{r}}{\event{e}{_1}} \}$} \Comment {Missing needed release event?} \label{line:fullConstruct:recompute-r}
	\doWhile{$G$ has changed}
	\State \textbf{return} \trPrime
\EndProcedure

\newcommand\legal{\ensuremath{\mathit{legal}}\xspace}
\newcommand\semilegal{\ensuremath{\mathit{next}}\xspace}

\medskip

\Procedure{\constructHelper}{$G, \event{e}{_1}, \event{e}{_2}$}
	\State $R \gets \{\event{e}{} \mid \Gpath{\event{e}{}}{\event{e}{_1}}{G} \lor \Gpath{\event{e}{}}{\event{e}{_2}}{G}\}$ \Comment{Compute reachable events}\label{line:fullConstruct:compute-r}
	\State $M \gets \GetCOREReads(\Gv, \event{e}{_1}, \event{e}{_2})$ \Comment{Compute \CORE reads} \label{line:backReorder:find-core-reads}
	\State $\trPrime \gets \langle \event{e}{_1}, \event{e}{_2} \rangle$ \label{line:backReorder:initialize-trPrime}
	\While {$R \setminus \trPrime \ne \emptyset$} \label{line:backReorder:Rloop}
		\State $ 
		\semilegal \gets \{ \event{e}{} \in R \setminus \trPrime \mid 
		( \nexists \event{e}{'} \mid (e, e') \in G \land e' \in R \setminus \trPrime ) \}$%
		\label{line:backReorder:semiLegalEvent}%
		\State
		$\legal \gets \{ \event{e}{} \in \semilegal \mid 
		(\langle \event{e}{} \rangle \oplus \trPrime \textnormal{ satisfies \LS}) \land 
		(\langle \event{e}{} \rangle \oplus \trPrime \textnormal{ satisfies \LW according to } M) \}$
		\label{line:backReorder:legalEvent}
		\If {$\legal = \emptyset$}

			\If{$\exists \event{r}{} \mid (\exists \event{e}{} \in \semilegal \mid \event{e}{} \in \CS{\event{r}{}} \land 
			                         \event{r}{} \notin R \land \langle \event{r}{} \rangle \oplus \trPrime \textnormal{ satisfies \LS})$}
			\label{line:backReorder:check-missing-release}
				\State \textbf{return} $\langle \event{r}{} \rangle$ \Comment{Return missing release} \label{line:backReorder:missing-release}
			\EndIf
			\State \textbf{return} $\langle \, \rangle$ \Comment{Failed to construct trace} \label{line:backReorder:stuck}
		\Else
			\State \textbf{let} $e \in \legal$ s.t.\  $\nexists e' \in \legal \mid \TROrdered{e}{e'}$ \Comment {Select latest legal event in \tr order} \label{line:backReorder:pick-latest-event}
 			\State $\trPrime \gets \langle \event{e}{} \rangle \oplus \trPrime$ \Comment{Prepend event to trace} \label{line:backReorder:append}
		\EndIf
	\EndWhile \label{line:backReorder:RloopEnd}
	\State \textbf{return} \trPrime \Comment{Return valid reordering} \label{line:backReorder:return-trPrime}
\EndProcedure

\end{algorithmic}
\label{alg:verify-race}
\label{alg:Approach}
\end{algorithm*}

The initial constraint graph \Gv is insufficiently constrained---a
reordered trace that obeys \Gv's constraints will not in general be a valid reordering that exposes a predictable race---both
because \WBR is unsound and because \Gv's constraints do not ensure that a \WBR-race's accesses occur consecutively.

\CheckWBRRace first calls \emph{\textsc{AddConstraints}} (at line~\ref{line:approach:call-wdc-b} in Algorithm~\ref{alg:verify-race}),
which adds \emph{necessary} constraints for a valid reordering exposing a predictable race: constraints that
(1) force the input \WBR-race's accesses to be consecutive (to satisfy the definition of predictable race),
(2) order some \CORE events (to satisfy the \LW rule of valid reordering), and
(3) order some critical sections (to satisfy the \LS rule of valid reordering).

\subsubsection{Consecutive-event constraints}\label{sssec:consecutive-event-const}

For an input \WBR-race pair \event{e}{_1} and \event{e}{_2}
to be a predictable race, \event{e}{_1} and \event{e}{_2} must be consecutive in a valid reordering.
To force this ordering, every event ordered before \event{e}{_1} \textbf{or} \event{e}{_2}
must be ordered before both \event{e}{_1} \textbf{and} \event{e}{_2}.
Lines~\ref{line:addConstraints:start-initial-constraints} and 
\ref{line:addConstraints:end-initial-constraints} of Algorithm~\ref{alg:Approach}
add \emph{consecutive-event} constraints to \Gv.
For every immediate predecessor event $\mathit{src}$ of \event{e}{_1} or \event{e}{_2} in \Gv,
\textsc{AddConstraints} adds a constraint from $\mathit{src}$ to \event{e}{_2} or \event{e}{_1}, respectively.

Figures~\ref{fig:constraints:wr-wr-conflict:consecutive-event} and \ref{fig:constraints:complex-noRace:consecutive-event}
show constraint graphs after adding consecutive-event constraints (dashed lines).
Note that edges from \event{e}{_2} to \event{e}{_1}'s predecessor
typically point \emph{backward} relative to \TROrdered{}{} order.

\begin{figure*}
\captionsetup{farskip=0pt} 
\subfloat[After adding consecutive-event constraints]{
\small
\centering
\sf
\makebox[.33\linewidth]{
\begin{tabular}{@{}lll@{}}
\textnormal{Thread 1} & \textnormal{Thread 2} & \textnormal{Thread 3}\\\hline
\Acquire{m}\tikzmark{10} \\
\tikzmark{5}\Write{x} \\
\textbf{\Write{y}}\tikzmark{4} \\
\Release{m} \\
					  & \Acquire{m} \\
				 	  & \Write{x}\tikzmark{7} \\
					  & \Release{m}\tikzmark{1} \\
                      &						  & \Acquire{m} \\
                      &						  & \tikzmark{8}\Read{x} \\
                      &						  & \tikzmark{2}\BrUniv \\
                      & 					  & \Release{m}\tikzmark{3}\tikzmark{9} \\
                      &						  & \tikzmark{6}\textbf{\Read{y}}
\end{tabular}
\textunderlink{1}{2}{\WBR}%
\init{3}{4}{0}{0}%
\init{5}{6}{180}{180}%
\label{fig:constraints:wr-wr-conflict:consecutive-event}%
}
}
\subfloat[After adding \LW constraints]{
\small
\centering
\sf
\makebox[.33\linewidth]{
\begin{tabular}{@{}lll@{}}
\textnormal{Thread 1} & \textnormal{Thread 2} & \textnormal{Thread 3}\\\hline
\Acquire{m}\tikzmark{10} \\
\tikzmark{5}\Write{x} \\
\textbf{\Write{y}}\tikzmark{4} \\
\Release{m} \\
					  & \Acquire{m} \\
				 	  & \Write{x}\tikzmark{7} \\
					  & \Release{m}\tikzmark{1} \\
                      &						  & \Acquire{m} \\
                      &						  & \tikzmark{8}\Read{x} \\
                      &						  & \tikzmark{2}\BrUniv \\
                      & 					  & \Release{m}\tikzmark{3}\tikzmark{9} \\
                      &						  & \tikzmark{6}\textbf{\Read{y}}
\end{tabular}
\textunderlink{1}{2}{\WBR}%
\init{3}{4}{0}{0}%
\init{5}{6}{180}{180}%
\textlink{7}{8}{\LW}%
\label{fig:constraints:wr-wr-conflict:last-writer}%
}
}
\subfloat[After adding \LS constraints]{
\small
\centering
\sf
\makebox[.33\linewidth]{
\begin{tabular}{@{}lll@{}}
\textnormal{Thread 1} & \textnormal{Thread 2} & \textnormal{Thread 3}\\\hline
\Acquire{m}\tikzmark{10} \\
\tikzmark{5}\Write{x} \\
\textbf{\Write{y}}\tikzmark{4} \\
\Release{m} \\
					  & \Acquire{m} \\
				 	  & \Write{x}\tikzmark{7} \\
					  & \Release{m}\tikzmark{1} \\
                      &						  & \Acquire{m} \\
                      &						  & \tikzmark{8}\Read{x} \\
                      &						  & \tikzmark{2}\BrUniv \\
                      & 					  & \Release{m}\tikzmark{3}\tikzmark{9} \\
                      &						  & \tikzmark{6}\textbf{\Read{y}}
\end{tabular}
\textunderlink{1}{2}{\WBR}%
\init{3}{4}{0}{0}%
\init{5}{6}{180}{180}%
\textlink{7}{8}{\LW}%
\back{9}{10}{30}{30}%
\label{fig:constraints:wr-wr-conflict:lock-semantics}
}}

\caption{Constraint graph after successive steps of \textsc{AddConstraints}
for \WBR-race on \code{y}, for an example execution.
The \WBR-race events are in \textbf{bold}.
\label{fig:constraints:wr-wr-conflict}}

\bigskip
\small
\centering
\subfloat[After adding consecutive-event constraints]{
\centering
\sf
\begin{tabular}{@{}l@{\;\;}l@{\;\;}l@{\;\;}l@{\;\;}l@{}}
\textnormal{Thread 1} & \textnormal{Thread 2} & \textnormal{Thread 3} & \textnormal{Thread 4} & \textnormal{Thread 5}\\\hline
					  & \Acquire{m}\tikzmark{18} \\
					  & \Acquire{n}\tikzmark{16} \\
					  & \tikzmark{7}\Sync{t} \\
\Sync{t}\tikzmark{8} \\
\tikzmark{11}\Acquire{p}\tikzmark{14} \\
\textbf{\Write{x}}\tikzmark{10} \\
\Sync{o}\tikzmark{1} \\
\Write{z}\tikzmark{21} & \tikzmark{2}\Sync{o} \\
\Release{p} & \Release{n} \\
					  & \Release{m} 		  &						  & \Acquire{n} \\
					  &						  &						  & \Sync{q}\tikzmark{3} \\
					  &						  &						  &						  & \tikzmark{4}\Sync{q} \\
					  &						  &						  &						  & \Acquire{p} \\
					  &						  &						  &						  & \tikzmark{22}\Read{z} \\
					  &						  &						  &						  & \Release{p}\tikzmark{9}\tikzmark{13} \\
                      &						  & \Acquire{m} 		  &						  & \tikzmark{12}\textbf{\Read{x}} \\
                      &						  & \Sync{s}\tikzmark{5} &					  & \Acquire{l} \\
                      &						  &						  &						  &	\tikzmark{19}\Write{y} \\
                      &						  &						  &			 			  & \Release{l} \\
                      &						  &	\Acquire{l} \\
                      &						  &	\Read{y}\tikzmark{20} & \tikzmark{6}\Sync{s}\\ 		  
                      &						  &	\BrUniv 		  & \Release{n}\tikzmark{15} \\
                      &						  &	\Release{l} 	 	  \\
                      & 					  & \tikzmark{17}\Release{m}
\end{tabular}
\label{fig:constraints:complex-noRace:consecutive-event}
\textlink{1}{2}{\WBR}%
\textunderlink{3}{4}{\WBR}%
\textlink{5}{6}{\WBR}%
\textlink{7}{8}{\WBR}%
\init{9}{10}{30}{30}%
\init{11}{12}{200}{180}%
}
\hspace*{1em}
\subfloat[After adding \LW and \LS constraints]{
\centering
\sf
\begin{tabular}{@{}l@{\;\;}l@{\;\;}l@{\;\;\;}l@{\;\;}l@{}}
\textnormal{Thread 1} & \textnormal{Thread 2} & \textnormal{Thread 3} & \textnormal{Thread 4} & \textnormal{Thread 5}\\\hline
					  & \Acquire{m}\tikzmark{18} \\
					  & \Acquire{n}\tikzmark{16} \\
					  & \tikzmark{7}\Sync{t} \\
\Sync{t}\tikzmark{8} \\
\tikzmark{11}\Acquire{p}\tikzmark{14} \\
\textbf{\Write{x}}\tikzmark{10} \\
\Sync{o}\tikzmark{1} \\
\Write{z}\tikzmark{21} & \tikzmark{2}\Sync{o} \\
\Release{p} & \Release{n} \\
					  & \Release{m} 		  &						  & \Acquire{n} \\
					  &						  &						  & \Sync{q}\tikzmark{3} \\
					  &						  &						  &						  & \tikzmark{4}\Sync{q} \\
					  &						  &						  &						  & \Acquire{p} \\
					  &						  &						  &						  & \tikzmark{22}\Read{z} \\
					  &						  &						  &						  & \Release{p}\tikzmark{9}\tikzmark{13} \\
                      &						  & \Acquire{m} 		  &						  & \tikzmark{12}\textbf{\Read{x}} \\
                      &						  & \Sync{s}\tikzmark{5} &					  & \Acquire{l} \\
                      &						  &						  &						  &	\tikzmark{19}\Write{y} \\
                      &						  &						  &			 			  & \Release{l} \\
                      &						  &	\Acquire{l} \\
                      &						  &	\Read{y}\tikzmark{20} & \tikzmark{6}\Sync{s}\\ 		  
                      &						  &	\BrUniv 		  & \tikzmark{15}\Release{n} \\
                      &						  &	\Release{l} 	 	  \\
                      & 					  & \tikzmark{17}\Release{m}
\end{tabular}
\label{fig:constraints:complex-noRace:last-writer-lock-semantics}
\textlink{1}{2}{\WBR}%
\textunderlink{3}{4}{\WBR}%
\textunderlink{5}{6}{\WBR}%
\textlink{7}{8}{\WBR}%
\back{13}{14}{30}{30}%
\back{15}{16}{260}{30}%
\back{17}{18}{170}{30}%
\textlink{19}{20}{\LW}%
\textcurvelink{21}{22}{-75}{180}{\LW}%
\init{9}{10}{30}{30}%
\init{11}{12}{200}{180}%
}


%



\caption{Constraint graph after successive steps of \textsc{AddConstraints} for the \WBR-race on \code{x},
for an example execution.
The \WBR-race events are in \textbf{bold}.
The execution has a \WBR-race 
(\nWBROrdered{\Write{y}{}}{\Read{y}{}}) but no predictable race.
\Sync{o} is an abbreviation for the sequence \Acquire{o}\code{;} \Read{oVar}\code{;} \BrUniv{}\code{;} \Write{oVar}\code{;} \Release{o}.}
\label{fig:constraints:complex-noRace}

\end{figure*}


\textsc{AddConstraints} adds the new constraint edges not only to \Gv
but also to a new set $C$ whose edges are the starting point
for discovering ordering constraints on critical sections (described later).

\subsubsection{Last-writer constraints}\label{sssec:last-writer-const}
\label{subsec:lw-constraints}

\later{
\mike{Why exactly do we add this particular kind of constraint?
If we don't, do we find that \textsc{ConstructReorderedTrace} sometimes gets stuck (and returns ``Don't know'')?
\jake{If we don't, \textsc{ConstructReorderedTrace} will get stuck when a
cycle should be detected and I think \textsc{ConstructReorderedTrace} could get
stuck when a valid reordering exists for a true race.}
\mike{After the submission, let's verify whether that's true empirically.}
\mike{Update: Practical benefits aside, last-writer edges preserve completeness,
so adding them is in accordance with the idea that \textsc{AddConstraints} adds as many (complete) constraints as it can
before calling \textsc{ConstructReorderedTrace}.}}
}

After adding consecutive-event constraints,
\textsc{AddConstraints} identifies and adds ordering constraints for \CORE reads,
called \emph{last-writer (\LW)} constraints
(lines~\ref{line:addConstraints:find-core-reads}--\ref{line:addConstraints:end-adding-last-writer-edges} in Algorithm~\ref{alg:verify-race}).
These constraints have the following form:
if a read event \event{e}{} may affect the outcome of some branch that must be in any valid reordering
(\ie, a branch ordered in \Gv before \event{e}{_1} or \event{e}{_2}),
then the last writer of \event{e}{}, if any, must be ordered to the read in any valid reordering.

These read events are exactly the set of \CORE reads (\Def~\ref{def:causal-events}).
\textsc{AddConstraints} identifies the \CORE reads by calling a procedure \emph{\GetCOREReads} at line~\ref{line:addConstraints:find-core-reads}.
\GetCOREReads (definition not shown) performs a backward traversal of \Gv
starting from \event{e}{_1} and \event{e}{_2} to identify reads
that are \CORE events according to the recursive definition of \CORE events.
This traversal includes not only constraint edges in \Gv but also includes
traversing backward from any \CORE read to its last writer (if any), as computed by \WBR analysis.
Note that edges between a \CORE read and its last writer
do \emph{not} in general exist in \Gv
(intuitively, \WBR analysis cannot know which reads must be \CORE reads in valid reorderings for a particular \WBR-race).

After identifying \CORE reads,
\textsc{AddConstraints} adds edges between each \CORE read and its last writer, if any,
at lines~\ref{line:addConstraints:begin-adding-last-writer-edges}--\ref{line:addConstraints:end-adding-last-writer-edges}.
These \LW constraints are \emph{necessary} constraints on a valid reordering;
we find that these constraints are useful for identifying false races during \textsc{AddConstraints}
(instead of during \textsc{ConstructReorderedTrace}; Section~\ref{subsec:ConstructReorderedTrace})
and for constraining reordering enough during \textsc{ConstructReorderedTrace}
so that true races can be successfully vindicated.
On the other hand, \LW constraints are \emph{insufficient} to ensure the \LW rule of valid reordering:
Enforcing merely that each read's last writer \emph{precede} it does not ensure preservation of each read's last writer.
\textsc{AddConstraints} could potentially add more constraints,
but we have found the added \LW constraints to be sufficient in practice
for identifying false races during \textsc{AddConstraints}
and for enabling \textsc{ConstructReorderedTrace} to vindicate true races.


Figure~\ref{fig:constraints:wr-wr-conflict:last-writer} shows the constraint graph
from Figure~\ref{fig:constraints:wr-wr-conflict:consecutive-event} after adding the sole \LW edge (solid arrow labeled \emph{\LW}).
The \Read{x} is a \CORE read (since it can affect the \BrUniv event, which must be in a valid reordering since \Gpath{\BrUniv}{\Read{y}}{\Gv}),
so \textsc{AddConstraints} adds an edge to the read from its last writer, \WriteT{x}{T2}
(\ie, the \Write{x} executed by Thread~2).


Likewise, Figure~\ref{fig:constraints:complex-noRace:last-writer-lock-semantics}
shows the constraint graph from Figure~\ref{fig:constraints:complex-noRace:consecutive-event}
after adding \LW edges (solid arrows labeled \emph{\LW}).
However, \textsc{AddConstraints} does not actually add these edges until after it has added \LS edges (dotted arrows):
Thread~3's \BrUniv event does not reach \Write{x} or \Read{x} without the \LS edges.
\textsc{AddConstraints} adds \LW and \LS constraints iteratively until convergence
(lines~\ref{line:addConstraints:start-adding-LSConstraints}--\ref{line:addConstraints:end-adding-LSConstraints}).
After adding the shown \LS constraints (explained below),
the next loop iteration of \textsc{AddConstraints} detects that
\Gpath{\BrUniv}{\Write{x}}{\Gv}
and thus that \Read{y} and \Read{z} are \CORE reads, and it adds the shown \LW constraints.


\later{
\mike{Not sure in what depth we've already discussed this,
but it seems to me that last writer--read edges are conceptually like critical sections in terms of ordering constraints,
and we can use a similar approach for ordering groups of a last writer and its readers.}
}

\subsubsection{Lock-semantics constraints}\label{sssec:lock-semantics-const}

\textsc{AddConstraints} also identifies and adds necessary ordering constraints on critical sections,
called \emph{lock-semantics (\LS) constraints}
(lines~\ref{line:addConstraints:outer-foreach-acq-rel}--\ref{line:addConstraints:add-acq-rel-to-G} in Algorithm~\ref{alg:verify-race}).
These constraints have the following form:
if two critical sections on the same lock 
are ordered in \Gv, at least in part, and
both critical sections are ordered in \Gv, at least in part,
before \event{e}{_1} or \event{e}{_2},
then the critical sections must be \emph{fully} 
ordered in any valid reordering.
(The helper function \getLock{e} returns \code{m} if $e$ is an \Acquire{m} or \Release{m}.)


Figure~\ref{fig:constraints:wr-wr-conflict:lock-semantics} shows the constraint
graph from Figure~\ref{fig:constraints:wr-wr-conflict:last-writer} after adding \LS constraints (dotted arrows).
The algorithm detects that \Gpath{\AcquireT{m}{T3}}{\ReleaseT{m}{T1}}{\Gv}
and that both critical sections reach at least one access to \code{y},
so it adds an edge from \ReleaseT{m}{T3} to \AcquireT{m}{T1}.
After the first iteration of its \code{do--while} loop, \textsc{AddConstraints} finds no new \LW or \LS edges to add (convergence).

Figure~\ref{fig:constraints:complex-noRace:last-writer-lock-semantics} shows the constraint graph
from Figure~\ref{fig:constraints:complex-noRace:consecutive-event} after adding all \LS constraints (dotted arrows) as well as \LW constraints.
During the first iteration of the loop
(lines~\ref{line:addConstraints:start-adding-LSConstraints}--\ref{line:addConstraints:end-adding-LSConstraints}),
\textsc{AddConstraints} detects that
\Gpath{\AcquireT{p}{T5}}{\ReleaseT{p}{T1}}{\Gv} and 
\Gpath{\AcquireT{n}{T4}}{\ReleaseT{n}{T2}}{\Gv} and that both pairs of critical sections reach at least one access to \code{x},
so it adds edges from \ReleaseT{p}{T5} to \AcquireT{p}{T1} and
from \ReleaseT{n}{T4} to \AcquireT{n}{T2}.
During the second loop iteration, because of the added edge from \ReleaseT{n}{T4} to \AcquireT{n}{T2},
\textsc{AddConstraints} detects the path \Gpath{\AcquireT{m}{T3}}{\ReleaseT{m}{T2}}{\Gv},
and thus adds the edge \ReleaseT{m}{T3} to \AcquireT{m}{T2}.
Finally, in the third loop iteration, as explained above,
\GetCOREReads identifies \Read{y} and \Read{z} as \CORE events and adds the corresponding \LW edges shown in the figure.

\subsubsection{Detecting cycles}

\CheckWBRRace will definitely \emph{not} be able to construct a valid reordering that satisfies \Gv's constraints
if \Gv contains a cycle that reaches \event{e}{_1} or \event{e}{_2}.
At the end of each loop iteration, 
\textsc{AddConstraints} checks for such a cycle at line~\ref{line:addConstraints:detect-cycle},
and returns an empty constraint graph if it detects a cycle.

Figure~\ref{fig:constraints:wr-wr-conflict:lock-semantics}'s constraint graph is acyclic,
and \textsc{AddConstraints} returns the shown constraint graph.
In contrast, Figure~\ref{fig:constraints:complex-noRace:last-writer-lock-semantics}'s constraint graph
contains a cycle that reaches \event{e}{_1} (and \event{e}{_2}):
\Gpath{\WriteT{x}{T1}}{\Gpath{\ReleaseT{p}{T5}}{\Gpath{\AcquireT{p}{T1}}{\WriteT{x}{T1}}{\Gv}}{\Gv}}{\Gv}.
\textsc{AddConstraints} detects the cycle and returns $\emptyset$,
and \checkWBRRace reports that no predictable race exists.

\later{
\subsubsection{Completeness of \textsc{AddConstraints}}

\kaan{Reviewer B asks ``do you have an example showing that these constraints are not sufficient?'', we can refer to \WDC's appendix example or add a revised version of it to this paper.
\mike{Do we have any examples of vindication failing that are \WBR specific or specific to this paper's definition of reordered trace
that incorporates dependencies,
or are we only aware of cases like those in \cite{vindicator} that involve pairs of critical sections with irreconcilable constraints?}}

Here we argue that \textsc{AddConstraints} is complete, meaning that it will not create a cycle if
there is a predictable race. To prove this, we first define the concept of \CORE events in the context of the constraint graph, then formally define the
edges contained in \Gv after \textsc{AddConstraints} has finished executing.

Since
\textsc{AddConstraints} runs separately for every \WBR-race,
the following definitions apply within the context of an observed trace \tr and some \WBR-race
between events \event{e}{_1} and \event{e}{_2}
(\ie, \nWBROrdered{\event{e}{_1}}{\event{e}{_2}} and
$\commonLocks{\event{e}{_1}}{\event{e}{_2}}=\emptyset$).

To identify read events that must have the same last writer in a reordering as in the original trace,
we introduce the concept of a \CORE event with respect to a race and constraint graph:

\begin{definition}[\CapitalCORE events for a constraint graph and race]
  Let event \event{e} be some event in \tr.
  An event \event{e} is a \emph{\CORE event for a constraint graph $G$ and a \WBR-race (\event{e}{_1},\event{e}{_2})}
  if and only if at least one of the following is true:
  \begin{enumerate}
    \item \event{e} is a read, and there exists a branch event \event{b} such that \BrDepsOn{\event{b}}{\event{e}},
    and \Gpath{\event{b}}{\event{e}{_1}}{\Gv} or \Gpath{\event{b}}{\event{e}{_2}}{\Gv}.\label{def:core-g-read}
    \item \event{e} is a write, and there exists a \CORE read \event{e}{'} for $G$ and (\event{e}{_1},\event{e}{_2})
    such that $\lastwr{\event{e}{'}}{\tr} = \event{e}$.\label{def:core-g-write}
    \item \event{e} is a read, and there exists a \CORE write \event{e}{'} for $G$ and (\event{e}{_1},\event{e}{_2})
    such that \POOrdered{\event{e}}{\event{e}{'}}.\label{def:core-g-recur}
  \end{enumerate}
\end{definition}

This definition is similar to the definition of a \CORE event in a trace (\Def~\ref{def:core-event}),
with the only difference being that the first rule only includes branches
that must be in the reordered trace (\ie, must reach \event{e}{_1} or \event{e}{_2}).

\begin{lemma}[Edges of \Gv after \textsc{AddConstraints}]
  If $\edge{\event{e}}{\event{e}{'}}\in\Gv$ then the edge must be either an initial
  constraint added during the execution, or a constraint added by \textsc{AddConstraints} for the race $(\event{e}{_1},\event{e}{_2})$ being vindicated,
  \ie, at least one of the following must hold:

  \begin{enumerate}
    \item \emph{Initial constraint (Section~\ref{subsec:constraint-graph}):} \WBROrdered{\event{e}}{\event{e}{'}}.\label{addConstraints:wbr-ordered}
    \item \emph{Consecutive-event constraint (Section~\ref{sssec:consecutive-event-const}):} $\event{e}{'}=\event{e}{_1}$ and \POOrdered{\event{e}}{\event{e}{_2}}, or $\event{e}{'}=\event{e}{_2}$ and \POOrdered{\event{e}}{\event{e}{_1}}.\label{addConstraints:consecutive}
    \item \emph{\LW constraint (Section~\ref{sssec:last-writer-const}):} \event{e}{'} is a \CORE read in \Gv for the race (\event{e}{_1},\event{e}{_2}), and $\lastwr{\event{e}{'}}{\tr}=\event{e}$.\label{addConstraints:last-writer}
    \item \emph{\LS constraint (Section~\ref{sssec:lock-semantics-const}):} \event{e} is a release and \event{e}{'} is an acquire over the same lock, \Gpath{\getAcquire{\event{e}}}{\getRelease{\event{e}{'}}}{\Gv}, and $(\Gpath{\event{e}}{\event{e}{_1}}{\Gv}\lor\Gpath{\event{e}}{\event{e}{_2}}{\Gv}) \land (\Gpath{\event{e}{'}}{\event{e}{_1}}{\Gv}\lor\Gpath{\event{e}{'}}{\event{e}{_2}}{\Gv})$.\label{addConstraints:lock-semantics}
  \end{enumerate}\label{lem:Gv-after-addConstraints}
\end{lemma}

Lemma~\ref{lem:Gv-after-addConstraints} is trivially true, as the initial edges of the graph and the added constraints are the only edges within the graph at the time of vindication.

To prove that \textsc{AddConstraints} is complete, we will first prove that all
edges in the graph are necessary for a valid reordering where the race is observed.

\begin{lemma}
  Given a \WBR-race $(\event{e}{_1},\event{e}{_2})$, events \event{e} and
  \event{e}{'} such that \Gpath{\event{e}}{\event{e}{'}}{\Gv}, and a reordering
  \trPrime of \tr; if \TRPrimeOrdered{\event{e}{'}}{\event{e}} or
  $\event{e}{'}\in\trPrime \land \event{e}\not\in\trPrime$ then \trPrime is invalid or there exists
  some event $\event{e}{''}\in\trPrime$ such that
  \TRPrimeOrdered{\event{e}{_1}}{\TRPrimeOrdered{\event{e}{''}}{\event{e}{_2}}} or \TRPrimeOrdered{\event{e}{_2}}{\TRPrimeOrdered{\event{e}{''}}{\event{e}{_1}}}.\label{lem:graph-edges-necessary}
\end{lemma}

\begin{proof}[Proof of Lemma~\ref{lem:graph-edges-necessary}]
  We prove the lemma by induction on the \emph{\Gdist} which is defined as follows
  (the rules refer to Lemma~\ref{lem:Gv-after-addConstraints}):
  
  \begin{align*}
    \dist{\event{e}{}}{\event{e}{'}} = \min
    \begin{cases}
      0 & \textnormal{ if } (\event{e},\event{e}{'})\in\Gv \textnormal{ by Rules~\ref{addConstraints:wbr-ordered},\ref{addConstraints:consecutive},\ref{addConstraints:last-writer} } \\
      1 + \dist{\getAcquire{\event{e}}}{\getRelease{\event{e}{'}}} & \textnormal{ if } (\event{e},\event{e}{'})\in\Gv \textnormal{ by Rule~\ref{addConstraints:lock-semantics} }\\
      \multicolumn{2}{r}{\textnormal{ if } \Gpath{\event{e}}{\event{e}{''}}{\Gv} \textnormal{ and } \Gpath{\event{e}{''}}{\event{e}{'}}{\Gv}}\\
      \multicolumn{2}{l}{\hspace{-.5em} 1 + \min_{\event{e}{''}}(\dist{\event{e}}{\event{e}{''}}+\dist{\event{e}{''}}{\event{e}{'}})}\\
      \infty & \textnormal{ otherwise }
    \end{cases}
  \end{align*}

  \paragraph{Base case} Let \event{e} and \event{e}{'} be events in a trace \tr
  such that $\dist{\event{e}}{\event{e}{'}}=0$ and
  $\Gpath{\event{e}}{\event{e}{'}}{\Gv}$. Since
  $\dist{\event{e}}{\event{e}{'}}=0$, \event{e} and \event{e}{'} must have a
  direct edge between them added to \Gv by one of the following rules:
  
  If this edge was added by Rule~\ref{addConstraints:wbr-ordered}, then
  Lemma~\ref{lem:wbr-necessary-ordering} proves that \trPrime is invalid.

  If this edge was added by Rule~\ref{addConstraints:consecutive}, then either
  $\event{e}{'}=\event{e}{_1}$ and \POOrdered{\event{e}}{\event{e}{_2}}, or
  $\event{e}{'}=\event{e}{_2}$ and \POOrdered{\event{e}}{\event{e}{_1}}. Then, a
  reordering \trPrime where \TRPrimeOrdered{\event{e}{'}}{\event{e}} would have
  either
  \TRPrimeOrdered{\event{e}{_1}}{\TRPrimeOrdered{\event{e}}{\event{e}{_2}}} or
  \TRPrimeOrdered{\event{e}{_2}}{\TRPrimeOrdered{\event{e}}{\event{e}{_1}}}
  respectively. Otherwise if $\event{e}\not\in\trPrime$, then \trPrime is invalid due to \PO rule since either \POOrdered{\event{e}}{\event{e}{_2}} or \POOrdered{\event{e}}{\event{e}{_1}}.

  If this edge was added by Rule~\ref{addConstraints:last-writer}, then
  \event{e}{'} is a \CORE read in \Gv for the race (\event{e}{_1},\event{e}{_2}) and
  $\lastwr{\event{e}{'}}{\tr}=\event{e}$. If \event{e}{'} is a \CORE read in \Gv,
  then it must be able to reach \event{e}{_1} or \event{e}{_2} following a chain of
  \CORE writes, reads, and finally a branch. Since \LW constraints are not added from branches,
  this branch then must be able to reach \event{e}{_1} or \event{e}{_2} via
  a constraint added by another rule.

  \paragraph{Inductive step} Given some $k > 0$, assume that the lemma statement holds for $\dist{\event{e}}{\event{e}{'}} < k$.

  Let \event{e} and \event{e}{'} be events in trace \tr such that $\dist{\event{e}}{\event{e}{'}}=k$ and \TROrdered{\event{e}}{\event{e}{'}}. Since $\dist{\event{e}}{\event{e}{'}}>0$, the distance must have been established by Rule~\ref{addConstraints:lock-semantics} or by transitivity.

  If Rule~\ref{addConstraints:lock-semantics} established the distance, then \event{e} is a
  release and \event{e}{'} is an acquire over the same lock and
  \Gpath{\getAcquire{\event{e}}}{\getRelease{\event{e}{'}}}{\Gv}. By definition,
  $\dist{\getAcquire{\event{e}}}{\getRelease{\event{e}{'}}}<k$, which means that
  there is no reordering \trPrime where
  \TRPrimeOrdered{\getRelease{\event{e}{'}}}{\getAcquire{\event{e}}}. This means
  that any valid reordering \trPrime must have
  \TRPrimeOrdered{\getAcquire{\event{e}}}{\getRelease{\event{e}{'}}}. As a result, any reordering \trPrime of \tr where \TRPrimeOrdered{\event{e}{'}}{\event{e}} is invalid as such a trace is either not well formed or violates the \LS rule. Otherwise, if $\event{e}\not\in\trPrime$, then \trPrime is invalid due to \LS rule as \TRPrimeOrdered{\getAcquire{\event{e}}}{\event{e}{'}}, and without \event{e} the critical sections would overlap.

  If the distance was established by transitivity, then there must be some event
  \event{e}{''} such that $\dist{\event{e}}{\event{e}{'}} =
  1+\dist{\event{e}}{\event{e}{''}}+\dist{\event{e}{''}}{\event{e}{'}}$. As a
  result, $\dist{\event{e}}{\event{e}{''}}<k$ and
  $\dist{\event{e}{''}}{\event{e}{'}}<k$. This shows that any valid reordering
  \trPrime must have \TRPrimeOrdered{\event{e}}{\event{e}{''}} and
  \TRPrimeOrdered{\event{e}{''}}{\event{e}{'}}. Therefore, any reordering
  \trPrime of \tr where \TRPrimeOrdered{\event{e}{'}}{\event{e}} must also have
  \TRPrimeOrdered{\event{e}{'}}{\event{e}{''}} or
  \TRPrimeOrdered{\event{e}{''}}{\event{e}}, which would make \trPrime invalid.
  Otherwise if $\event{e}\not\in\trPrime$ then \trPrime is invalid or there exists
  some event $\event{e}{''_{'}}\in\trPrime$ such that
  \TRPrimeOrdered{\event{e}{_1}}{\TRPrimeOrdered{\event{e}{''_{'}}}{\event{e}{_2}}} or \TRPrimeOrdered{\event{e}{_2}}{\TRPrimeOrdered{\event{e}{''_{'}}}{\event{e}{_1}}}, due to the inductive hypothesis since $\dist{\event{e}}{\event{e}{''}}<k$.
\end{proof}

\begin{theorem}[\textsc{AddConstraints} completeness]
  Given a \WBR-race $(\event{e}{_1},\event{e}{_2})$ and the graph \Gv after
  \textsc{AddConstraints} has processed it; if $(\event{e}{_1},\event{e}{_2})$ is a predictable race then \Gv does not contain a cycle before \event{e}{_1} or \event{e}{_2} such that for some events \event{e} and \event{e}{'}, \Gpath{\event{e}}{\event{e}{'}}{\Gv}, \Gpath{\event{e}{'}}{\event{e}}{\Gv}, \Gpath{\event{e}}{\event{e}{_1}}{\Gv} or \Gpath{\event{e}}{\event{e}{_2}}{\Gv}.\label{thm:addConstraints-complete}
\end{theorem}

\begin{proof}[Proof of Theorem~\ref{thm:addConstraints-complete}]
  Let us prove the theorem by contradiction. Assume that
  $(\event{e}{_1},\event{e}{_2})$ is a predictable race but \Gv does contain a
  cycle.

  Then, there must be a pair of events \event{e} and \event{e}{'} such that
  \Gpath{\event{e}}{\event{e}{'}}{\Gv}, \Gpath{\event{e}{'}}{\event{e}}{\Gv},
  and \Gpath{\event{e}}{\event{e}{_1}}{\Gv} or
  \Gpath{\event{e}}{\event{e}{_2}}{\Gv}. Let \trPrime be a valid reordering of
  \tr where the race is observed. Since the race is observed, there can not
  exist some event $\event{e}{''}\in\trPrime$ such that
  \TRPrimeOrdered{\event{e}{_1}}{\TRPrimeOrdered{\event{e}{''}}{\event{e}{_2}}}
  or
  \TRPrimeOrdered{\event{e}{_2}}{\TRPrimeOrdered{\event{e}{''}}{\event{e}{_1}}},
  and $\event{e}{_1}\in\trPrime$ and $\event{e}{_2}\in\trPrime$ must be true.
  Lemma~\ref{lem:graph-edges-necessary} shows that $\event{e}\in\trPrime$ since
  $\event{e}{_1}\in\trPrime$ and $\event{e}{_2}\in\trPrime$. Again using
  Lemma~\ref{lem:graph-edges-necessary}, \Gpath{\event{e}}{\event{e}{'}}{\Gv}
  gives us \TRPrimeOrdered{\event{e}}{\event{e}{'}} and
  \Gpath{\event{e}{'}}{\event{e}}{\Gv} gives us
  \TRPrimeOrdered{\event{e}{'}}{\event{e}}, which is a contradiction.
\end{proof}

While the constraints in the graph returned by \textsc{AddConstraints} are necessary, they are \emph{not} however sufficient:
an acyclic graph does not ensure that a predictable race exists,
and the constraints are not sufficient to ensure a valid reordering even if a predictable race does exist.
\notes{
We have never seen this happen in practice.
\mike{I think we can say this after talking about \textsc{ConstructReorderedTrace}.}
}%
However, \checkWBRRace is sound overall because it ensures that a predictable race exists,
by attempting to construct a valid reordering.
}%

\subsection{\textsc{ConstructReorderedTrace}}
\label{subsec:ConstructReorderedTrace}

\later{
\kaan{Reviewer B asks ``do you have a (contrived) example where one obtains a "don't know" result?''. The example from \DC can be adapted for here as well.
\mike{If there's a ``don't know'' result unique to \WBR, we can try to show that.
Otherwise, referring to \DC's example(s)~\cite{vindicator} is probably fine.}}
}

If \textsc{AddConstraints} returns an acyclic constraint graph \Gv,
a predictable race does not necessarily exist; and if a predictable race does exist,
what is a valid reordering that exposes the race?
\CheckWBRRace addresses both issues by calling \textsc{ConstructReorderedTrace}
(at line~\ref{line:approach:backReorder} in Algorithm~\ref{alg:verify-race}),
which attempts to construct a valid reordering
that satisfies \Gv's constraints.

\paragraph{Construction algorithm}

\textsc{ConstructReorderedTrace} calls \textsc{AttemptToConstructTrace},
which does the work of attempting to construct a valid reordering \trPrime satisfying \Gv's constraints.

\textsc{AttemptToConstructTrace} first computes the set of (backward) \emph{reachable} events $R$
that reach \event{e}{_1} or \event{e}{_2} (line~\ref{line:fullConstruct:compute-r}).
These events, along with \event{e}{_1} and \event{e}{_2}, must be in a valid reordering.
The algorithm then computes the set of \CORE reads (line~\ref{line:backReorder:find-core-reads}),
using the \GetCOREReads procedure described in Section~\ref{subsec:lw-constraints}.
The \CORE reads are events that must be in \trPrime and that must each have the same last writer as in \tr.

\textsc{AttemptToConstructTrace} constructs \trPrime
in \emph{reverse} order, starting with \event{e}{_2} and \event{e}{_1},
prepending events to \trPrime until all events in $R$ have been added
(lines~\ref{line:backReorder:initialize-trPrime}--\ref{line:backReorder:append}).
A prepended event must satisfy the constraints in \Gv (line~\ref{line:backReorder:semiLegalEvent}), 
and must not violate the \LW or \LS rules of valid reordering (line~\ref{line:backReorder:legalEvent}).

The algorithm omits the detailed logic for checking \LW and \LS rules;
we describe the logic here briefly.
An event in a critical section on \code{m} cannot be
prepended if \code{m} is currently held by a different thread,
and a critical section on \code{m} must be prepended in its entirety if \trPrime
already contains events from another critical section on \code{m}.
For each \CORE read, the last write must appear in the reordered trace before the read,
and no write to the same variable can interleave the latest write and read. 
If a \CORE read has no last writer in \tr,
then the read must have no last writer \trPrime.

\textsc{AttemptToConstructTrace} is a greedy algorithm:
it chooses the latest event in \TROrdered{}{} order among the legal events
(line~\ref{line:backReorder:pick-latest-event}).
The intuition for choosing the latest event
(shared by \checkDCRace~\cite{vindicator})
is that the original order of critical sections and memory accesses
is most likely to avoid failure to produce a valid reordering, especially for events executed \emph{before} \event{e}{_1}
in the original execution. 
\later{
\jake{We have not tested if this insight is valid for \WBR. I would expect it to be, but we don't know.}
}%

\paragraph{Retrying construction}

As mentioned above in the context of enforcing the \LS rule,
if \trPrime already contains an \Acquire{m} event,
then the algorithm cannot prepend an event $\event{e} \in \CS{\event{r}}$,
where \event{r} is a \Release{m} event,
before first prepending \event{r} to \trPrime.
However, \event{r} may \emph{not} be in $R$.
If \textsc{AttemptToConstructTrace} encounters this case (line~\ref{line:backReorder:check-missing-release}),
it returns the missing event \event{r} (line~\ref{line:backReorder:missing-release}).
\textsc{ConstructReorderedTrace} then ensures that $r$ will be part of $R$ (line~\ref{line:fullConstruct:recompute-r})
and again calls \textsc{AttemptToConstructTrace}, which recomputes $R$ and the set of \CORE reads.
In the worst case, $R$ might be missing release events for each critical section
that contains a thread's last event in $R$,
bounding the number of times that \textsc{ConstructReorderedTrace} can retry \textsc{AttemptToConstructTrace}.

\textsc{AttemptToConstructTrace} eventually returns either a valid reordered trace \trPrime
that demonstrates a predictable race (line~\ref{line:backReorder:return-trPrime}),
or it fails if no release events are missing and no legal events can be added to \trPrime,
in which case it returns an empty trace (line~\ref{line:backReorder:stuck}).

\paragraph{Discussion}

\textsc{ConstructReorderedTrace} is sound:
if it returns a non-empty trace \trPrime, then
\trPrime is a valid reordering in which
\event{e}{_1} and \event{e}{_2} are consecutive.

\textsc{ConstructReorderedTrace} is however incomplete.
Because of its greedy algorithm for choosing among legal events to prepend,
\textsc{AttemptToConstructTrace} may fail to construct a valid reordering even when one exists.
(If \textsc{AttemptToConstructTrace} backtracked, it would be complete but incur exponential worst-case complexity.)
Since \textsc{ConstructReorderedTrace} is incomplete,
if \textsc{ConstructReorderedTrace} fails to construct a valid reordering,
\checkWBRRace reports ``Don't know'' (line~\ref{line:approach:questionableRace}).
However, \textsc{ConstructReorderedTrace} never reported ``Don't know'' in our experiments (Section~\ref{sec:evaluation}).




If \tr has a \emph{predictable deadlock} (a deadlock in some valid reordering)
but no predictable race,
then no valid reordering exposing a predictable race exists,
so \checkWBRRace will report either ``No race predictable'' or ``Don't know.''
It might be possible to extend \checkWBRRace to check for predictable deadlocks.


\subsection{Asymptotic Complexity}

\later{
\kaan{Reviewer B says ``this argument is too brief.''
\mike{Meh. Maybe we should be downplaying stating asymptotic complexity, and instead mention somewhere
(not in its own subsection) that it's worst-case polynomial like \checkDCRace,
but both algorithms rely on doing much better than the worst case in practice.}}
}

As for \checkDCRace~\cite{vindicator},
\checkWBRRace's worst-case time complexity is O($N^k$), \ie, polynomial in $N$,
the number of events in $G$.
Briefly, every loop iteration count is bounded by the number of $G$'s nodes and edges.
The exact degree $k$ of the polynomial depends on various implementation details.
\CheckWBRRace uses $\Omega(N)$ space for \Gv and \trPrime.


\section{Results with Confidence Intervals}
\label{sec:extended-results}

This section presents Tables~\ref{tab:extended-races-counts-hbwcpbr}---%
\ref{tab:extended-performance-vindication}, which show the same results as Section~\ref{sec:evaluation},
but include 95\% confidence intervals for every result, delimited with a $\pm$ sign.

\begin{table*}
\centering
\small
\newcommand{\none}{\mc{1}{r@{$\;\;\textcolor{white}{\rightarrow}\;\;$}}{0}}
\newcommand{\K}[2]{#2$\;\!$K}
\newcommand{\MM}[2]{#2$\;\!$M}

\begin{tabular}{@{}l|r@{$\;\;\pm\;\;$}lr@{$\;\;\pm\;\;$}l|r@{$\;\;\pm\;\;$}lr@{$\;\;\pm\;\;$}l|r@{$\;\;\pm\;\;$}lr@{$\;\;\pm\;\;$}l@{}}
    Program & \mc{4}{c|}{\HB-races} & \mc{4}{c|}{\WCP-races} & \mc{4}{c}{\BR-races} \\
    \hline
    \bench{avrora}& 5 & 0 & (\K{202034}{202} & \K{1079}{1}) & 5 & 0 & (\K{202960}{203} & \K{1140}{1}) & 5 & 0 & (\K{202963}{203} & \K{1139}{1})  \\
    \bench{batik}& 0 & 0 & (0 & 0) & 0 & 0 & (0 & 0) & 0 & 0 & (0 & 0)  \\
    \bench{h2}& 12 & 1 & (\K{52704}{53} & 344) & 12 & 1 & (\K{52788}{53} & 340) & 12 & 1 & (\K{53538}{54} & 367)  \\
    \bench{luindex}& 1 & 0 & (1 & 0) & 1 & 0 & (1 & 0) & 1 & 0 & (1 & 0)  \\
    \bench{lusearch}& 0 & 0 & (0 & 0) & 0 & 0 & (0 & 0) & 0 & 0 & (0 & 0)  \\
    \bench{pmd}& 8 & 1 & (436 & 286) & 8 & 1 & (443 & 283) & 10 & 1 & (651 & 283)  \\
    \bench{sunflow}& 2 & 0 & (20 & 1) & 2 & 0 & (26 & 17) & 2 & 0 & (26 & 17)  \\
    \bench{tomcat}& 98 & 2 & (\K{35935}{36} & 159) & 99 & 2 & (\K{35956}{36} & 164) & 103 & 2 & (\K{38752}{39} & 163)  \\
    \bench{xalan}& 7 & 1 & (208 & 15) & 15 & 1 & (\K{531236}{531} & 920) & 37 & 1 & (\MM{2072070}{2} & \K{3648}{4})  \\
\end{tabular}
\caption{Static and dynamic (in parentheses) race counts from \HB, \WCP, and \BR analyses running together.
A ``$\pm\,$0'' confidence interval is the result of rounding to the nearest integer and
implies that the interval is smaller than $\pm\,$0.5.\label{tab:extended-races-counts-hbwcpbr}}
\end{table*}

\begin{table*}
    \centering
    \small
    \newcommand{\none}{\mc{1}{r@{$\;\;\textcolor{white}{\rightarrow}\;\;$}}{0}}
    \newcommand{\K}[2]{#2$\;\!$K}
    \newcommand{\MM}[2]{#2$\;\!$M}

    \begin{tabular}{@{}l|r@{$\;\;\pm\;\;$}lr@{$\;\;\pm\;\;$}l|r@{$\;\;\pm\;\;$}lr@{$\;\;\pm\;\;$}l@{}}
        Program & \mc{4}{c|}{\WDC-races} & \mc{4}{c}{\WBR-races} \\
        \hline
        \bench{avrora}& 5 & 0 & (\K{203497}{203} & 673) & 5 & 0 & (\K{406047}{406} & 806)  \\
        \bench{batik}& 0 & 0 & (0 & 0) & 0 & 0 & (0 & 0)  \\
        \bench{h2}& 11 & 13 & (\K{53768}{54} & \K{1360}{1}) & 12 & 13 & (\K{62526}{63} & 781)  \\
        \bench{luindex}& 1 & 0 & (1 & 0) & 1 & 0 & (1 & 0)  \\
        \bench{lusearch}& 0 & 0 & (0 & 0) & 1 & 0 & (30 & 0)  \\
        \bench{pmd}& 9 & 1 & (\K{1678}{2} & 191) & 10 & 1 & (\K{3183}{3} & 332)  \\
        \bench{sunflow}& 2 & 0 & (49 & 3) & 2 & 0 & (100 & 3)  \\
        \bench{tomcat}& 94 & 4 & (\K{35875}{36} & 593) & 284 & 72 & (\K{125078}{125} & \K{5383}{5})  \\
        \bench{xalan}& 17 & 1 & (\K{648507}{649} & 940) & 170 & 1 & (\MM{15324942}{15} & \K{37438}{37})  \\
    \end{tabular}
    \caption{Static and dynamic (in parentheses) race counts from \DC and \WBR analyses running together.
    A ``$\pm\,$0'' confidence interval is the result of rounding to the nearest integer and
implies that the interval is smaller than $\pm\,$0.5.\label{tab:extended-races-counts-dcwbr}}
    \end{table*}

\begin{table*}
    \centering
    \small
    \newcommand{\none}{\mc{1}{r@{$\;\;\textcolor{white}{\pm}\;\;$}}{0}}
    \newcommand{\K}[2]{#2$\;\!$K}
    \newcommand{\MM}[2]{#2$\;\!$M}

    \begin{tabular}{@{}l|r@{$\;\;\pm\;\;$}lr@{$\;\;\pm\;\;$}l|r@{$\;\;\pm\;\;$}lr@{$\;\;\pm\;\;$}l|r@{$\;\;\pm\;\;$}l@{$\;\;\rightarrow\;\;$}l@{$\;\;\pm\;\;$}l@{}}
        Program & \mc{4}{c|}{\BR-races} & \mc{4}{c|}{\WBR-races} & \mc{2}{c}{\WBR-only} & \mc{2}{c}{Verified} \\
        \hline
        \bench{avrora}& 5 & 0 & (\K{201613}{202} & 431)& 5 & 0 & (\K{406567}{407} & 354)&\none \\
        \bench{batik}& 0 & 0 & (0 & 0)& 0 & 0 & (0 & 0)&\none \\
        \bench{h2}& 12 & 2 & (\K{53431}{53} & 814)& 12 & 1 & (\K{62606}{63} & \K{1032}{1})&1 & 0 &  0 & 0 \\
        \bench{luindex}& 1 & 0 & (1 & 0)& 1 & 0 & (1 & 0)&\none \\
        \bench{lusearch}& 0 & 0 & (0 & 0)& 1 & 0 & (30 & 0)&1 & 0 &  0 & 0 \\
        \bench{pmd}& 9 & 1 & (470 & 152)& 11 & 1 & (\K{2944}{3} & 806)&1 & 1 & 1 & 1 \\
        \bench{sunflow}& 2 & 0 & (38 & 21)& 2 & 0 & (100 & 3)&\none \\
        \bench{tomcat}& 99 & 3 & (\K{37642}{38} & 563)& 321 & 74 & (\K{125743}{126} & \K{9980}{10})&222 & 72 & 53 & 10 \\
        \bench{xalan}& 33 & 2 & (\MM{1898266}{2} & \K{3952}{4})& 170 & 1 & (\MM{15277795}{15} & \K{79528}{80})&137 & 1 & 135 & 1 \\
    \end{tabular}

    \caption{Static and dynamic (in parentheses) race counts. The \colname{\WBR-only $\rightarrow$ Verified} column reports static \WBR-only races,
followed by how many static \WBR-only races were verified as predictable races by \checkWBRRace.
A ``$\pm\,$0'' confidence interval is the result of rounding to the nearest integer and
implies that the interval is smaller than $\pm\,$0.5.\label{tab:extended-races-vindication}}
\end{table*}

\begin{table*}
    \newcommand{\seconds}[1]{#1~s}
    \small

    \begin{tabular}{@{}l|r@{$\;\;\pm\;\;$}l|r@{$\;\;\pm\;\;$}lr@{$\;\;\pm\;\;$}l@{}}
        Program & \mc{2}{c|}{Base} & \mc{2}{c}{w/o br} & \mc{2}{c}{w/ br} \\
        \hline
        \bench{avrora} & \seconds{7.7} & \seconds{1.0} & \slowdown{2.1} & \slowdown{0.1} & \slowdown{2.6} & \slowdown{0.1} \\
        \bench{batik} & \seconds{3.3} & \seconds{0.0} & \slowdown{3.9} & \slowdown{0.0} & \slowdown{5.5} & \slowdown{0.0} \\
        \bench{h2} & \seconds{9.4} & \seconds{0.3} & \slowdown{6.6} & \slowdown{0.4} & \slowdown{9.3} & \slowdown{0.3} \\
        \bench{luindex} & \seconds{1.4} & \seconds{0.1} & \slowdown{5.9} & \slowdown{0.0} & \slowdown{11} & \slowdown{0.2} \\
        \bench{lusearch} & \seconds{3.8} & \seconds{0.8} & \slowdown{4.4} & \slowdown{0.8} & \slowdown{5.2} & \slowdown{0.6} \\
        \bench{pmd} & \seconds{2.6} & \seconds{0.2} & \slowdown{7.6} & \slowdown{0.4} & \slowdown{10.0} & \slowdown{0.7} \\
        \bench{sunflow} & \seconds{2.3} & \seconds{0.3} & \slowdown{11} & \slowdown{0.9} & \slowdown{13} & \slowdown{0.5} \\
        \bench{tomcat} & \seconds{1.6} & \seconds{0.1} & \slowdown{5.9} & \slowdown{0.4} & \slowdown{6.2} & \slowdown{0.2} \\
        \bench{xalan} & \seconds{5.3} & \seconds{1.3} & \slowdown{2.6} & \slowdown{0.3} & \slowdown{3.3} & \slowdown{0.2} \\        
    \end{tabular}
    \caption{Slowdowns of program instrumentation over uninstrumented execution.
    \label{tab:extended-performance-base}}
\end{table*}
\begin{table*}
    \small
    \newcommand{\none}{\mc{1}{r@{$\;\;\textcolor{white}{\pm}\;\;$}}{0}}
    \begin{tabular}{@{}l|r@{$\;\;\pm\;\;$}l|r@{$\;\;\pm\;\;$}l||r@{$\;\;\pm\;\;$}l|r@{$\;\;\pm\;\;$}l@{}}
        Program & \mc{2}{c|}{\WCP} & \mc{2}{c||}{\BR} & \mc{2}{c|}{\BR{}+\WDC} & \mc{2}{c}{\BR{}+\WBR} \\
        \hline
        \bench{avrora} & \slowdown{17} & \slowdown{0.6} & \slowdown{20} & \slowdown{1.3} & \slowdown{28} & \slowdown{1.3} & \slowdown{34} & \slowdown{2.1} \\
        \bench{batik} & \slowdown{19} & \slowdown{1.2} & \slowdown{20} & \slowdown{2.0} & \slowdown{24} & \slowdown{2.4} & \slowdown{33} & \slowdown{5.4} \\
        \bench{h2} & \slowdown{129} & \slowdown{13} & \slowdown{139} & \slowdown{21} & \slowdown{202} & \slowdown{20} & \slowdown{212} & \slowdown{22} \\
        \bench{luindex} & \slowdown{91} & \slowdown{11} & \slowdown{93} & \slowdown{12} & \slowdown{118} & \slowdown{20} & \slowdown{139} & \slowdown{32} \\
        \bench{lusearch} & \slowdown{20} & \slowdown{2.9} & \slowdown{23} & \slowdown{3.6} & \slowdown{29} & \slowdown{3.7} & \slowdown{32} & \slowdown{4.1} \\
        \bench{pmd} & \slowdown{27} & \slowdown{3.0} & \slowdown{29} & \slowdown{4.7} & \slowdown{32} & \slowdown{3.2} & \slowdown{36} & \slowdown{5.8} \\
        \bench{sunflow} & \slowdown{142} & \slowdown{9.8} & \slowdown{196} & \slowdown{31} & \slowdown{210} & \slowdown{34} & \slowdown{221} & \slowdown{29} \\
        \bench{tomcat} & \slowdown{30} & \slowdown{4.3} & \slowdown{37} & \slowdown{4.4} & \slowdown{55} & \slowdown{11} & \slowdown{59} & \slowdown{7.5} \\
        \bench{xalan} & \slowdown{40} & \slowdown{7.1} & \slowdown{46} & \slowdown{6.9} & \slowdown{63} & \slowdown{7.8} & \slowdown{86} & \slowdown{9.0} \\

    \end{tabular}
    \caption{Slowdowns of various analyses over uninstrumented execution.
    \label{tab:extended-performance-analysis}}
\end{table*}
\begin{table*}
    \newcommand{\seconds}[1]{#1~s}
    \newcommand{\none}{\mc{1}{r@{$\;\;\textcolor{white}{\pm}\;\;$}}{}}
    \small
    \begin{tabular}{@{}l|r@{$\;\;\pm\;\;$}l|r@{$\;\;\pm\;\;$}l|r@{$\;\;\pm\;\;$}l@{\;\;}r@{$\;\;\pm\;\;$}l@{}}
                & \mc{2}{c|}{}           & \mc{6}{c}{\BR{}+\WBR{}+graph} \\
        Program & \mc{2}{c|}{\BR{}+\WBR} & \mc{2}{c|}{Slowdown} & \mc{2}{c}{Failed} & \mc{2}{c}{Verified} \\
        \hline
        \bench{avrora} & \slowdown{34} & \slowdown{2.1} & \slowdown{42} & \slowdown{7.0} & \none &  & \none &  \\
        \bench{batik} & \slowdown{33} & \slowdown{5.4} & \slowdown{32} & \slowdown{5.5} & \none &  & \none &  \\
        \bench{h2} & \slowdown{212} & \slowdown{22} & \slowdown{280} & \slowdown{51}& \seconds{233}& \seconds{185} & \none &  \\
        \bench{luindex} & \slowdown{139} & \slowdown{32} & \slowdown{154} & \slowdown{12} & \none &  & \none &  \\
        \bench{lusearch} & \slowdown{32} & \slowdown{4.1} & \slowdown{42} & \slowdown{3.8}& \seconds{< 0.1}& \seconds{< 0.1} & \none &  \\
        \bench{pmd} & \slowdown{36} & \slowdown{5.8} & \slowdown{36} & \slowdown{1.3} & \none & & \seconds{8.6}& \seconds{20} \\
        \bench{sunflow} & \slowdown{221} & \slowdown{29} & \slowdown{264} & \slowdown{10} & \none &  & \none &  \\
        \bench{tomcat} & \slowdown{59} & \slowdown{7.5} & \slowdown{56} & \slowdown{3.1}& \seconds{2.0}& \seconds{< 0.1}& \seconds{55}& \seconds{26} \\
        \bench{xalan} & \slowdown{86} & \slowdown{9.0} & \slowdown{124} & \slowdown{6.1}& \seconds{30}& \seconds{9.4}& \seconds{0.1}& \seconds{< 0.1} \\        
    \end{tabular}
    \caption{Slowdowns of \WBR analysis without and with graph generation, and the average time taken to vindicate \WBR-only races.
    \label{tab:extended-performance-vindication}}
\end{table*}

}{}

\later{
\section{\WBR SmartTrack optimizations}
\label{sec:wbr-smarttrack}

As \WBR doesn't form a total order on conflicting writes, only storing the clock
of the last write is not sufficient. An access may be ordered to the latest
write while racing with an earlier write, even if that is the first race in the
execution. We solve this by introducing the cases \emph{After Ordered Writes}
and \emph{After Unordered Writes}. An access is after ordered writes only if it
is ordered to all prior conflicting writes, otherwise it is after unordered
writes.

\begin{algorithm*}

\newcommand{\AcqMQ}[3]{\ensuremath{\mathit{Acq}_{#1,#2}(#3)}\xspace}
\newcommand{\RelMQ}[3]{\ensuremath{\mathit{Rel}_{#1,#2}(#3)}\xspace}
\newcommand{\LockVarQ}[3]{\ensuremath{L_{#1,#2}^{#3}}\xspace}
\newcommand{\PerThrLockQ}[2]{\ensuremath{T_{#2}^{#1}}\xspace}
\newcommand{\LockVarThrQRd}[3]{\ensuremath{R_{#1,#2,#3}}\xspace}
\newcommand{\LockVarThrQWr}[3]{\ensuremath{W_{#1,#2,#3}}\xspace}
\newcommand{\epLeq}{\ensuremath{\preceq}\xspace}
\newcommand{\epochLeqVC}[2]{\ensuremath{#1 \epLeq #2}}
\newcommand{\epochLeqVCThr}[3]{\ensuremath{#1 \epLeq #2}}
\newcommand{\case}[1]{\textsc{\small [#1]}}
\newcommand{\ReadSameEpoch}{\case{Read Same Epoch}\xspace}
\newcommand{\WriteSameEpoch}{\case{Write Same Epoch}\xspace}
\newcommand{\SharedSameEpoch}{\case{Shared Same Epoch}\xspace}
\newcommand{\ReadOwned}{\case{Read Owned}\xspace}
\newcommand{\ReadExclusive}{\case{Read Exclusive}\xspace}
\newcommand{\ReadShare}{\case{Read Share}\xspace}
\newcommand{\ReadSharedOwned}{\case{Read Shared Owned}\xspace}
\newcommand{\ReadShared}{\case{Read Shared}\xspace}
\newcommand{\WriteOwned}{\case{Write Owned}\xspace}
\newcommand{\WriteExclusive}{\case{Write Exclusive}\xspace}
\newcommand{\WriteShared}{\case{Write Shared}\xspace}

\newcommand{\epoch}[2]{\ensuremath{#1\:\!@\:\!#2}\xspace}
\newcommand{\initE}{\ensuremath{\bot}\xspace}

\caption{\hfill SmartTrack-\WBR (SmartTrack-based \WBR analysis)}

\newcommand\Ht{\ensuremath{H_t}\xspace}
\newcommand\Rc{\ensuremath{C}\xspace}
\newcommand\dereference[1]{\ensuremath{#1}\xspace}

\small

\vspace*{-1.5em}
\begin{algorithmic}[1]

	\Procedure{Acquire}{$t$, $m$}
		\lForEach{$t' \neq t$}{$\AcqMQ{m}{t'}{t}$.Enque($C_t(t)$)} \label{line:RE-VC:AcqMQEnque}
		\State \textbf{let} $\Rc =$ reference to new vector clock \label{line:RE-VC:newCm}
		\State $\dereference{\Rc}(t) \gets \infty$ \label{line:RE-VC:inftyCm}
		\State $\Ht \gets \langle \Rc, m \rangle \oplus \Ht$ \Comment{Prepend $\langle \Rc, m \rangle$ to head of list}\label{line:RE-VC:addHt}
		\State $C_t(t) \gets C_t(t) + 1$ \label{line:RE-VC:CtincAcq}
	\EndProcedure

	\Procedure{Release}{$t$, $m$}
		\ForEach{$t' \neq t$} \label{line:RE-VC:AcqMQFrontLoop}
			\While{\epochLeqVCThr{\AcqMQ{m}{t}{t'}\textnormal{.Front()}}{C_t}{t'}}
				\State $\AcqMQ{m}{t}{t'}$.Deque()
				\State $C_t \gets C_t \sqcup \RelMQ{m}{t}{t'}$.Deque() \label{line:RE-VC:RelMQDeque}
			\EndWhile
		\EndFor
		\lForEach{$t' \neq t$}{$\RelMQ{m}{t'}{t}$.Enque($C_t$)} \label{line:RE-VC:RelMQEnque}
		\State \textbf{let} $\langle \Rc, \_ \rangle = \textnormal{head}(\Ht)$ \Comment{head() returns first element} \label{line:RE-VC:headHt}
		\State $\dereference{\Rc} \gets C_t$ \Comment{Update vector clock referenced by $C$} \label{line:RE-VC:setCm}
		\State $\Ht \gets$ rest(\Ht) \Comment{rest() returns list without first element} \label{line:RE-VC:removeHt}
		\State $C_t(t) \gets C_t(t) + 1$ \label{line:RE-VC:CtincRel}
	\EndProcedure

	\Procedure{Write}{$t$, $x$}
		\lIf {$W_x = C_t(t)@t$} \textbf{return} \CaseComment{\WriteSameEpoch} \label{line:RE-VC:wrSameEpoch}
		\lIf {$W_x(t) = C_t(t)$} \textbf{return}

		\If {$R_x = \epoch{c}{u}$} \CaseComment{\textsc{[Write After Ordered Read]}}
			\State \textbf{check} $u=t \lor c \leq C_t(t) \lor L^r_x \cap L \neq \emptyset$
		\Else \CaseComment{\textsc{[Write After Unordered Reads]}}
			\ForEach{$u \neq t$}
				\State \textbf{check} $ R_x(u) \leq C_t(u) \lor L^r_x(u) \cap L \neq \emptyset$
			\EndFor
		\EndIf

		\If {$W_x = \epoch{c}{u}$} \CaseComment{\textsc{[Write After Ordered Writes]}}
			\If {$u = t \lor c \leq C_t(u)$}
				\State $W_x \gets C_t(t)@t$ \Comment{Current write is ordered too}
				\State $L^w_x \gets H_t$
			\Else
				\State \textbf{check} $L^w_x \cap L \neq \emptyset$
				\State $W_x \gets C_t$ \Comment{Current write unordered, use VC}
				\State $L^w_x \gets \{ L^w_x , \Ht \}$
				\State $T_x \gets t$
			\EndIf
		\Else \CaseComment{\textsc{[Write After Unordered Writes]}}
			\State $o \gets \textnormal{true}$
			\ForEach{$u \neq t$}
				\lIf {$R_x(u) \leq C_t(u)$} \textbf{continue}
				\State \textbf{check} $L^r_x \cap L \neq \emptyset$
				\State $o \gets \textnormal{false}$
			\EndFor
			\If {$o = \textnormal{true}$}
				\State $W_x \gets C_t(t)@t$ \Comment{Write ordered, use epoch}
				\State $L^w_x \gets \Ht$
			\Else
				\State $W_x(t) \gets C_t(t)$ \Comment{Write unordered, keep VC}
				\State $L^w_x(t) \gets \Ht$
				\State $T_x \gets t$
			\EndIf
		\EndIf
	\EndProcedure


	\Procedure{Read}{$t$, $x$, $L$}
		\lIf {$R_x = C_t(t)@t$}{\textbf{return}} \CaseComment{\ReadSameEpoch}
		\lIf {$R_x(t) = C_t(t)$}{\textbf{return}} \CaseComment{\SharedSameEpoch}

		\If {$W_x = \epoch{c}{u}$} \CaseComment{\textsc{[Read After Ordered Writes]}}
			\If {$u \neq t$}
				\State \textbf{check} $c \leq C_t(u) \lor L^w_x \cap L \neq \emptyset$
				\If {$c > C_t(u) \land L^w_x \cap L \neq \emptyset$}
					\State $\textsc{CSCheck}(L^w_x)$
				\EndIf
			\EndIf
		\Else \CaseComment{\textsc{[Read After Unordered Writes]}}
			\ForEach{$u \neq t$}
				\State \textbf{check} $W_x(u) \leq C_t(u) \lor L^w_x(u) \cap L \neq \emptyset$
				\If {$T_x = u \land W_x(u) > C_t(u) \land L^w_x(u) \cap L \neq \emptyset$}
					\State $\textsc{CSCheck}(L^w_x(u))$
				\EndIf
			\EndFor
		\EndIf

		\If {$R_x = \epoch{c}{u}$} \CaseComment{\textsc{[Read After Ordered Read]}}
			\If {$u = t \lor c \leq C_t(u)$}
				\State $R_x \gets C_t(t)@t$
				\State $L^r_x \gets H_t$
			\Else
				\State $R_x \gets C_t$ \Comment{Unordered read, use VC}
				\State $L^r_x \gets \{ L^r_x , \Ht \}$
			\EndIf
		\Else \CaseComment{\ReadShared}
			\If {$R_x \lessvc C_t$}
				\State $R_x \gets C_t(t)@t$
				\State $L^r_x \gets \Ht$
			\Else
				\State $R_x(t) \gets C_t(t)$
				\State $L^r_x(t) \gets \Ht$
			\EndIf
		\EndIf
	\EndProcedure

\kaan{Branches are not using source locations, which makes it faster since we only need a single vector clock for pending rule-a edges then.}

	\Procedure{Branch}{$t$}
		\lIf {$B_t = -1@-1$} \textbf{return}
		\State $C_t \gets C_t \sqcup B_t$
		\State $B_t \gets -1@-1$
	\EndProcedure

	\Procedure{CSCheck}{$L$}
		\ForEach {$\langle \Rc, m \rangle \textnormal{ in } L$ in tail-to-head order}
			\If {$m \in \textnormal{heldby}(t)$}
				\State $B_t \gets B_t \sqcup \dereference{\Rc}$
				\State \textbf{return}
			\EndIf
		\EndFor
	\EndProcedure

\end{algorithmic}
\vspace*{-1em}
\label{alg:RE-VC}
\end{algorithm*}

}

\end{document}